\documentclass[a4paper,twocolumn,11pt,accepted=2023-03-10]{quantumarticle}
\pdfoutput=1

\usepackage{natbib}
\setcitestyle{square,comma,numbers,sort&compress}

\usepackage{amsmath,amssymb,amsthm,lipsum}

\usepackage{float}
\makeatletter
\let\newfloat\newfloat@ltx
\makeatother

\usepackage{tikz}
\usetikzlibrary{trees,positioning}
\usepackage{forest}
\usepackage{enumitem}
\usepackage{makecell}
\usepackage{xspace}
\usepackage{nicefrac}
\usepackage{graphicx}% Include figure files
\usepackage{dcolumn}% Align table columns on decimal point
\usepackage{bm}% bold math
\usepackage{hyperref}% add hypertext capabilities
\usepackage[mathlines]{lineno}% Enable numbering of text and display math
\usepackage{complexity}

\usepackage{comment}
\usepackage{algorithm}
\usepackage[noend]{algpseudocode}
\usepackage{color}
\input{Qcircuit}
\usepackage{mathtools}
\usepackage{thmtools}
\usepackage{cleveref}
\usepackage{fancybox}
\usepackage{tikz-qtree}
\usepackage{todonotes}
\floatname{algorithm}{Algorithm~}

\def\EQ#1{\begin{eqnarray}#1\end{eqnarray}}
\newcommand{\cost}[1]{\text{Space}(#1)}
\newcommand{\dm}[1]{|#1\rangle \langle #1 |}

\newcommand{\SIAB}{\texttt{SIAB}}

\newcommand{\SIAR}{\texttt{SIAR}}

\newtheorem{theorem}{Theorem}[section]
\newtheorem{proposition}[theorem]{Proposition}
\newtheorem{corollary}[theorem]{Corollary}
\newtheorem{lemma}[theorem]{Lemma}

\newtheorem{definition}[theorem]{Definition}

\newtheorem{remark}[theorem]{Remark}

\usepackage{listings}% http://ctan.org/pkg/listings
\usepackage{tikz}
\usetikzlibrary{trees}
\usepackage{forest}

\newcommand{\iw}{\mathrm{iw}}

\usepackage{xspace}

\newcommand\concept[1]{\textit{#1}}
\newcommand\defmath[2]{\newcommand#1{\ensuremath{#2}\xspace}}

\algnewcommand{\IfThenElse}[3]{%\IfThenElse{<if>}{<then>}{<else>}
  \State \algorithmicif\ #1\ \algorithmicthen\ #2\ \algorithmicelse\ #3}
\algnewcommand{\IfThen}[2]{%\IfThen{<if>}{<then>}
  \State \algorithmicif\ #1\ \algorithmicthen\ #2}
  
\defmath\xor{~\oplus\hspace{-1.4mm}=~}

\providecommand{\set}[1]{\ensuremath{\left\lbrace #1 \right\rbrace}}
\providecommand{\sizeof}[1]{\ensuremath{\left\vert{#1}\right\vert}}

\providecommand{\vect}[1]{\ensuremath{\left\langle #1 \right\rangle}}
\providecommand{\floor}[1]{\lfloor#1\rfloor}
\providecommand{\ceil}[1]{\lceil#1\rceil}

\defmath\tildeo{\tilde{O}}
\defmath\staro{{O}^\ast}

\newcommand{\defn}{\,\triangleq\,}
\defmath\lolo{\mathrm{ll}}

\defmath\as{|\vec{a}_{s}|}
\defmath\ssiar{S_{\mathit{r}}}
\defmath\tsiar{T_{\mathit{r}}}
\defmath\ssiab{S_{\mathit{b}}}
\defmath\tsiab{T_{\mathit{b}}}

\defmath\ssias{S_{\mathit{s}}}
\defmath\tsias{T_{\mathit{s}}}

\defmath\ks{{k_{s}}}
\defmath\kr{{k_{r}}}

\defmath\jcut{J_{c}}
\defmath\jleaves{J_{\ell}}

\newcommand{\Avg}{\text{Avg}}

\begin{document}

\title{Hybrid divide-and-conquer approach for tree search algorithms~:~possibilities and limitations}

\author{Mathys Rennela}
\email{mathys.rennela@inria.fr}
\affiliation{Laboratoire de Physique de l’Ecole Normale Supérieure, Inria, CNRS, ENS-PSL, Mines-Paristech, Sorbonne Université, PSL Research University, Paris, France}%
 
\author{Sebastiaan Brand}
\email{s.o.brand@liacs.leidenuniv.nl}
\affiliation{LIACS, Leiden University, Leiden, The Netherlands}%

\author{Alfons Laarman}
\email{a.w.laarman@liacs.leidenuniv.nl}
\affiliation{LIACS, Leiden University, Leiden, The Netherlands}%
  
\author{Vedran Dunjko}
\email{v.dunjko@liacs.leidenuniv.nl}
\affiliation{LIACS, Leiden University, Leiden, The Netherlands}%


%Quantum: here we make the title
\maketitle

\begin{abstract}
One of the challenges of quantum computers in the near- and mid- term is the limited number of qubits we can use for computations. Finding methods that achieve useful quantum improvements under size limitations is thus a key question in the field.
In this vein, it was recently shown that a hybrid classical-quantum method can help provide polynomial speed-ups to classical divide-and-conquer algorithms, even when only given access to a quantum computer much smaller than the problem itself.
In this work, we study the hybrid divide-and-conquer method in the context of tree search algorithms, and extend it by including quantum backtracking, which allows better results than previous Grover-based methods.
Further, we provide general criteria for threshold-free polynomial speed-ups in the tree search context, and provide a number of examples where polynomial speed ups, using arbitrarily smaller quantum computers, can be obtained.
We provide conditions for speedups for the well known algorithm of DPLL, and we prove threshold-free speed-ups for
the PPSZ algorithm (the core of the fastest exact Boolean satisfiability solver) for well-behaved classes of formulas.
%We study possible speed-ups for the well known algorithm of  DPLL and prove threshold-free speed-ups for the so-called PPSZ algorithm -- which is the core of the fastest exact Boolean satisfiability solver -- for certain classes of formulas. 
We also provide a simple example where speed-ups can be obtained in an algorithm-independent fashion, under certain well-studied complexity-theoretical assumptions. Finally, we briefly discuss the fundamental limitations of hybrid methods in providing speed-ups for larger problems.
%\begin{description}
%\end{description}
\end{abstract}

%APS: here we make the title
%\maketitle

\section{Introduction}
%Quantum computers surpass the computing abilities of classical computers in the near future, but current limitations, such as, the few number of qubits available and noisy architectures, 
%%suggest that the cost of practical implementations will be 
%hamper their applicability. 
Years of progress in experimental quantum physics have now brought us to the verge of real-world quantum computers. These devices will, however, for the near term remain quite limited in a number of ways, including fidelities, architectures, decoherence times, and, total qubit numbers.
Each of the constraints places challenges on the quantum algorithm designer. Specifically the limitation on qubit count -- which is the focus of this work -- motivates the search for  space-efficient quantum algorithms, and the development of new methods which allow us to beneficially apply smaller devices.
%in order to run them on smaller quantum computers.

Recent works introduced an approach to extend the applicability of smaller devices by proposing a hybrid divide-and-conquer scheme~\cite{dnq-schoning,dnq-eppstein}. This method exploits the pre-specified sub-division of problems in such algorithms, and delegates the work to the quantum machine when the instances become small enough. This regular structure also allowed for analytic expressions for the asymptotic run-times of hybrid algorithms.
%The goal of the hybrid algorithm is to provide a polynomial speedup with respect to the original classical algorithm

The hybrid divide-and-conquer method was applied to two cases of divide and conquer algorithms, that of derandomized Sch\"{o}ning's algorithm for solving Boolean satisfiability~\cite{dnq-schoning}, and to the problem of finding Hamilton cycles on cubic graphs~\cite{dnq-eppstein}. These schemes achieved 
asymptotic polynomial speedups given a quantum computer of size $m$, where $m$ is a fraction of the instance size $n$, i.e., $m =\kappa n$. Interestingly, these speed-ups are obtainable for all fractions~$\kappa,$ i.e.~the improvements are \emph{threshold free}.

The space efficiency of the quantum subroutines was identified as a key criterion for determining whether threshold-free speed-ups are possible, as one may expect.

In these works, the quantum algorithmic backbone was Grover's search, which is space-frugal, but known to be sub-optimal for the cases when the underlying search spaces, the \emph{search trees}, are not complete nor uniform. 
To obtain speedup in these cases, more involved quantum search techniques, namely quantum backtracking~\cite{qbacktracking}, need to be employed.
However, until this work, it was not clear how the space demands of quantum backtracking would influence the applicability of the hybrid approach.

Here, we resolve this issue, and investigate the generalizations of the hybrid divide-and-conquer scheme from the perspective of algorithms which reduce to tree search, in particular, backtracking algorithms.  Our approach is then applied to the two of the arguably best known exact algorithms for Boolean satisfiability: the algorithm of  Davis-Putnam-Logemann-Loveland (DPLL) algorithm~\cite{og-dll} (which is still the backbone of many heuristic real-world SAT solvers) and the Paturi-Pudl\'ak-Saks-Zane (PPSZ) algorithm~\cite{og-ppsz} (which is the backbone of the best-known exact SAT solver). 

The main contributions of this work are summarized as follows:

\begin{itemize}
 \item We analyze the hybrid divide-and-conquer scheme from the perspective of search trees and provide very general criteria which can ensure polynomial-time speed-ups over classical algorithms (Section~\ref{sec:hybrid}). We also consider the limitations of the scheme in the context of online classical algorithms, which terminate as soon as a result is found.
 
 \item We demonstrate that quantum backtracking methods can be employed in a hybrid scheme. This implies an improvement over the previous hybrid algorithm for Hamilton cycles on degree 3  graphs. While the performance of Grover's search and the quantum backtracking algorithm have been compared in the context of the $k$-SAT problem, with hardware limitations in mind~\cite{campbell-khurana-montanaro}, this is the first time that the hybrid method and quantum backtracking is combined.

 \item We exhibit the first, and very simple example of an algorithm-independent provable hybrid speed-up with quantum backtracking, under well-studied complexity-theoretic assumptions.
 
 \item We study a number of settings from the search tree structure and provide all algorithmic elements required for space-efficient quantum hybrid enhancements of DPLL and PPSZ; specially for the case of PPSZ tree search, we demonstrate, under some assumptions on the variable order, a threshold-free hybrid speed-up for a class of formulas, including settings where our methods likely beat not just PPSZ tree search but any classical algorithm.
 
 \item We discuss of the fundamental limitations of our and related hybrid methods.
\end{itemize}

To achieve the above results, we provide space and time-efficient %implementations of 
{quantum versions} of various subroutines specific to these algorithms, but also of routines which may be of independent interest. This includes a simple yet exponentially more space efficient implementation of the phase estimation step in quantum backtracking, over the direct implementation \cite{practical-qbacktracking}.

%We point out that while our results do not imply guaranteed threshold-free speed-ups for all DPLL and PPSZ runs, they do provide the first steps in that direction by provable speed-ups for fully characterized classes of inputs.

The structure of the paper is as follows.
The background material is discussed in section~\ref{sec:background}. This section lays the groundwork for hybrid algorithms from a search tree perspective, by elaborating on how backtracking defines those trees. 
Section~\ref{sec:hybrid} then introduces a tree decomposition which defines our hybrid strategy, i.e.~from what points in the search tree will we start using the quantum computer, and analyzes their impact.
In Section~\ref{sec:criteria}, we discuss sufficient criteria for attaining speedups with both Grover-based search and quantum backtracking over the original classical algorithms, including online algorithms. 
Section~\ref{sec:concrete} provides concrete examples of algorithms and problem classes where threshold-free, and algorithm independent speed-ups can be obtained. Finally, in Section~\ref{sec:discussion}, we discuss the potential and limitations of hybrid approaches for the DPLL algorithm, also in the more practical case when all run-times are restricted to polynomial. This section also briefly addresses the question of the limits of possible speed-ups in any hybrid setting. 
The appendix collects all our more technical results, and some of the background frameworks.

\section{Background}
\label{sec:background}

This section introduces %backtracking in a way that facilitates both the design and analysis of our hybrid algorithms.
SAT and a general backtracking framework, which is instantiated for two exact algorithms: DPLL and PPSZ.
While DPLL performs better in practice on many instances, PPSZ-based algorithms provide the best worst-case runtime guarantee at the time of writing~\cite{biased-ppsz}. 
The section ends with an explanation of the hybrid divide-and-conquer method.

%\subsection{Backtracking, search trees, and SAT}
%\label{sub:backtracking-search-trees-SAT}
%
%
%Backtracking is an algorithmic search method finding a solution in a (combinatorial) search space, e.g., a string satisfying some set of constraints within the space of all length $n$ bit strings.
%The essential benefit of the method is its ability to greatly prune search spaces for many practical instances~\cite{biere2009handbook}.
%Critical for backtracking is that partial candidate-solutions, e.g., strings where only some of the $n$-bits are set and others are unspecified, can be used to determine that no valid solution can be found in the subspace that the partial solution represents, e.g., among all strings extending the partial solution to a full solution.
%
%The space of all partial solutions can generally be organized in a tree structure.
%And the basic idea in backtracking is to incrementally build solutions (going deeper in the tree), and giving up on the search in a given branch as soon as a partial candidate violates the desired constraints; at this instance, the algorithm ``backtracks'' and explores other branches, pruning the search space. 
%Backtracking is, for our purposes, best exemplified on the problem of Boolean satisfiability (SAT).
%Before we discuss backtracking for SAT, we introduce SAT in more detail.

\subsection{Satisfiability}
\label{subsub:satisfiability}

A Boolean formula $F$ over $n$ variables corresponds to a function 
$F:\{0,1\}^{n} \rightarrow \{0,1\}$ in the natural way.
The Boolean satisfiability problem (SAT) is the constraint satisfaction problem of determining whether a given Boolean formula $F$ in conjunctive normal form (CNF) has a satisfying assignment, i.e.~a bit string $y\in\set{0,1}^n$ such that $F(y)=1$. 
A CNF formula is a conjunction (logical `and') of disjunctions (logical `or') over variables or their negations (jointly called literals).
For ease of manipulation, we view CNF formulas  as a set of clauses (a conjunction), where each clause is a set of $k$ literals (a disjunction),
with positive or negative polarity (i.e.~a Boolean atom $x$, or its negation $\bar{x}$).
In the $k$-SAT problem, the formula $F$ is a $k$-CNF formula, meaning that all the clauses have $k$ literals.\footnote{Without loss of generality, we will allow that the formula has clauses with fewer literals, but assume that at least one clause has $k$ literals, and no clauses have more.}

It is well known that solving SAT for $3$-CNF formulas ($3$-SAT) is an NP-complete problem. As a canonical problem, it is highly relevant both in computer science~\cite{biere2009handbook} and outside, e.g.~finding ground energies of classical systems can often reduced to SAT (see~e.g.~\cite{ipc-book}).

Since it is a disjunction, a clause $C \in F$ evaluates to true (1) on a (partial) assignment $\vec{x}$, 
 if at least one of the literals in $C$ attains the value true (1) according to $\vec{x}$.
A formula $F$ evaluates to true on an assignment $\vec{x}$, 
i.e.~$F_{\vec x} = 1$, if all its clauses evaluate to true on this partial assignment.
Note that to determine the formula evaluates to true on some (partial) assignment, the assignment has to fix of at least one variable per clause of $F$, i.e.~the assignment should be linear in the number of variables. On the other hand, a formula can evaluate to false (0) on very sparse partial assignments: it suffices that the given partial assignment renders any of the clauses false
by setting $k$ variables, i.e.~a constant amount.
In such a case, we say that the partial assignment \emph{establishes a contradiction}. This is the equivalent of saying that the constraint formula is \emph{inconsistent}.

%We can equate a formula with its set of clauses.
%A formula $F$ is a subformula of $G$ if the set of clauses of $F$  is a subset of the set of clauses of $G$. Formulas can be restricted by setting some of the values.
Given a partial assignment $\vec{x}\in \{0,1, \ast\}^n$,
%\footnote{Here we imagine $\vec{x}$ to either be a set of pairs of variables, with a chosen assignment, or, a sequence of $n$ elements in the set , where $\ast$ means the corresponding variable is not assigned. The exact model to represent partial objects is not relevant for the moment. } 
we denote with $F_{| \vec{x}}$ the subfunction of $F$ restricted to that assignment.
This subfunction can also be obtained by setting the variables in the formula $F$ to the values specified in the partial assignment. For CNF formulas this means the following: for an assigned variable $x_j$, for every clause of $F$ where  $x_j$ appears as a literal (of some polarity), and the setting of $x_j$ renders the corresponding literal true, that clause is dropped in $F_{| \vec{x}}$. For every clause of $F$ where  $x_j$ appears as a literal (of some polarity), and the setting of $x_j$ renders the corresponding literal false, that literal is removed from the corresponding clause.
Finally, an assignment $\vec{x}$ that establishes a contradiction introduces an \emph{empty clause}, i.e.~$\emptyset \in F_{|\vec{x}}$, whereas a satisfying assignment $\vec{y}$ yields an \emph{empty formula}, i.e.~$F_{|\vec{y}} = \emptyset$.
We will call such formulas, where some variables have been fixed ($F_{|\vec{x}})$, \concept{restricted formulas}, and say that $F_{|\vec{x}}$ is a restriction of $F$ by the partial assignment $\vec{x}$.
%they use this term in  https://cseweb.ucsd.edu/~paturi/myPapers/pubs/PaturiPudlakSaksZane_2005_jacm.pdf
We will also use a similar notation for setting literals to true. Given a literal $l \in \{x_i, \bar{x}_i\}$ we write $F_{| l}$ for the restricted formula given by $F$ with the value of the variable $x_i$ set such as to render $l$ true.

Additionally, we say a formula $G$ is semantically entailed by a formula $F$ if all satisfying assignments (models) of $F$ also satisfy $G$, denoted $F \models G$. 
In general $G$ can be any formula over a subset of variables of $F$, although we only consider cases where $G$ consists of a single literal $l \in \{x_i, \bar{x}_i\}$.

Finally, \concept{resolution} is a logical inference rule used for satisfiability proving~\cite{davis-putnam}. Basic resolution rules takes  two clauses $(x \vee A)$ and $(\bar{x} \vee B)$, where $A,B$ are clauses,
and derives a clause $(A\vee B)$, since both clauses are satisfiable if and only the inferred clause is. The clausal inference rule is refutation complete, meaning that a complete, recursive search over all inferred clauses will find a refutation if the original formula, i.e.~a set of clauses, is unsatisfiable.
In \concept{unit resolution}, we have $A = \emptyset$ (or $B=\emptyset$) in the above, i.e., one of the input clauses is a \concept{unit clause} $(x)$ (or  $(\bar x)$). 

%Syntactically, 

\subsection{The DPLL algorithm family}
\label{sub:dpll}

The algorithm of Davis~Putnam~Logemann~Loveland (DPLL) is a backtracking algorithm which recursively explores the possible assignments for a given formula, applying (incomplete) reduction rules along the way to prune the search space~\cite{og-dll}. Originally designed to resolve the memory intensity of solvers based only on resolution (the DP algorithm uses a complete resolution system~\cite{davis-putnam}), this backtracking-based algorithm can now be found in some of the most competitive SAT solvers~\cite{heule2018proceedings}. DPLL was also generalized to support various theories, including linear integer arithmetic and uninterpreted functions~\cite{abstract-dpll}, and is using consequently in many combinatorial domains~\cite{biere2009handbook} and automated theorem proving~\cite{blanchette2013extending}.

%DPLL, shown in Algorithm~\autoref{alg:dpll}, recursively assigns 0 and 1 to all $n$ variables, while eagerly simplifying the formula (see Line~\ref{dpll:try}). It backtracks when a (partial) assignment establishes a contradiction early (see Line~\ref{dpll:bt}). Moreover, it uses reduction rules $\mathcal R$ to prune paths. In this paper, we consider the following two:
%\VD{I have modified the below so that they fit the definition of reduction rules, which take the variable on input}
%\todo[inline]{I saw the need to extend the pure literal rule. See also the above note on partial assignments yielding SAT:}

DPLL branches on assignments to individual variables $x$ and $\bar x$.
It uses \emph{reduction} rules based on resolution to prune the resulting search tree.
In order to prune a branch, a reduction rule $\mathcal{R}$ (efficiently) \emph{under-estimates} whether a literal $l$ is entailed by a formula $F$, i.e.,
\[
F~\models_{\mathcal{R}} l \implies F \models l.
\]
We consider two rules:
\begin{itemize}
 \item Unit resolution: if there exists a unit clause $\set{x}$ (or $\set{\bar x}$),  then set value $x$  to true (false). 
 \item Pure literal: if variable $x$ only appears positively (negatively), then set $x$ to true (false).
\end{itemize}
%A \concept{branching heuristic} ensures that the pruning and backtracking happens early in the search (see \ref{dpll:heur}).
%Finally, the algorithm returns immediately when a satisfying assignment is found, since  disjunction at Line~\ref{dpll:return} short-circuits in that case.

\begin{figure}
\begin{algorithm}[H]
  \caption{The DPLL algorithm}
  \label{alg:dpll}
   \begin{algorithmic}[1]
   \Function{DPLL}{$F\colon CNF$}
   	\IfThen{$F = \emptyset$}{\Return{1}} \label{dpll:return}%\Comment{$\models F$}
   	\IfThen{$\emptyset \in F$}{\Return{0}} \label{dpll:bt}%\Comment{$\models F$}
   	\State $x \gets$ next var {acc. to branching heuristic}\label{dpll:heur}
   	%(e.g.\ unit/pure/random)
   	\If{$F~\models_h~(x=c)$ \textbf{with} $c\in \set{0,1}$}
       \State \Return{ DPLL($F_{|x=c}$)}\label{dpll:pos}
%    \ElsIf{$\bar{x}~\models_h~F$}
%       \State \Return{ DPLL($F_{|\bar{x}}$)}\label{dpll:neg}
	\Else
       \State \Return{ DPLL($F_{|x}$) $\vee$ DPLL($F_{|\bar{x}}$)} \label{dpll:try}
   \EndIf
   \EndFunction
   \end{algorithmic}
\end{algorithm}
\end{figure}

Algorithm~\ref{alg:dpll} corresponds to the DPLL algorithm in its classical form. 
It recursively assigns truth values to the variables of $F$ in a heuristically chosen order (see Line~\ref{dpll:heur}), while eagerly simplifying the formula (see $F_{|x}, F_{|\bar x}$ at Line~\ref{dpll:pos},\ref{dpll:try}). %,\ref{dpll:neg}
Branches are possibly pruned with the reduction rule at Line \ref{dpll:pos}. %\ref{dpll:neg}.
When the formula becomes (un)satisfiable, the recursion backtracks (Line \ref{dpll:return},\ref{dpll:bt}). Note that a short circuiting or at Line \ref{dpll:try} will cause the algorithm to terminate early when the formula is found to be satisfiable.

The algorithm explores a subtree of the full binary search tree. We can consider each node of the tree as uniquely labeled with a (restricted) formula or a partial assignment, as explained in more detail in \autoref{sec:hybrid}.
We call the nodes where pruning happens, the \concept{forced nodes}  and the others \concept{guessed nodes.}

A key element of the (heuristic) efficiency of backtracking
algorithms, is its ability to simplify subproblems as early as possible in order to prune the tree search. This is what DPLL achieves through its branching heuristic, which determines the order in which variables should be considered.
Notice that the order may change in the different branches of the tree, for example, because different unit clauses appear under different assignments in the concrete heuristic discussed above. 
On the other hand, one can of course always implement a static branching heuristic that simply selects variables according to a predetermined order, as we do in the next section for PPSZ.

\subsection{The PPSZ algorithm family}
\label{sub:ppsz-dnc}

%\VD{a lot of new text about resolutions...} Alfons: fixed.

%\red{The new text about resolutions was added by VD. Im okay with not having it here, but we need to explain what a resolution is. Noone in the q. community will know.}

The best exact algorithms for the (unique) $k$-SAT problem have for many years been based on the PPSZ (Paturi, Pudl\'ak, Saks, and Zane~\cite{og-ppsz}) algorithm. %The latest record holder version is Biased PPSZ~\cite{biased-ppsz}.
%Some of the best exact algorithms for solving SAT problems have for many years algorithm have been based on the PPSZ (Paturi, Pudl\'ak, Saks, and Zane \cite{og-ppsz}) algorithm.
The PPSZ algorithm is a Monte Carlo algorithm, i.e.~it returns the correct answer with high probability using randomization, although a derandomized version also exists \cite{rolf2005derandomization}.
Like DPLL, PPSZ uses heuristics to force or guess variables. The upper bound on the runtime is derived from the probability that a variable is guessed, as the expected number of guesses dictates the size of the search space.
We first explain the intuition based on a simple version that achieves a runtime of 
$\staro(2^{(1-\nicefrac{1}k+\varepsilon) n})$, for $\varepsilon > 0$.\footnote{In the remainder of the text, we will be using the standard notation for ``lazy'' scalings: with the superscript $\ast$, e.g.~$\staro$, we denote scalings which ignore only polynomially contributing terms (relevant when the main costs scale exponentially), just as tilde, e.g. $\tildeo$ highlights we ignore logarithmically contributing terms (when the main costs are polynomial in the relevant parameters).}
%It exploits lower bounds on the number of guessed nodes along the paths of the DPLL~\cite{og-ppsz,biased-ppsz}. These lower bounds are obtained from properties of the chosen {resolution} scheme.

Assuming there is only a single satisfying assignment $\vec u$, it is easy to see that the unit resolution heuristic can force on average $1/k$ variables for a random variable permutation $\pi \in S_n$: In the first place, for every variable $x$ there has to be a clause $C$ which forces the value of $x$, i.e., where the only true literal under assignment $\vec u$ is $x$ or $\bar x$, otherwise there would be multiple satisfying assignments. Second, if the variables are assigned according to the order $\pi$ and $\pi$ places $x$ last with respect to the other $k-1$ variables~of~$C$, then unit resolution will indeed identify the forced variable. As a result, a PPSZ algorithm merely has to search a space proportional to the number of guesses. For a $k$-SAT problem we let $\gamma_k$ denote the expected fraction of the $n$ variables which are guessed by PPSZ, such that $\gamma_k n$ is the expected number of guessed variables. The amount of variables which need to be guessed depends on the heuristic $h$ used to check if $F \models_h l$. If the heuristic consists solely of unit resolution,\footnote{The PPSZ algorithm where unit resolution is the only heuristic is in fact the PPZ algorithm~\cite{scheder-steinberger}.} then for unique\footnote{This approach can be generalized to a setting with multiple satisfying assignments~\cite{scheder-steinberger}.} $k$-SAT PPSZ needs to make at most $(1 - 1/k)n$ guesses.
%By trying $2^{(\gamma_k+\varepsilon) n}$ random assignments on random variable orders, it succeeds with constant probability.
When unit resolution is replaced with the more general $s$-implication (discussed below) something equivalent to the original PPSZ algorithm is obtained.

%To obtain better values for $\gamma_k$, the original formulation of PPSZ~\cite{og-ppsz} pre-processes the formula $F$ using an incomplete resolution scheme called $s$-resolution.
%This procedure repeatedly adds to $F$ all possible clausal resolvents of maximum clause size $s$, until a fixpoint is reached.
%Since $s$ is constant, this pre-processing step is poly-time.
%Thereby it achieves~$\gamma_k\approx 0.38$ for $k=3$.
 
%As a consequence, the backtracking search tree can be obtained using only:
%\begin{itemize}
%    \item the unit resolution rule,
%    \item a deterministically ordered branching heuristic,
%    \item and a search predicate that marks tree nodes corresponding to unsatisfying assignment or more than $(\gamma_k+\varepsilon)$ guesses as false, and
%        nodes corresponding to a satisfying assignments as true leaves.
%\end{itemize}
%Note that this is clearly a backtracking algorithm, as leaves can occur at different depths.
%To make the probability of finding a satisfying assignment approach 1, the above search has to be repeated a constant number of times with different variable orders.

\begin{figure}
\hspace{-1em}
\begin{algorithm}[H]
  \caption{dncPPSZ$_s$($F$,$\pi$)}
  \label{algorithm:dnc-ppsz}
   \begin{algorithmic}
%   \State $\pi$ permutation in $S_n$, $s \in \N$
   \IfThen{$F = \emptyset$}{\Return{1}}       %\Comment{$\models F$}
   \IfThen{${\emptyset \in F}$ \textbf{or} $n - \sizeof{\pi} > (\gamma_k+\varepsilon) n $}{\Return{0}}   %\Comment{$F\models 0$ or $d > (\gamma_k +\varepsilon) n$}
   \State $x,~ \pi \gets \pi[0],~ \pi[1\dots] $ \Comment{first var in $\pi$, postfix}  % first non-assigned variable in $\pi$
   	\If{$F \models_s (x = c)$ \textbf{with} $c \in \set{0, 1}$}
       \State \Return{dncPPSZ$_s$($F_{|x=c},\pi$)}
%    \ElsIf{$F \models_s \bar x$}
%       \State \Return{dncPPSZ$_s$($F_{|\bar{x}},\pi$)}
	\Else
       \State \Return{dncPPSZ$_s$($F_{|x}$,$\pi$)$\vee$dncPPSZ($F_{|\bar{x}}$,$\pi$)}
   \EndIf
   \end{algorithmic}
\end{algorithm}
\end{figure}

There are currently various versions of PPSZ
(see e.g.~\cite{rolf-ppsz,hertli-ppsz,hertli-breaking,biased-ppsz}),
but recent analysis~\cite{scheder21}
shows that the original PPSZ~\cite{og-ppsz} still has the best worst-case runtime of  $2^{(\gamma_k+\varepsilon) n}$, where for example $\gamma_3 \approx 0.386229$.
%or equivalently a search space of size $1.306973^n$. 
%replaces the resolution-based pre-processing by
The original PPSZ can be seen as using a reduction rule
called $s$-implication during the tree search~\cite{hertli-ppsz}.
A literal $l= x,\bar x$ is $s$-implied, written $F \models_s l$, if there is a sub formula of $s$ clauses, which implies it, i.e.~$G\subseteq F$ with $\sizeof{G}=s$ and $G \models l$.
In other words, if all satisfying assignments of $G$ set the variable $x$ to the same value.
%So in the modern version of PPSZ, the formula is not pre-processed, but the resolution rule is set to $s$-implication, rather than unit resolution.
%Due to better analysis of $s$-implication the value of $\gamma_k$ was slightly reduced~\cite{rolf-ppsz}.
% \begin{itemize}
%    \item the $s$-implication resolution rule,
%    \item a deterministically ordered branching heuristic, \&
%    \item a search predicate that marks tree nodes as false leaves when it corresponds to an unsatisfying assignment, or after $(\gamma_k+\varepsilon)$ guesses, and as true leaves when it corresponds to a satisfying assignment.
%\end{itemize}

The original PPSZ algorithm evaluates random assignments on random variable orders and is not immediately amenable for (quantum) backtracking.
Here we introduce a backtracking version of PPSZ that offers the same runtime guarantees as recent PPSZ versions, while allowing for the benefits of backtracking on a quantum computer (i.e., heuristically pruning search). Algorithm~\ref{algorithm:dnc-ppsz} shows dncPPSZ, which is similar to DPLL with $s$-implication as reduction rule. It takes a fixed (random) variable order $\pi$ and prunes branches after $(\gamma_k+\varepsilon) n$ guesses have been performed. Since the expected number of guesses over all variable orders equals $\gamma_k n$
for $s$-implication~\cite{scheder-steinberger}, Markov's inequality tells us that, for any $\varepsilon > 0 \in \Theta(1)$, with constant probability we can find the satisfying with fewer than $(\gamma_k+\varepsilon) n$ guesses for a random~$\pi$.
Hence we obtain \autoref{prop:ppsz-is-dnq} (see Appendix~\ref{proof:ppsz-is-dnq} for more detail and a proof).

%While the notion of $s$-implication is weaker than the notion of $s$-resolution, it can also be computed in polynomial time for constant $s$,
%%\footnote{Unlike general implication, $s$-implication can be performed in polynomial time for a constant $s$, making it a suitable reduction rule.}
%and leads to the same lower bounds on the success probability of PPSZ.\footnote{PPSZ with $s$-implication and PPSZ with $s$-resolution are commonly referred to as \textit{weak} PPSZ and \textit{strong} PPSZ respectively. Upper bounds on the success probability of PPSZ are sensitive to this distinction (see e.g.~\cite{pudlak-scheder-talenbanfard})}
%Both of the discussed PPSZ versions are fundamentally Monte Carlo algorithms, as sketched above.

%In a sentence, the difference between the original PPSZ and our interpretation is as follows: original PPSZ explores one path at a time for both a random permutation and a random set of branching choices, whereas our algorithm explores all branching choices for a given permutation, before moving on. Nonetheless, they are equally efficient.
\begin{proposition}
\label{prop:ppsz-is-dnq}
Executing dncPPSZ a constant number of times, using random variable orders $\pi$, is sufficient to decide satisfiability. The run-time of dncPPSZ is upper bounded by the run-time of the standard PPSZ algorithm.
\end{proposition}

We highlight that the PPSZ algorithm involves two steps: the ``core'' of the algorithm which is the \emph{PPSZ tree search} performed by the dncPPSZ algorithm, which takes an ordered formula on input (i.e.~the variable ordering is fixed, by e.g. the natural ordering over the indices). And an outer loop that repeats the PPSZ tree search a constant number of times, randomizing the order for each call. This is critical in our subsequent analysis.

\subsection{Quantum algorithms for tree search}
\label{sec:quantum}

The extent to which quantum computing can help the exploration of trees, and more generally, graphs is a long-standing question with infrequent but noticeable progress.
In the context of backtracking (for concreteness, for SAT problems), up until relatively recently, the best methods involved Grover search~\cite{grover1996fast} over all possible satisfying assignments (the leafs in the full search tree).
In such an approach, we would introduce a separate register of length $d$ equal to the number of variables and a \concept{search predicate}, implementing the satisfiability criterion based on the input formula.
This brute-force approach however yields a redundant search over branches which would be pruned away by the resolution rules, and, in the worst case may force a quantum computer to search an exponentially larger space than the classical algorithm would.
%\footnote{We note that there are very successful SAT solving algorithms which rely on sampling and not tree search, e.g. Sch\"{o}ning's algorithm, in which case a full quadratic speed-up is possible~\cite{ambainis2004quantum}.}
%We discuss this further in Section~\ref{sec:gvb}.

Nonetheless, Grover provides an advantageous strategy whenever the classical search space is larger than $2^{\nicefrac n2}$, hence this was the method used in previous work on hybrid divide-and-conquer strategies~\cite{dnq-eppstein,dnq-schoning}.
One example is Sch\"{o}ning's algorithm for SAT, which allows a full quadratic quantum speed-up~\cite{ambainis2004quantum} (but was since surpassed by PPSZ).
%We discuss this further in Section~\ref{sec:gvb}.
A particular advantage of Grover's search is that it is frugal regarding time and space: beyond what is needed to implement the oracle (the search predicate), it requires at most one ancillary qubit, and very few other gates.

More recently, Montanaro~\cite{qbacktracking}, Ambainis \& Kokainis~\cite{ambainis-kokainis}, and Jarret \& Wan~\cite{jarret-wan} have given quantum-walk based algorithms which allow us to achieve an essentially full quadratic speed-ups in the number of queries to the tree (when trees are exponentially sized). We will refer to the underlying method as the \emph{quantum backtracking method}.

The quantum backtracking method can be applied whenever we have access to local algorithms which specify the children of a given vertex, and whenever we can implement the search predicate, indicating the satisfiability of a (partial) assignment.
At the heart of the routine is the construction of a Szegedy-style walk operator $W$ over the bipartite graph specified by the even and odd depths of the search tree (details provided in Appendix~\ref{app:qpe}); Montanaro shows that the spectral gap of $W$ reveals whether the underlying graph contains a marked element (a satisfying assignment, as defined by the search predicate $P$). The difference in the eigenphases of the two cases dictates the overall run-time, as they are detected by a \concept{quantum phase estimation} (QPE) overarching routine. The overall algorithm calls the operator $W$ no more than $O(\sqrt{Tn} \log(1/\delta))$ times, for a correct evaluation given a tree of size $T$, over $n$ variables (depth), except with probability $\delta$. This is essentially a full quadratic improvement when $T$ is exponential (even superpolynomial will do) in~$n$. For the algorithm to work, one assumes that the tree size is known in advance, or at least, a good upper bound is known.

{The original paper on quantum backtracking implements DPLL with the unit rule~\cite{qbacktracking}, and led to a body of work focused on improvements~\cite{ambainis-kokainis,jarret-wan} and applications~\cite{practical-qbacktracking,moylett-linden-montanaro,ambainis-kokainis,campbell-khurana-montanaro}.}

\paragraph*{Tree search implementation.}
The implementation details are important for the discussion of the efficiency of later approaches,
so we provide a framework for the implementation of quantum backtracking. If the functions specified in the framework are implemented (reversibly),
then quantum backtracking can be implemented with almost no space overhead, as we provide a space efficient implementation of QPE (see \autoref{app:qpe}). The runtime overhead is multiplicative,
so any polynomial implementation of the framework will do since we focus on obtaining a reduction in the exponent in the hybrid method.

The trees $\mathcal T$ that we consider are characterized by the backtracking functions specifying the local 
tree structure for some input formula $F$, as described in Section~\ref{sec:background}. The nodes $v$ of the tree contain the information to represent the restriction $F_{|\vec x}$ for some partial variable assignment $\vec x$. In the classical approach (DPLL) each node is represented by the restricted formula itself (as a clause database~\cite{biere2009handbook}).

Here, we highlight a duality between partial assignments $\vec{x}$ and corresponding restricted formulas $F_{|\vec{x}}$ -- instead of defining the tree in terms of partial assignments, one can define it in terms of restricted formulas, and we will often use this duality as storing partial assignments for a fixed formula requires less memory than storing a whole restricted formula.
This duality is also important as it allows us to see backtracking algorithms as divide-and-conquer algorithms: each vertex in a search tree -- the restricted formula  $F_{|\vec{x}}$ -- corresponds to a restriction of the initial problem, and children in turn correspond to smaller instances where one additional variable is restricted. 
%%The leaves of the tree consist of empty formulas, corresponding to satisfying assignments, and formulas containing an empty clause, corresponding to assignments that establish a contradiction (an empty conjunction is vacuously true and an empty disjunction is vacuously false).
Consequently backtracking algorithms can often be neatly rewritten in a recursive form. 

%In quantum backtracking, we construct a quantum walk operator whose spectral properties reveal whether there is a leaf representing a satisfying assignment. We omit the details here, because 
To implement the walk operator for quantum backtracking, it suffices to be able to implement a unitary operator which given a vertex $v$, produces the children of $v$ and one that decides whether the partial assignment corresponding to node $v$ makes the formula true.
For this, we specify the functions $ch1,ch2, chNo, P$ below. 
Appendix~\ref{app:backtracking} shows how the walk operator can be implemented
based on these functions.
% Appendix~\ref{app:qpe} provides a space efficient quantum circuit for 
\begin{itemize}
 \item a function $ch2(v,b)$, which takes on input a vertex $v$, and returns one of the children, specified by the bit $b$;
 \item a function $ch1(v)$ which returns the single child if $v$ is forced;
 \item a function $chNo(v)$, which returns whether $v$ has zero, one or two children, and
 \item a search predicate $P$ which identifies leafs and solutions, i.e., $P(\vec x) = b$ if
$F_{|\vec x} = b$ with $b \in \set{0,1}$, and $P(\vec x) = \bot$ otherwise.
\end{itemize}

%Note that %the above function implement the branching heuristic and the reduction rule. 
%for quantum backtracking, it suffices to implement functions $ch1,ch2, chNo, P$ \emph{reversibly} (for Grover-based search, only an oracle $P$ suffices). 
As in many cases, the classical, non-reversible implementations are efficient, the runtime of the reversible version is often not a problem, as long as it stays polynomial (although even sub exponential will suffice).
However, in the context of hybrid methods, the space complexity of these implementations becomes vital.
Standard approaches efficiently implement a data structure which can represent the last vertex visited, and update this structure with the next vertex reversibly. The reversibility prevents deleting information efficiently, so the simplest solution is to store the entire branching sequence, reconstructing the state from this succinct representation.
Indeed, the technical core of the papers~\cite{dnq-eppstein,dnq-schoning}, and parts of the remainder of this paper, focus on low space algorithms solving this problem.

%The notion of partial assignments naturally leads to backtracking algorithms for SAT. The search tree of such an algorithm is most often a tree graph, where the vertices of the graph are partial assignments, the root of the tree is the ``empty'' partial assignment $\ast^n$, the children of a node extend its partial assignment by  restricting one more variable, and the leaves are either full assignments, or partial assignments which establish a contradiction early. Notice that the latter can lead to pruning of the search tree.

%Backtracking-based SAT algorithms differ in how the search tree is defined, but generically, they contain a number of elements, all of which fix how the children of a partial assignment are selected, and how contradictions are established.

\paragraph*{Online algorithms.} Note, the quantum algorithm for quantum backtracking always performs $O(\sqrt{Tn} \log(1/\delta))$ queries no matter which vertex/leaf is satisfying, if any.
In contrast, classical backtracking is an \emph{online algorithm}, meaning that it can terminate the search early when a satisfying assignment is found.
This naturally depends on the traversal order of the tree, which is also specified by the algorithm (and perhaps the random sequence specifying any random choices), and thus the maximal speed-ups are only achieved for the classical worst cases when either no vertices are satisfying, or when the very last leaf to be traversed is the solution.

In~\cite{ambainis-kokainis}, the authors provide an extension to quantum backtracking, which allows one to estimate the size of the search tree. Moreover, it can also estimate the size of a search space corresponding to a partial traversal of a given classical backtracking algorithm, according to a certain order.

With this, it is possible to achieve quadratic speed ups not in terms of the overall tree, but rather the search tree limited to those vertices that a classical algorithm would traverse, before finding a satisfying assignment.
In this case, we have a near-full quadratic improvement, whenever this effective search tree, which depends on the search order, is large enough (superpolynomial).

\subsubsection{Grover vs quantum backtracking in tree search}
\label{sec:gvb}

One can employ Grover search-based techniques to explore trees of any density and shape. 
Yet in many cases using Grover can be significantly slower than by employing the quantum backtracking strategy, or even classical search.
While the query complexities of classical and quantum backtracking depend on the tree size, the efficiency of quantum search based on Grover is bounded by the maximal number of branches that occur in any path from root to the leaf of the tree. The branching number is also vital in space complexity analyses.
\begin{definition}[Maximal branching number]\label{def:br}
Given a tree $\mathcal{T}$, the maximal branching number~$br(\mathcal{T})$ is the maximal number of branchings on any path from root to a leaf. 
If a sub-tree $\mathcal{T}',$ is specified by its root vertex $v$, with $br(v)$ we denote  $br(\mathcal{T}')$.
\end{definition}

Here is a sketch of how Grover's search would be applied in tree search leading to the query complexity dependence on branching numbers.  For the moment, we assume that all the leaves of the tree are at distance $n$ from the root; this can be obtained by attaching sufficiently long line graphs to each leaf occurring earlier. This may cause a blow up of the tree by at most a factor of $n$, but this will not matter in the cases where trees are (much) lager, i.e.~exponential.

Our application of Grover's search over a rooted tree structure follows the following principle. We make use of an ``advice'' register, which tells the algorithm which child to take, when two options are possible, i.e.~select the $b$ parameter of $ch2$ when $chNo(v) = 2$.
Let $v_1, v_2, \ldots, v_n$ be a path from the root $v_1 = r$ to a leaf $v_n$.
For each guessed $v_i$ along the path, i.e.~where $chNo(v_i) = 2$, a subsequent unused bit of the advice string determine the choice. Finally, the $n^{th}$  vertex is checked by the search predicate $P$. For this process to be well defined, we clearly require as many bits in the advice register as the largest number of branches along any path.

%$F=(\overline{x_1}\lor \overline{x_2}) \land (\overline{x_3} \lor \overline{x_4}, {x_5})$

We now show that the size of the advice can also be independent from the actual tree size.
To do so, we first discuss tree shapes for which Grover exhibits  extremal behavior.

Take a ``comb'' graph for instance, where one edge is added to every node of a line graph, except to the last node. This graph has $2n-1$ vertices, and $(n-1)$ branches whereas a full binary tree with $2^n$ vertices has the same maximal number of branches (this disparity persists even when we ``complete'' the comb graph by extending the single line to ensure every path is length $n$).
The Grover-based algorithm thus always introduces an effective tree which is exponential in the number of branchings.

Note this need not lead to exponential search times; e.g., in the example of the comb graph, if the leaf-child of the root is satisfying,
% (or leading to the satisfying leaf without branches along the way), 
then half ($2^{n-1}=2^n/2$) 
of the strings in the advice register will result in finding the satisfying assignment,\footnote{If it holds that $P(v)=1$ for node $v$, then $P(v')=1$ for all descendants of $v'$ of $v$.}
leading to an overall constant query complexity. 

More generally, in the case of search with a single satisfying assignment $v$, the worst-case query complexity using Grover's approach, will always be $O(2^{\nicefrac b2})$, where $b$ is the number of branches on the path from the root to $v$, i.e, $b = br(r) - br(v)$,
where $r$ is the root of the tree.
This is because all the  vertices below $v$ will represent solutions, as explained in the previous paragraph.

Comparing Grover's search with backtracking in the sense of query complexity, we find the following: any binary tree of size $O(2^{\gamma n})$ must contain a path from the root to a leaf which has more than $\gamma n$ branches -- if the tree contains no more than $\gamma n$ branches along any path, i.e.~its branching number is $\gamma n$, then Grover's method will achieve a run-time of $O^\ast(2^{\gamma n/2})$ as well.

In other words, a limitation on the tree size -- which upper-bounds the classical search run-time -- directly limits the quantum backtracking query complexity, whereas it does not (as directly) limit Grover-based search. However,  as we will show in the case of PPSZ, when the tree size is estimated on the basis of ($\gamma n$) maximum branching numbers, this does provide a way to directly connect classical search with Grover-based method.

Unlike backtracking, Grover can also be used as an online algorithm, as discussed earlier in this section.
In the majority of this paper, we will be concerned with worst-case times, so the online methods of~\cite{ambainis-kokainis} will not be critical. Although, for practical speed-ups, they certainly are. We reiterate that all the results that we will present, can accommodate these tree size-estimation based methods of~\cite{ambainis-kokainis}.

\subsection{The hybrid divide-and-conquer method}
\label{sub:hybrid-dnc}

The hybrid divide-and-conquer method was introduced to investigate to which extent smaller quantum computers can help in solving larger instances of interesting problems. Here the emphasis is placed on problems which are typically tackled by a divide-and-conquer method. This choice is one of convenience.
 Any method which enables a smaller quantum computer to aid in a computation of a larger problem must somehow reduce the problem to a number of smaller problems. How this can be done is critical for the hybrid algorithm performance.
 
But in divide-and-conquer strategies, there exists an obvious solution. Divide-and-conquer strategies recursively break down an instance of a problem into smaller sub-instances, so at some point they become suitably small to be run on a quantum device of (almost) any size.\footnote{More precisely, the device must large enough to handle any size instance of a given problem, which is in not a trivial condition.}

In previous work on the hybrid divide-and-conquer method~\cite{dnq-eppstein,dnq-schoning} this approach was used for a de-randomized version of the algorithm of Sch\"{o}ning, and for an algorithm for finding Hamilton cycles in degree-3 graphs. In general, in these works (and the present work) the question of interest is to identify criteria when speed-ups, in the sense of provable asymptotic run-times of the (hybrid) algorithms, are possible. 

%In this framework, we assume a problem of size $n$ (the number of variables in SAT, or the vertices in the graph for the Hamilton cycles problem), and access to a smaller quantum computer of size $m = \kappa n,$ with $\kappa = O(1)$.

{
\subsubsection{Quantifying speed-ups of hybrid divide-and-conquer methods}
}
For the quantum part of the computation,
the complexity-theoretic analysis of meta-algorithms (more precisely, oracular algorithms),
%As tis ypically the case in quantum computing, 
 such as Grover and quantum backtracking, predominantly measures \emph{query complexity}. This is the number of calls to the e.g. Grover oracle, or the walk operator, respectively, i.e.~the black boxes that implement the predicate detection, and the search tree structure. 
 %In classical algorithms on the other hand, we often measure complexity in more basic operations.

As the present work is only concerned with exponentially sized trees, and sub-exponential time sub-routines, the complexity measure we use for the classical algorithm $\mathcal{A}_C(n)$ is the size of the tree (i.e, the classical query complexity) and for quantum algorithm the query complexity. 
In the hybrid cases, we will thus measure the totals of classical and quantum query complexities, treated on an equal footing.

We are thus interested in the (provable) relationships between quantities Time($\mathcal{A}_C(n)$) describing the run-time of the classical algorithm given instance size $n$, and 
Time($\mathcal{A}_H(n,\kappa)$),\footnote{As we will discuss later, in the cases of the quantum algorithms we will consider, the relevant notion of instance size may be different from what is relevant for the classical algorithm complexity.  However to compare the run times, we will always have to bound both hybrid and classical complexities in terms of same quantities. } describing the run-time of the hybrid algorithm, having access to a quantum computer of size $m= \kappa n$  with $\kappa = O(1)$.\footnote{Since we consider speed-ups in the sense of asymptotic run-times, considering smaller sizes, e.g. constant sized quantum computers, or log sized quantum computers makes little sense as these are efficiently simulatable.}
We focus on exact algorithms for NP-hard problems (which typically have exponential runtime).
\ \\
\begin{definition}[Hybrid speed-ups]
\label{def:speedups}
\ \\
We say that \emph{genuine} speed-ups (i.e.~polynomial speed-ups) are possible using the hybrid method, if there exists a hybrid algorithm such that
\EQ{
\textup{Time}(\mathcal{A}_H(n,\kappa)) = O(\textup{Time}(\mathcal{A}_C(n))^{1-\epsilon_\kappa} ),
}
for a constant $\epsilon_{\kappa}>0$.
If such an $\epsilon_{\kappa}$ exists for all $\kappa>0$, then we say the speed-up is \emph{threshold-free}.
\end{definition}

What originally sparked the interest in hybrid algorithms, is the fundamental question whether \emph{threshold-free} speed-ups are possible at all. This was answered in the positive in the previous works mentioned above~\cite{dnq-eppstein,dnq-schoning}.

In the above, we have assumed that the time complexity can be fully characterized in terms of the instance size $n$. In general, as discussed in Section~\ref{sec:gvb}, the complexities may be precisely established only given access to a number of parameters, such as, search tree size or even less obvious measures like the location of the tree leaf in the search tree.% We will discuss these options in Section~\ref{sec:gvb}.

\subsubsection{Limitations of existing hybrid divide-and-conquer methods}
\label{subsub:limitations}
The key impediment to threshold-free speed-ups is the space-complexity of the quantum algorithm for the problem.
To understand this, assume that with $\kappa' n$ we denote the size of the instance we can solve on a $\kappa n$ sized device (so $\kappa n = \textup{Space}(\mathcal{A}_Q(\kappa' n))$, where $\textup{Space}(\mathcal{A}_Q(n))$ denotes the space complexity of the quantum algorithm).
If the space complexity is super-linear, then $\kappa'$ itself becomes dependent on $n$, and in fact, decreasing in $n$. In other words, as the instance grows, the \emph{effective fraction of the problem} that we can delegate to the quantum device decreases. Then, in the limit, all work is done by the classical device, so no speed-ups are possible (for details see~\cite{dnq-eppstein,dnq-schoning}).

In~\cite{dnq-eppstein}, these observations also lead to a characterization of when genuine speed-ups are possible for recursive algorithms, whose run-times are evaluated using standard recurrence relations. 
Notably, none of the classical algorithms employed backtracking nor early pruning of the trees. 
These properties ensured that the search space could be expressed as a sufficiently dense trees, ensuring that no matter where we employ the (faster) quantum device, a substantial fraction of the overall work will be done by the quantum device, yielding a genuine speed-up. 

Consequently, the main technical research focus in these earlier works were to establish highly space-efficient variants of otherwise simple quantum algorithms (in essence, at most linear in $n$), which are all based on Grover's algorithm over appropriate search spaces.
Both the de-randomized algorithm of Sch\"{o}ning~\cite{dnq-schoning}, and Eppstein's algorithm for Hamilton cycles~\cite{dnq-eppstein} traverse search trees dense enough that Grover-based search (which is itself space-efficient) over appropriate search spaces, yields a polynomial improvement. (Appendix~\ref{app:eppstein} describes an improvement over the hybrid solution of~\cite{dnq-eppstein} based on the new framework developed in Section~\ref{sec:hybrid}.)
%\MR{Refer to the  here. I can't really find a nice way to bring it up yet.}

An additional limitation is that any speed-up obtained through this hybrid approach is always relative to a fixed classical algorithm, as established in the previous works. This also implies that even in the event that we obtain a hybrid speed-up for the best algorithm for a given problem, this speed-up, and specially, any chance of a threshold free speed-ups disappear whenever a new faster algorithm is devised. In other words, the speed-ups are algorithm-specific. 
This will always be true, unless lower bounds are proved for the problem that these classical algorithms attack, in which case we may talk about \emph{algorithm independent} speed-ups.
 
In the following two sections we study hybrid divide-and-conquer in the context of search trees, and provide a number of new algorithms using the hybrid method.

\section{Hybrid divide-and-conquer for backtracking algorithms}
\label{sec:hybrid}

In the present work, we consider the hybrid divide-and-conquer method for algorithms which operate by searching over a suitable tree. This framework naturally captures backtracking algorithms, and recursive algorithms as studied in~\cite{dnq-eppstein}.
%\footnote{In the case of recursively specified algorithms, the recursive calls establish a tree structure, where the vertices are labeled by the specific sub-problems that are tackled in the recursive step, or rather, by any specification of such sub-problems.}
Our new framework focuses on scenarios in which the search is over unbalanced trees, unlike the Grover-based methods utilized in~\cite{dnq-eppstein,dnq-schoning}.

This section investigates the structure and properties of hybrid divide and conquer algorithms from the search tree perspective. Most notably, it introduces the concept of tree search decomposition. These considerations influence the design of the new hybrid divide-and-conquer algorithms discussed later.

%To help the reader navigate this section we provide its outline.
The outline of this section is as follows.
In Subsection~\ref{sub:new-hybrid}, we discuss how specifically the tree structure influences whether or not (provable) polynomial speed-ups can be obtained on an intuitive level. Then, in Subsection~\ref{sec:criteria}, we provide a theorem providing a general, albeit not very operative, characterization of when speed-ups are possible.
In Subsection~\ref{sect:spec-intro}, we connect more closely the properties of the quantum algorithms with the tree structure identifying assumptions which allow for significant simplifications. In Subsection~\ref{sec:special}, quantitative speed-ups are proved for constrained (but still rather generic) cases.
These special cases are then exemplified in Section~\ref{sec:concrete}, and Section~\ref{sec:discussion} addresses some of the scenarios were these simplifying conditions cannot be met.

\subsection{Search tree structure and potential for hybrid speed-ups}
\label{sub:new-hybrid}

In previous works, the importance of the space efficiency of quantum algorithms was put forward as the key factor determining
whether asymptotic (threshold-free) hybrid speed-ups can be achieved, as discussed in detail in~\cite{dnq-eppstein}.
Naturally, in the hybrid backtracking setting, we will inherit the same limitations, and some new ones, which we focus on next.

%Before any other details, we highlight a critical assumption we are always making which connects the properties of the search tree, or rather, the overall algorithmic approach, and the definition of the hybrid setting we wish to consider, namely, the definition of the parameter $n$. 
A crucial aspect of our hybrid framework is that we always assume that the size of the quantum computer is $\kappa n$, where $\kappa \in (0,1]$ and $n$ is a \emph{natural} problem instance size. 
Note that the instance size is not unambiguous (as discussed in Section~\ref{sec:discussion}). Nonetheless, we assume this is well-defined, and that the classical backtracking algorithm we consider generates search trees of height (at most) $n$.\footnote{In the next section's examples, with the exception of the enhanced hybrid algorithm for Hamilton cycles on cubic graphs discussed in Appendix~\ref{app:eppstein}, $n$ specifically designates the number of variables of a given Boolean formula (and not the formula length itself).}
%, i.e.~each child designates a strictly smaller instance with respect to the relevant size measure that $n$ quantifies.\footnote{In previous works, the subtle separation between two distinct notions of instance size in the context of both Hamilton cycles and Sch\"oning SAT is what allows for space efficient algorithms~\cite{dnq-eppstein}.}\todo{incomprehensible footnote}
%TODO: Also not true that each not strictly decrease problem size (think of disappearing vars)

Next, to understand how the tree structure may influence the overall algorithm performance, we introduce the search tree decomposition. % although one can provide a fully abstract definition, we provide one in context. 
Consider a search tree $\mathcal{T}$ generated by algorithm $\mathcal A$ on some instance of size $n$, where every vertex in the tree denotes a sub-problem which can (in principle) be delegated to a quantum computer, running  the algorithm $\mathcal A_{Q}$ in a hybrid scheme, which we will denote with subscript~$H$. Note, $\mathcal A_{Q}$ can be the quantum backtracking or Grover version of~$\mathcal A$, or a related algorithm for solving the same problem.

%Recall that, to each vertex of the search tree, we can associate a sub-instance of the problem specified by the root.

In general, $\kappa n$ does not provide enough space for $\mathcal A_{Q}$ to be run on the entire problem, i.e.~on the root of the tree, but it can be run on some of sub-instances represented by vertices deeper in the tree.
This of course depends on the space efficiency features of $\mathcal A_{Q}$ -- the \emph{effective size} of the quantum computer we have, but we will not focus on this for the moment. 

We merely assume that whether or not  $\mathcal A_{Q}$ can be run on an instance $v$ given an $\kappa n$-sized device is monotonic with respect to the tree structure; that is, if it can be run on $v$, it can be run on all descendants of $v$.\footnote{It is conceivable that a sub-instance, as defined by a classical backtracking algorithm, for some reason takes more quantum space than the overall problem, and our formalism can be adapted to treat this case. However, in the cases of all quantum algorithms which we consider here, this will not be the case, hence we focus on this more intuitive scenario.}
%TODO: too vague

With this in mind, we can identify the collection of $\jcut$ \textit{cut-off vertices} $\{c_k\}_{k=1}^{\jcut}$. That is, all the vertices that can be run on the quantum device, whose parents cannot be run on the same device due to size considerations. In principle, $\jcut$ may be zero for some trees.
These vertices correspond to the largest instances in the search tree where we can use the quantum computer.

Now the \emph{search tree decomposition} is characterized by the set of sub-trees $\{\mathcal{T}_j\}_{j=0}^{\jcut}$, where the subtree $\mathcal{T}_j$ is the entire sub-tree rooted at the cut-off vertex $c_j$, and where $\mathcal{T}_0$  denotes the tree rooted at the root of $\mathcal{T}$, whose leaves are parents of the cut-points $\{c_j\}_j$. We refer to $\mathcal{T}_0$ as the ``top of the tree'', and it is traversed by the classical algorithm alone.
Note that by ``gluing'' all the subtrees of $\{\mathcal{T}_j\}_j$ as the corresponding leaves of $\mathcal{T}_0$, we obtain the full tree $\mathcal{T}$.

With $T_j$ we denote the size of the tree $\mathcal{T}_j$.
The set $\{\mathcal{T}_j\}_j$ we call the search tree decomposition at cut-off $c$, and note that it holds
\EQ{
 T = T_0 + \sum_{1 \leq j \leq J_c} T_j.
}

Next, we briefly illustrate in terms of the search tree decomposition, on an intuitive level, in which cases speed-ups can be neatly characterized and achieved and the cases where the speed-up fails.

The classical algorithm will (in the worst case) explore the entire search tree requiring $T_j$ steps, for each sub-tree $\mathcal{T}_j$.
Note that $T_j$ characterizes the upper bound on the query complexity of the classical algorithm. Let us for the moment assume that this roughly equal to the overall run-time of the classical algorithm (we will discuss shortly when this is a justified assumption).

Further, assume that the quantum algorithm achieves a pure quadratic speed-up in run-time over the classical algorithm. Then
the {hybrid} algorithm will take $T_H = O^*(T_0 + \sum_{1 \leq j \leq J_c} \sqrt{T_j})$ time (queries), where $O^*(f(x)) \defn O(f(x)\cdot \poly(x))$
to ignore polynomial factors.
To achieve a genuine speed up in the query complexity (from which we will be able to discuss total complexity), it must hold that $T_H$ {is upper bounded by} $T^{1-\epsilon}$ for some~$\epsilon$ (see Definition~\ref{def:speedups}), %\footnote{{The equality need not be exact, but we stick with this for simplicity of exposition.}} 
 where $T = \sum_{k=0}^{J_c} T_k$ is the total tree size.
 
Now this actual speed-up clearly depends on the cut-off points, and the structure of the original tree, as nothing a priori dictates the relative sizes of all the sub-trees. 
We first consider the case where the tree is a balanced, complete binary tree: The tree size is exactly $2^{n}-1$ where $n$ is the tree height. Further, assuming that the quantum algorithm can handle $\gamma n$-sized instances (in terms of this natural instance size) on the $\kappa n$-sized device, we get the following decomposition:
 $T_j = 2^{\lfloor \gamma n \rfloor},\ j>0 $, $T_0 = 2^{n - \lfloor \gamma n \rfloor}$. For simplicity, we shall ignore rounding, with which we obtain:
 
\begin{align}
T_H &=  T_0 + \jcut \sqrt{T_1}\\
&=    2^{n - \gamma n} -1 +   2^{n - \gamma n} (2^{\gamma n/2}-1) \\
%& \approx   2^{(1-\gamma)n} (1 + 2^{\gamma/2}) \\
&\approx 2^{(1-\gamma/2) n}
\end{align}
%where the approximations come from the ignoring the $\pm 1$ contributions.
It is clear that the above example constitutes a genuine speed up with $\epsilon = \gamma/2$, with a full quadratic speed-up when $\gamma =1 $, i.e.~when we can run the entire tree on the quantum computer.

At this point it also becomes clear how the space efficiency of the quantum algorithm comes into play: If it is linear in the natural size $n$, then $\gamma$ is a fraction of $\kappa$, and we obtain threshold-free asymptotic speed-ups. However, if the space required is super-linear in $n$, $\gamma$ becomes a decaying function of $n$, in which case all speed-ups vanish. We shall briefly discuss these aspects shortly, and for more on the space complexity constraints see \cite{dnq-schoning,dnq-eppstein}. Here we first focus on the issues stemming from the tree structure alone.

We can equally imagine a tree where $\mathcal{T}_0$ is a full tree of size $2^{n/2}$, and $\jcut =1$ with another full tree $\mathcal{T}_1$. In this case the run-time (more precisely, query complexity) of the classical algorithm is $2\times 2^{n/2}= O(2^{n/2})$, but so is the quantum query complexity:  $2^{n/2} + 2^{n/4} =O(2^{n/2})$, so no speed-up is obtained. 
In what follows, we define $\jleaves$ to count all the leaves of $\mathcal{T}_0$. This is to take into account that not only do the individual subtrees need to be large, but also need to be many in number, compared to $\mathcal{T}_0$ itself, which can be bounded by the number of leaves.

In the next section, we make the above structural-only considerations fully formal.

\subsection{Criteria for speed-ups from tree decomposition}
\label{sec:criteria}

In the following, we assume to have access to a quantum computer of size $m = \kappa n$,  we consider a classical backtracking algorithm $\mathcal A$, which generates a search tree~$\mathcal{T}$. Through the study of a particular tree, one can extrapolate to a family of trees induced by the problem.

%Each vertex of the search tree specifies a sub-instance of the problem and the corresponding sub-tree (which is the search tree of the sub-instance).

We imagine designing a hybrid algorithm based on the classical algorithm $\mathcal A$, and a quantum algorithm $\mathcal{A}_Q$ used when instances are sufficiently small. 

We now need to characterize space and query complexities of the classical and quantum algorithm in terms of the properties of the trees.
\paragraph{Query complexity}
While the query complexity of the classical algorithm $\mathcal A$ is exactly the size of the search tree of the problem instance,
the query complexity of the quantum algorithm may be more involved. In the case of quantum backtracking, it is a function of the tree size and height (variable number). The situation simplifies when the trees are large enough (super-polynomial in size), %when we can ignore the height,
because the query complexity then becomes equal to the square root of the tree size
 up to polynomial factors.

In the case of Grover-based search, an upper bound can established by brute forcing all branching decisions along all tree paths of maximum length $n$.
Therefore, the query complexity is exponential in  the max-branching number, 
which in general has no simple relationship to the tree size.
%In the case of quantum backtracking, the relationship to the tree size is more direct. 
%As the trees need not be regular in any way, the relation to tree height may not be.

Further, we highlight that the connection between the query complexities overall run-time will not aways be possible, except when: 
\begin{enumerate}
 \item the complexities of the subroutines realizing a query are assumed to be polynomial~in~$n$,
 \item the query complexity is exponential, so polynomial factors can be ignored.
\end{enumerate}
We embed these assumptions in our main result.
%These assumptions we will embed in our main theorem.\\

\paragraph{Space complexity}
The space complexity of $\mathcal{A}_Q$ determines the cut-off points in the search tree. 
The space complexity may not be a simple function of the tree height, as it depends on how the vertices are represented in memory (for instance as satisfying assignment or as the branching choices in the case of a low maximum branching number).\footnote{While Section~\ref{sec:quantum} discussed the duality between formula and satisfying assignment representations, we will never employ the formula representation in the quantum algorithms due to the limited available memory.}
%We shall discuss the impact of the vertex representation later.

To achieve clean polynomial speed-ups, as discussed previously,  two main factors must conspire:
\begin{itemize}
 \item the space complexity must yield a tree search decomposition where many of the subtrees which will be delegated to the quantum computer are large enough for a substantial advantage to be even in principle possible,
 \item and the quantum algorithm must actually realize such a polynomial advantage.
 \end{itemize}

%Let the  space complexity of the quantum algorithm, together with the quantum computer size constraint yield the search tree decomposition  $\{\mathcal{T}_j\}_{j=0}^{J}$, with sizes $\{{T}_j\}$, where as usual $\mathcal{T}_0$ denotes the ``top tree'' explored by the classical algorithm alone. The problem sub-instances with trees $\mathcal T_j$ are rooted in the vertex $v_j$.

One thing to highlight in the above is the observation that the sub-trees $\mathcal T_j$ must not only be large on average, they need to be numerous, relative to the total size $T_0$ of the top of the tree.
In order to take both the average size and the amount of sub-trees $\mathcal{T}_j$ into account, we do the following: Instead of considering the sub-trees $\{\mathcal{T}_j\}_{j=1}^{J_c}$, with $0 \leq J_c \leq J_\ell$, we instead consider there to be $J_\ell$ sub-trees, of which $J_\ell - J_c$ have size~0.

%With this in mind, and as mentioned earlier, in the following it will be convenient to consider the \emph{extended} search tree decomposition, where we add two empty trees for each leaf in $\mathcal{T}_0$ which is not a parent to a root of a sub-tree $\mathcal{T}_{j>0}$. For this extended tree, $\jleaves$ is equal to the number of non-trivial sub-trees plus twice the number of leaves in the original $\mathcal{T}_0$ which are above the cut-off in effective size.

We can now identify certain sufficient conditions ensuring that an overall polynomial speed-up is achieved.

\begin{theorem}
\label{th:metatheorem}
Suppose we are given the algorithms $\mathcal{A}$ and $\mathcal{A}_Q$, and a $\kappa$ (the relative size of the quantum computer) inducing a search tree decomposition $\{\mathcal{T}_j\}_{j=0}^{\jleaves}$ as described above, for a problem family, such that:
\begin{enumerate}
\item\label{it:meta1} \textit{(subtrees are big on average)} The sizes of the induced sub-trees $\{{\mathcal T}_j\}_j$ are on average exponential in size,
so $\sum_{j=1}^{\jleaves}{T_j /\jleaves} \in \Theta^\ast(2^{\lambda n}) $ for some $\lambda>0$.
%so $\sum_j{T_j} /J \in \Theta(2^{\lambda n}) $ for some $\lambda>0$.
In particular, the overall tree is also exponential in size.
\item\label{it:meta2} \textit{(quantum algorithm is faster)}
The query complexity of $\mathcal{A}_Q$ on exponentially large subtrees of size $\Theta^\ast(2^{\gamma n})$ is polynomially better than $\mathcal A$'s, i.e.~$\Theta^\ast(2^{\gamma(1-\delta) n}),$ for some $\delta>0$. 
For convenience, 
we assume that in the quantum case the query complexity is given by some increasing concave function $\Phi(T')\leq T'$ (for a tree $\mathcal T'$), such that $\Phi(T')$ is essentially $T'$ for small (sub-exponential) trees, and for bigger trees $\Phi(T')$ is essentially no larger than $(T')^{1-\delta}$.\footnote{To ensure this, we can imagine the algorithm switch from a classical to a quantum strategy only when the quantum strategy becomes faster.}

\item\label{it:meta3} (queries are efficient) The time complexities of the subroutines realizing one query in both  $\mathcal{A}$ and $\mathcal{A}_Q$ are polynomial.\footnote{In fact, sub-exponential would suffice for our definitions, but reasoning is easier with polynomial restrictions}
\end{enumerate}
Then the hybrid algorithm achieves a genuine (polynomial) speed-up. If the conditions above hold for all $\kappa,$ then the algorithm achieves a threshold-free genuine speed-up.
\end{theorem}

%We highlight one aspect of the theorem above; 
The assumption of the existence of  the function $\Phi(T)$ which meaningfully bounds the quantum query complexity in terms of the tree size, is in apparent contradiction with our previous explanations that, in the case of Grover-based search, such trivial connections cannot always be made. 

Here, we want to establish claims of polynomial improvements relative to the classical method. Since the query complexities of the classical method do depend on tree sizes alone, and since we assume that the quantum algorithm is polynomially faster (so the quantum query complexity is a power of the classical complexity), this implies that we can only consider cases where indeed the quantum query complexity \emph{can} be meaningfully upper bounded by some function of the tree size. Consequently our theorem will only apply to the Grover case when the trees are sufficiently large (of size $2^{\nicefrac n2}$ and higher), where such a non-trivial bound can be established. 

The intuition behind the theorem statement is as follows. Considering exponentially large search trees together with the limitation on the run-times in Item 2 of subroutines to be efficient ensures that we can safely consider only query complexities to establish polynomial separations. 

Secondly, Item~1 also ensures that the quantum algorithm will need to process sufficiently large problems to achieve a speed-up. This prohibits, e.g.~the counterexample we gave earlier with only one tree $\mathcal{T}_1$ of size $2^{\nicefrac n2}$ -- the top tree has height $\nicefrac n2$, meaning $\jleaves=2^{\nicefrac n2}$, and thus also the average quantity $\sum_j{T_j} / \jleaves$ is in this case exactly~1. 

Finally, item 3 ensures that the quantum algorithm will be polynomially faster on the non-trivial subtrees.
We of course implicitly assume that both the classical and the quantum algorithm solve the same problem, so that their hybridization also solves the same problem.\\

\begin{proof}%[Proof of Theorem~\ref{th:metatheorem}]
Under the conditions of the theorem it will suffice to show that the classical query complexities given by $T = T_0 + \sum_{j} T_j$ and the hybrid query complexity $T_H$, which we determine shortly, are polynomially related.

We make use of the following Lemma:

\begin{lemma}
\label{lemma:tree-size-decomposition}
For any binary tree $\mathcal{T}$ of size $T$, the number of leaves $K$ is bounded as follows: 
\[
(T/n+1)/2 \leq K \leq (T+1)/2.
\]
\end{lemma}

\begin{proof}
Consider $T'$ to be the size of the binary tree $\mathcal{T}'$ obtained by replacing every sequence of one-child nodes by a single node in $\mathcal{T}$.  
Note that $\mathcal{T}'$ is now a full binary tree (with 0 or 2 children), yet the number of leaves of  $T'$ and $T$ are the same. 
By construction we see that 
$
T'\leq T \leq n \cdot T'
$.
Since full binary trees have $(T'+1)/2$ leaves, we have that $T$ has more or equal to $(T/n+1)/2$ leaves, and less than $(T+1)/2$ leaves.
\end{proof}

Let $\jleaves$ be the number of leaves of $\mathcal{T}_0$ (which includes roots of subtrees $\mathcal{T}_j$, as well as actual leaves of $\mathcal{T}$). Let $\Avg_c = \sum_j{T_j}/ \jleaves$ denote the average tree size, and by assumption $\Avg_c = \Theta(2^{\lambda n}) $, so that 
\[
T = T_0 + \jleaves \times \Avg_c
\]
Since $\jleaves$ equals the number of leaves of $\mathcal{T}_0$, by assumptions and Lemma~\ref{lemma:tree-size-decomposition} we have that 
\EQ{\hspace{-1.4em}
T_0 +  \Avg_c \frac{T_0/n+1}2 \leq T  \leq T_0 +  \Avg_c \frac{T_0+1}{2}
}
%in other words
%\EQ{
%T_0(1+\Avg_c (\nicefrac1{2n} + \nicefrac1{2T_{0}})) \leq T\\
%T \leq T_0 (1 +\Avg_c (1+\nicefrac1{T_0})),
%}
which can be simplified to, 
\[
{T_0\Avg_c}/{(2n)} \leq T 
\text{ and }
T \leq T_0 (1 + 2\Avg_c ).
\]
On the other hand, by similar reasoning, the hybrid query complexity $T_H$ is upper bounded by 
\EQ{
 T_H \leq T_0 (1+2\Avg_q), \label{eq:th-upper}} where $\Avg_q$ is given with $\Avg_q = \sum_j{\Phi(T_j)}/\jleaves$ where $\Phi(T_j)$ is the quantum query complexity on the $j^{th}$ sub-tree.

By considering the ratio of the lower bound on the classical runtime $T_0 \cdot \Avg_c/(2n) \leq T$, and an even weaker upper bound on the hybrid complexity 
\[
T_H \leq T_0 (2+\epsilon)\Avg_q,
\] 
\[
T/T_{H} > \Avg_c/\Avg_q (1/(2n(2+\epsilon))),
\] 
we see that polynomial improvements are guaranteed whenever $\Avg_c$ and $\Avg_q$ are polynomially related (since they are both exponentially sized by assumption, the prefactors linear in $n$ can be neglected). 

Since we have that $\Avg_c = \sum_j{T_j}/\jleaves$ and $\Avg_q = \sum_j{\Phi(T_j)}/\jleaves$, where $\Phi(x)$ is concave and increasing and $\Phi(T_j) \leq T_j$, to lower bound the speed-up characterized by $\Avg_c/\Avg_q$, we need to minimize this quantity, which is equivalent to maximizing
$f({T_j}) =  \sum_j{\Phi(T_j)}$, subject to constraints $\sum_j T_j = c,$ for some constant $c$, and $T_j\geq 0$.

By concavity of $\Phi$, this maximum is obtained when $T_i =T_j = \Avg_c$ for all $i,j$, hence in the worst case we have  
$\Avg_q = \sum_j \Phi(\Avg_c)/\jleaves = \Phi(\Avg_c)$. By assumptions, this means $\Avg_q \in \Theta(2^{\lambda(1-\delta) n})$ whereas $\Avg_c \in \Theta(2^{\lambda n})$, which is an overall polynomial separation, as stated.
This polynomial separation in the query complexity will also be observed in the final separation of the classical and hybrid runtimes because Condition 3 ensures that queries take only polynomial time.

Finally, the threshold free speedup is obtained if the above separation is obtained for any $\kappa$,
which means that the tree decomposition induced by $\kappa$ (and the space complexity of $\mathcal A_q$) yields polynomial separation between classical and hybrid runtimes.
\end{proof}

In the discussion above we only touch upon the sizes of search trees, without considering the possibility that the classical (online) algorithm may get lucky. Indeed, it may encounter a solution much earlier in the search tree, whereas the quantum algorithm always explores all the trees. In essence, the above considerations are for the worst cases for the classical algorithms (e.g.~when no solutions exist).
However, this is easily generalized.

In Section~\ref{sec:quantum}, we discussed a method to circumvent this issue~\cite{ambainis-kokainis}, by enabling the quantum algorithm to only explore the effective tree, which the classical algorithm would explore as well, before termination. By utilizing these algorithms, all the results above remain valid in all cases, except the concepts of search trees must be substituted with the concepts of ``effective search trees,'' which are the search trees the classical algorithm would actually traverse before hitting a solution.

Additionally, we observe that our framework can be generalized to allow for $p$-ary trees instead of binary trees.

The presented theorem has the advantage of being very general, but the major drawback is that it is non-quantitative and difficult to verify.
This can be improved upon in special cases.

\subsection{Space complexity, effective sizes, and tree decomposition}
\label{sect:spec-intro}
Three aspects complicate the task of determining the runtime of our hybrid algorithms: tree structure, time-complexity of the quantum algorithm and the space complexity of the quantum algorithm.
All three issues are individually non-trivial. The tree structure can be very difficult to characterize (as is the case for the DPLL algorithm, see Section~\ref{sec:discussion}), and the space and time complexities can depend on different features of the sub-instance. Yet, their interplay is what determines the overall algorithm.

%As discussed in the previous section, the tree-size decomposition, together with the speed-ups the quantum algorithm achieves on the sub-trees (time complexity) determine the overall performance. In turn the tree-size decomposition depends on the tree structure, and on the space complexity of the algorithm, which determines the cut-off points. And finally, the tree structure also influences the shape of the sub-trees, looping back to the issue of quantum speed-ups on the sub-trees. 

We address these three issues separately, focusing on settings where the situation can be made simpler.
First and foremost, we will almost always assume that we deal with exponentially large trees and query complexities, and in all the cases we will consider, the runtimes involving a single query will be polynomial (Condition~\ref{it:meta3} of Theorem~\ref{th:metatheorem}). For this reason, we can ignore polynomial factors, which implies that the query complexities and overall run-times are equated.

Next, since we deal with tree search algorithms, we focus on quantum algorithms $\mathcal{A}_Q$, obtained by either performing Grover-based search, or quantum backtracking over the trees. For concreteness, we will imagine the search space to be one of partial assignments to a Boolean formula (so strings in $\{0,1, \ast\}$, with $\ast$ denoting an unspecified value), although all our considerations easily generalize.
We refer to the \concept{natural representation},
which is the vertex representation where to each vertex we associate the entire partial assignment (i.e.~$n$ symbols in $\{0,1,\ast\}$, so $\log_2(3) n$ bits). 

In this setting, the possible speed-ups and the space complexity can be more precisely characterized in the terms of the tree structure. In what follows, we focus on a concrete subtree representing a subproblem associated with the (restricted) formula $F$ that $\mathcal{A}_Q$ should solve.
This subtree $\mathcal{T}$ is of height $n$
and maximum branching number $br(\mathcal{T})$, comprising partial assignments of length~$n$.

\paragraph{Query complexity}
Recall from Section~\ref{sec:quantum}, that the query complexity of backtracking is essentially $\tilde{O}(\sqrt{T} n)$, and therefore depends strongly on the tree size and the tree height. So for large trees, we have an essentially quadratic improvement over the classical search algorithm. In the case of Grover's search algorithm, the query complexity is in $O(2^{\nicefrac{br}2})$, provided that the maximum branching number is $br(\mathcal T) \le n$, which is also a quadratic genuine speedup if the tree is full.\footnote{And, more generally a  polynomial speed up when the tree size is at least $\Omega^\ast(2^{br/2 + \delta}),$ for some $\delta>0$.}

We can already highlight the important feature that the natural instance size (number of variables) may be quite unrelated to the actual features of the sub-instance which dictate runtimes: tree size and branching number. 
For this reason, we introduce assumptions which allow us to relate the \textit{variable number} with the tree size, as well as the \textit{branching number}.

This difference between natural and effective problem sizes is even more important in the case of space complexity, where speed-ups may be impossible if the wrong measure is considered.

\paragraph{Space complexity} 
%Regarding space complexity the situation is more involved. Conceptually, however we can separate two sources of memory requirements. 
We can separate two sources of memory requirements.
The first is the specification of the search space, as for any methods of quantum search  we require a unique representation of every vertex in the tree.
Clearly, the {natural representation} of full partial assignments suffices, but often we can work with the specification of only the choices at every branching point.

\begin{figure}[t!]
\tikzset{main/.style={draw,circle}} % internal bdd nodes
\tikzset{leaf/.style={draw,minimum width=1.2em,minimum height=1.2em}} % bdd leaf
\tikzset{lddnode/.style={draw,minimum width=1.2em,minimum height=1.2em}}
\tikzset{node distance=10mm}

\begin{tikzpicture}[node distance=3mm and 3mm, thick, edgelabel/.style={fill=white,rounded corners=4pt,inner sep=1.5pt}]

	\node[] (root) [] {};
	\node[main] (r0) [node distance=3mm,right=of root,inner sep=1.5pt] {$x_1$};

	\begin{scope}[node distance=9mm and 9mm,inner sep=1.5pt]

	\node[main,inner sep=1.5pt] (w1) [right=of r0] {$x_2$};
	\node[] (w11) [below=of w1] {};
	\node[main,inner sep=1.5pt] (w2) [right=of w1] {$x_3$};
	\node[] (w21) [below=of w2] {};
	\node[main,inner sep=1.5pt] (w3) [right=of w2] {$x_4$};
	\node[main,inner sep=1.5pt] (w4) [right=of w3] {$x_5$};
	\node[] (w41) [below=of w4] {};
	\node[] (w5) [right=of w4] {};

    \node[draw] (x0) [above = 1.3cm of r0] {\small \phantom{111}};
    \node[] (x1) [right = of x0] {\small \phantom{111}};
	\node[] (x11) [below=of x1] {};
    \node[draw] (x2) [right = of x1] {\small 1\phantom{11}};
    \node[draw] (x3) [right = of x2] {\small 11\phantom{1}};
	\node[] (x31) [below=of x3] {};
    \node[draw] (x4) [right = of x3] {\small 111};
	\node[] (x41) [below=of x4] {};

    \draw[->] (x0) to node[edgelabel,above] {$ch2$} (x2);
    \draw[->,dotted] (x0) to node[edgelabel,,below left] {$ch2$}  (x11);
    \draw[->] (x2) to node[edgelabel,above] {$ch2$} (x3);
    \draw[->,dotted] (x2) to node[edgelabel,,below left] {$ch2$}   (x31);
    \draw[->] (x3) to node[edgelabel,above] {$ch2$} (x4);
    \draw[->,dotted] (x3) to node[edgelabel,,below left] {$ch2$}   (x41);

%	\draw[->] (root) to (r0);
    \draw[->] (r0) to node[edgelabel,above] {1} (w1);
    \draw[->] (r0) to node[edgelabel,below left] {0} (w11);
	\draw[->,dashed] (w1) to node[edgelabel,above] {0} (w2);
    \draw[->] (w1) to node[edgelabel,below left] {0} (w21);
	\draw[->] (w2) to node[edgelabel,above] {1} (w3);
	\draw[->] (w3) to node[edgelabel,above] {1} (w4);
    \draw[->] (w3) to node[edgelabel,below left] {0} (w41);
	\draw[->,dashed] (w4) to node[edgelabel,above] {1} (w5);
	
	\end{scope}

	\begin{scope}[node distance=6mm and 7mm,inner sep=1.5pt]

    \node[draw] (z0) [above = 1.4cm of x0] {\small \phantom{1011}};
    \node[draw] (z1) [right = of z0] {\small 1\phantom{011}};
	\node[] (z11) [below=of z1] {};
    \node[draw] (z2) [right = of z1] {\small 10\phantom{11}};
	\node[] (z21) [below=of z2] {};
    \node[draw] (z3) [right = of z2] {\small 101\phantom{1}};
	\node[] (z31) [below=of z3] {};
    \node[draw] (z4) [right = of z3] {\small 1011};
	\node[] (z41) [below=of z4] {};

    \draw[->] (z0) to node[edgelabel,above] {$ch2$} (z1);
    \draw[->,dotted] (z0) to node[edgelabel,,below left] {$ch2$}  (z11);
    \draw[->] (z1) to node[edgelabel,above] {$ch1$} (z2);
%    \draw[->,dotted] (z1) to node[edgelabel,,below left] {$ch2$}  (z21);
    \draw[->] (z2) to node[edgelabel,above] {$ch2$} (z3);
    \draw[->,dotted] (z2) to node[edgelabel,,below left] {$ch2$}  (z31);
    \draw[->] (z3) to node[edgelabel,above] {$ch2$} (z4);
    \draw[->,dotted] (z3) to node[edgelabel,,below left] {$ch2$}  (z41);

	\end{scope}

	\node[node distance=2mm,above=of z0] {$F_{\phantom{x_1}}$};
	\node[node distance=2mm,above=of z1] {$F_{x_1}$};
    \node[node distance=2mm,above=of z2] {$F_{x_1\overline{x_2}}$};
    \node[node distance=2mm,above=of z3] {$F_{x_1\overline{x_2}x_3}$};
	\node[node distance=2mm,above=of z4,xshift=2mm] {$F_{x_1\overline{x_2}x_3x_4}$};
	
\end{tikzpicture}
\caption{A path in a search tree in natural representation (above) and in branching representation  (middle) for
$F=(\overline{x_1}\lor \overline{x_2}) \land (\overline{x_3} \lor \overline{x_4} \lor {x_5})$ and the corresponding path in the decision tree of a DPLL algorithm execution (below), where dashed lines are forced and solid lines guessed assignments.
The node 111 in branching representation represents the partial assignment $x_1,x_2,x_3,\overline{x_4}$.}
\label{fig:tree}
\end{figure}

This more efficient representation, which associates to each vertex $v$ the unique branching choices on the path from the root to $v$, we refer to as the \emph{branching representation} (see Fig.~\ref{fig:tree}). The branching representation requires no more than  $br(\mathcal{T}) \leq n$ trits or, more efficiently, $br(\mathcal{T}) + \log(n) \leq n$ bits. Technically, we need $\log(n)$ bits to fix at the depth at which our path terminates to specify a vertex (after the last branching choice there can still be a path of non-trivial length remaining to our target vertex).\footnote{Our focus is not on vertices per se, but on uniquely specified paths from vertex to root, along the path of which we look for contradictions and satisfying assignments, so that the $\log(n)$ specification is not needed.}
This is still larger than the information-theoretic limit $\log_2(T),$ which is achieved by some enumeration of the vertices; however this representation is difficult to work with locally, and difficult to manipulate space-efficiently.
This memory requirement is the only one which we cannot circumvent for obvious fundamental reasons.

%For concreteness, we will also assume access to functions specifying the tree structure as defined in Section~\ref{sec:quantum}: the function $ch2(v,b)$, which takes on input a vertex $v$, and returns one of the children, specified by the bit $b$, a function $ch1(v)$ which returns the single child if $v$ has only one child.
We assume access to a function $chNo(v)$ which returns whether $v$ has one or two children, as it was introduced in Section~\ref{sec:quantum} along with the functions
$ch1(v)$ and $ch2(v,b)$ which return the children of forced and guessed nodes respectively. We assume that for each vertex we know the level it belongs to, so there is no need to check if a vertex is a leaf. 
Such functions take as input a vertex specified in the natural representation, as is the case in most algorithms. The construction of functions which take the branching representation on input will require additional work and space.

To understand the second source of memory consumption we need to consider in more detail what each search method entails.

We begin with Grover-based search.
In the natural representation, if the position of the leaves is unknown, we need to perform a brute force search over all possible partial assignments, leading to the query complexity of $\Omega(2^{\log_2(3) \nicefrac n2})$. This is prohibitively slow.\footnote{Furthermore in some cases it may also lead to invalid results, if there is no explicit mechanism to recognize legal vertices in the tree, and if the connectivity encodes properties of the problem.}

A more efficient approach relies on the branching choice representation. In this case, the search space is $2^{br}$,\footnote{Note, this does not enumerate all the vertices or leaves, but possible paths in trees with no more than $br$ branchings. This suffices to uniquely specify a leaf. There may be multiple specifications for a single leaf, if it occurs on a path with fewer branchings (in which case, the remaining choices are then simply ignored).} and all that is required  is a subroutine which checks whether a given sequence in this representation leads to a satisfying assignment.

In other words, we require a reversible implementation of the search predicate $P$ which is defined on branching choices. 

In the natural representation, a reversible version of $P$ for a node
$\vec x \in \{0,1,\ast\}^n$ can be implemented by a circuit computing the number of satisfied clauses in $F_{|\vec x}$, which requires the values of the actual variables. 
So, in the branching representation, to evaluate $P$, we must in some (implicit) way first reconstruct the values of the variables that occur in the partial assignment corresponding to the vertex fixed by the branching representation.
Such a modified $P$ which evaluates the branching representation can be run reversibly, as the Grover predicate to realize the search over all possible branch representation strings. We will refer to this string as the \emph{advice string}. Intuitively, it advises the search algorithm on its branching choices.

Now, to translate the branching representation to variable values, again intuitively, we must follow the path in the search tree, and for this, we will utilize the (reversible) implementations of the operations $ch1(v), ch2(v,b), chNo(v)$, to trace the path. This is easily done reversibly if we are allowed to store each partial assignment along the path. However, due to the size constraints imposed in the present work, realizing an efficient predicate is non-trivial task, which can nonetheless be achieved in specific cases (see Section~\ref{sec:concrete}).
%{We nonetheless later show that this is possible in some cases.}

In the case of backtracking-based algorithms, we also require sufficient space to represent each vertex uniquely. Aside from this, the space requirements stem from the implementation of the walk operator, and from the implementation of the quantum phase estimation subroutine, see Appendix~\ref{app:qpe}. The latter cost we prove can be done logarithmically in the natural representation size, and therefore can safely be ignored.\footnote{Note, a linear dependence would not immediately prohibit the application of a hybrid method, but it would cause a multiplicative decrease in the effective size of the instance we can handle, i.e.~our usable work space would effectively become a fraction of what it could be. }

We identified two subroutines as the bottleneck for a space-efficient realization of the quantum walk operator. The first is the construction of a unitary which takes a vertex specification on input (and an appropriate amount of ancillas initialized in a fiducial state), and produces (the same vertex, due to reversibility), and its children. The second subroutine is the same as in the Grover-based case: we need means of detecting whether a vertex satisfies the predicate $P$ in whichever representation it is given.

In principle, the subroutines which detect whether a vertex is satisfying in some representation may be very different than subroutines which generate children specifications. However, in every case we discuss in this paper, the detection of satisfying vertices given in the branching representation will be implemented by essentially sequentially going through the entire path, in an appropriate representation. 
For this reason, the methods that we present for Grover-based search in the branching representation can readily be used for backtracking-based search. 
We note that the sizes of the natural, and the branching representation constitute two main parameters, \emph{effective size measures},  associated with a problem instance, which determine the quantum space and time complexities. 

%Except for the near-trivial example in Section~\ref{sub:seth}, where search in the natural space is possible, and in Section~\ref{sub:ppsz} where we establish more involved trade-offs, we will provide particularly efficient schemes in terms of time complexity. These achieve \emph{linear space complexities} in the size of either the natural representation size or in terms of the branching representation sizes \emph{for all the routines discussed above}. Also they can be applied both in the Grover-based and backtracking cases.
In the present work, we design time-efficient schemes which achieve \emph{linear space complexities} in the size of either the natural representation size or in terms of the branching representation sizes \emph{for all the routines discussed above}.
Also they can be applied both in the Grover-based and backtracking cases.
Additional sub-linear space contributions are effectively negligible: since we assume quantum computers that are proportional the instance size, so $\kappa n$, so we can simply sacrifice any arbitrary sized fraction $\epsilon n$ of $\kappa n$ (so decreasing $\kappa$ by an arbitrarily small $\epsilon$) for all sub-linear space requirements.

With a full understanding of what parameters of the sub-instances influence space and query complexities of the quantum algorithms we consider, we can now focus on the overall tree decomposition, and identify settings where it all can be made to provide simple criteria for quantifiable speed-ups.

\subsection{Speed-up criteria: special cases}
\label{sec:special}
%To achieve more simple expressions in this section we will introduce a number of assumptions. 
In this section, we introduce several assumptions to simplify the expressions of the runtimes obtained through our hybrid approach.
One of the main assumptions is that the effective size of the quantum computer can be expressed as $\kappa' n'$, given a $\kappa n'$ sized device; in other words, that the space complexity of our algorithms is linear in $n'$, where $n'$ quantifies the original problem size measure. 
In the beginning of this section $n'$ will refer to the natural instance size (and thus the tree height), but later we will show how it can also quantify branching numbers.

Given a $\kappa' n'$ effective size quantum computer, the search tree decomposition will then produce a cut-off points at each vertex where the vertex effective size $n$ (and the branching number $br$) is below $\kappa' n'$.

To achieve polynomial speed-ups resulting sub-trees must be exponential (on average, relative to the number of leaves of the top tree $\mathcal{T}_0$), i.e.~in $\Theta^\ast(2^{\lambda n}),$ where speed-ups are achieved using backtracking for any $\lambda$, and using Grover if $\lambda > \nicefrac 12,$ or if $\lambda > br/2$
when the branching representation is used.

To make this more concrete we can consider a setting  where exponential subtree sizes are guaranteed and easily computable, i.e.~when the subtrees are uniform at a scale determined by $\kappa' n$. 

For the weakest case, where the space complexity of the quantum algorithm we consider is governed by $n$ (in the natural representation), we have the following setting.

We will say that a fully balanced\footnote{All paths from the root to a leaf are of length $n$.} tree $\mathcal{T}$ is \emph{uniformly dense with density larger than $\lambda$ at scale $\eta n$}, if all the sub-trees $\mathcal T'$ of height higher than $\eta n$ (i.e.~all trees with a root at a vertex $v$ of $\mathcal{T}'$ at a level higher than $n-\eta n$, which contain all the descendants of $v$ at distance no more than $\eta n$ from $v$) satisfy $T' \in \Theta(2^{\lambda \eta n})$.

In this case, if $\eta$ is matching the effective size of the  quantum computer ($\eta = \kappa'$), we get a polynomial speed-up for all densities $\lambda$ for backtracking, and whenever $\lambda>\nicefrac 12$ for Grover-based search, assuming access to a quantum computer with effective size $\kappa' n$ (with $\kappa'$ proportional to $\kappa$). The assumption that the trees are exponentially sized makes the proof of this statement trivial.

The definition of such strictly uniform trees can be further relaxed, by allowing that just a $1/\poly(n)$ fraction of the sub-trees beginning at level $n-\eta n$ is exponentially sized with the exponent $\lambda$, and by somewhat freeing the size of the top tree $\mathcal{T}_0$ (we essentially only need that this tree is not too large), and still obtain a polynomial improvement.
Again if $\eta = \kappa'$, we get polynomial speed-ups when $\kappa'$ is proportional to $\kappa$.

Since the effective size is $\kappa' n$, for uniform trees we obtain a search tree decomposition, where we cut at level (instance size given by height) $n - \kappa' n$. The obtained trees are of size height $\kappa ' n$, and by assumption, $1/\poly(n)$ of the trees are exponentially sized (i.e.~are in $\Theta^\ast(2^{\kappa' \lambda n})$).
But then the worst case average size is given with $2^{\lambda \kappa'n} /\poly(n)$ which is still exponential. 

With this we satisfy all the assumptions of the Theorem~\ref{th:metatheorem} and conclude polynomial improvements. But in this case we can be more precise about the achieved improvement.
From the proof of Theorem~\ref{th:metatheorem}, Eq. (\ref{eq:th-upper}) the hybrid query complexity can be upper bounded with 
$T_0 (1+2\Avg_q)$  where $\Avg_q$ is given with $\Avg_q = \sum_j \Phi(T_j)/J_\ell$.

Here $\Phi(T_j)$ is the quantum query complexity on the $j^{th}$ sub-tree. Note that  we assume that we can meaningfully bound the quantum query complexity as some function of the tree size.
 If we are using backtracking since the trees are exponentially large we have that $\Phi(T_j) = \Theta^\ast((T_j)^{\nicefrac 12})$.
 In the case of Grover's search, we will have meaningful statements of this type when the tree sizes are exponential in depth (so $T_j = \Theta(2^{\lambda n'})$), with exponent $\lambda>\nicefrac 12$.

By the same proof, $\Avg_q$ is maximized when $T_i = \sum_j T_j / J_\ell,$ for all $i$, in which case 
 \[
 \Avg_q  = \Theta(2^{\lambda \kappa'n/2} /\poly(n)),
 \] 
 which is dominated by $\Theta(2^{\lambda \kappa'n/2})$. 
\EQ{
 T_H%\textup{Time}(\mathcal{A}_H(root)) 
 =O^\ast(
T_0 (2^{ \nicefrac{\lambda\kappa'}2 n})\\ 
 =O^\ast(
2^{ (  \nicefrac{\kappa'}2  +  (1-\kappa')) \lambda n }
 ) = O^\ast(
2^{ ( 1 - \nicefrac{\kappa'}2 ) \lambda n } \label{Eq:runtime}
 )} 
 which is a polynomial improvement over the classical strategy which obtains
  \EQ{
  %\textup{Time}(\mathcal{A}(root))
  T_H =\Omega^\ast(
2^{  \lambda n }
 )} (recall, the asterisk denotes we omit polynomial terms).

As noted earlier, in the above analysis, if we use Grover's search, then 
\[
\Phi(T_j) = \min(2^{\lambda \kappa' n} ,2^{ \nicefrac{\kappa' n}2} ),
\]
in which case we only obtain a speed-up if $\lambda> \nicefrac 12$.

However, in Subsection~\ref{sub:ppsz} we describe a setting where,  although the tree is quite uniform, it is not sufficiently uniform relative to the natural measure -- the tree height. However, it is uniform relative to the branching number measure.
We can easily adapt the uniformity definition to this case.

We will say that a tree $\mathcal{T}$ is \emph{uniformly dense at branching level $\eta n$}, if all the sub-trees $\mathcal{T}'$ for which we have $br(\mathcal{T}') \geq \eta n$  (i.e. all trees with a root at a vertex $v$ of $\mathcal{T}$ which have at least $\eta n$ branches on any path) satisfy $T' \in \Theta(2^{\lambda \eta n})$.

Now, if we have access to algorithms whose space complexity is linear in the branching size measure, given a $\kappa' n$ effective quantum computer size -- now, relative to the branching size, meaning we can handle instances with $\kappa'n$ branchings --, if  the trees are uniformly dense at branching level $\eta n$, with $\eta n \leq \kappa' n$, we can run the quantum subroutines on exponentially sized subtrees and a  similar analysis holds.
 But this time, Grover's approach yields speed-ups whenever $\nicefrac{br}2 < \lambda$.
 
Although these results seem very similar, in general the latter setting allows us to start search much earlier as $br\leq n$, and indeed, it can be much smaller. Furthermore, in some cases, the distribution of branch cuts is not (guaranteed to be) uniform with respect to the natural size $n$, which prevents naive algorithms to be successful.

{In the next sections, we provide examples of both cases:
\begin{itemize}
 \item an example where the effective size $n$ plays a role,
 		in the context of $k$-SAT formulas for large $k$ and uniform trees, with speed-ups obtained via Grover-based search ($\lambda > \nicefrac 12$) and backtracking (Section~\ref{sub:seth});
 \item an example where the number of maximal branches $br$ plays a role, with speed-ups that depend only on the branching number measure (Section~\ref{sub:ppsz}).
\end{itemize}
}
We provide in Appendix~\ref{app:eppstein} an example where backtracking provably provides better hybrid performances than Grover 
%(due to variations on the density and shape the trees are dense, allowing speed-up, but not maximally dense, so backtracking is better) 
for the Hamiltonian cycle problem, building on prior work~\cite{dnq-eppstein}.

\section{Hybrid speed-ups for tree search in satisfiability problems}
\label{sec:concrete}
In this section, we provide examples of when various types of speed-ups are attainable using the hybrid tree-search-based framework we introduced in previous sections.

\subsection{Algorithm-independent improvements under the Strong Exponential Time Hypothesis}
\label{sub:seth}

The simplest and most ideal setting for hybrid divide and conquer approaches is when the search trees of any SAT algorithm are essentially everywhere maximally dense. The Strong Exponential Time Hypothesis (SETH) implies this.
Under this hypothesis, we can provide the simplest hybrid algorithm which still beats the best possible classical algorithm for $k$-SAT (where $k$ is large).

\newcommand{\NN}{\mathbb{N}}
\newcommand{\unit}{[0,1]}
Let us write $\gamma_k$ for the smallest $\gamma_k \in \unit$ such that there exists a $k$-SAT randomized algorithm of complexity $O(2^{\gamma_k n})$. 
SETH stipulates that the sequence of all the $\gamma_k$'s is increasing, and $\lim_{k \to \infty} \gamma_k = 1$. Then $\gamma_k$ can be made arbitrarily close to $1$. In other words, this also means the best possible classical algorithm is close to brute-force search, which is itself a divide-and-conquer algorithm, for large $k$.

We can apply the hybrid approach to brute-force search as classical algorithm, and Grover's search as quantum subroutine, with access to a $\kappa n$-qubits quantum computer.
As we show in the Appendix~\ref{app:efficient-grover}, we can implement Grover-based, brute-force search for SAT solving over $n$-variable formulas using $n+O(1)$ space, meaning that, asymptotically the effective size of the problem we can handle is $\kappa' n$ with $\kappa = \kappa'$.
 The result is a hybrid algorithm of time complexity $O^\ast(2^{(1-\kappa/2) n})$ for $k$-SAT for every $k \in \NN$.

But then by SETH, for every $\kappa>0$ there exists a $k$ s.t. $\gamma_k > 1-\kappa/2$, which implies a polynomial speed-up over any classical algorithm.
In summary we have the following result.
\begin{theorem}
Under the Strong Exponential Time Hypothesis, for every classical algorithm for $k$-SAT and for every $\kappa > 0$ (such that we are given access to a $\kappa n$-qubits quantum computer), there exists a $k$ such that we obtain a speed-up for the hybrid divide-and-conquer algorithm based on classical brute-force search and Grover's algorithm. In other words, the hybrid divide-and-conquer approach can offer an algorithm-independent speed-up under SETH.
\end{theorem}
To connect to the previous discussions on uniform trees, SETH guarantees that the trees of any tree-search based algorithm for $k$-SAT will, become arbitrarily dense (close to $\lambda=1$), at every constant scale $\kappa n$, for large enough $k$.

\subsection{Threshold-free speed-ups in PPSZ tree search for special formulas}
\label{sub:ppsz}

In this section we describe a setting in which, under some assumptions on the variable order, hybrid, threshold free speed-ups are possible for PPSZ tree search, which, as we mentioned, is at the core of the best-known classical exact SAT solvers.

%To achieve this we have two requirements: first we are restricted to cases where the formula has bounded index width (a precise definition is given later), which can be checked efficiently for any formula. Second, we assume that running PPSZ using the initially given variable order yields an expected number of guessed variables of at most $\gamma_k n$. This is usually guaranteed by picking a variable order uniformly at random, although in our case we cannot, and therefore introduce it as an assumption.
%We explain the details of these requirements shortly.

%To achieve this we require the assumption that we have a variable order for which the formula has bounded index width and for which the number of guessed variables during the PPSZ run is at most the expected number of guesses $\gamma_k$
%We previously discussed in Section~\ref{sec:background} two variants of the algorithm, based on the following resolutions: unit resolution (applicable when the formula $F$ underwent $s$-bounded resolution first), or $s$-implication, which makes pre-processing redundant, but the reduction rule more complicated as unit resolution is replaced by $s$-implication. 

%For the moment, we shall restrict our analysis to the $s$-implication variant, for reasons which will be clarified shortly. Nonetheless,  all the results presented in this section can be amended to work for unit resolution as well. A space efficient implementation of unit resolution is discussed in the Appendix~\ref{sub:implement-unit}.

\subsubsection{Characterizing the search trees}\label{secsec:biw}
Recall that, towards realizing hybrid tree search, we already introduced dncPPSZ,
a tree search version of the original Monte Carlo algorithm
(see Algorithm~\ref{algorithm:dnc-ppsz} in Section~\ref{sec:background}).
It finds a satisfying assignment in a constant number of repetitions
if one exists (see \autoref{prop:ppsz-is-dnq}).
As input, the dncPPSZ subroutine takes an $k$-CNF formula $F$ and an order $\pi$,
then sequentially either resolves the next variable by $s$-implication,
yielding a forced tree node, or it branches on that given variable, yielding a guessed tree node. 
As detailed in the Appendix~\ref{proof:ppsz-is-dnq}, 
there exists a constant 
$\varepsilon > 0$
%$\gamma_k< 1$
such that, if satisfying assignments exist, then
$E_{\pi \in S_n}[ \text{dncPPSZ}(F, \pi) = 1 ]$ is non-nil and constant  (and thus a finite amount of repetitions of dncPPSZ suffices to find the satisfying assignment).
In other words, with constant probability dncPPSZ will find a path from the root to a satisfying assignment which contains no more than $b = (\gamma_k + \varepsilon) n$ branches
(guessed variables).
By pruning tree paths beyond $b$ branches,
we obtain a runtime of $O^\ast(2^{(\gamma_k + \varepsilon) n})$ and
a search tree with maximal branching number $b$ (see Def.~\ref{def:br}).
%We will call such paths \emph{good paths}.
%, and orderings which have good paths, \emph{good orderings}. 
%Since our dncPPSZ algorithm explores all branches for a sufficient number of variable orderings, a good path is guaranteed to occur with high probability, if the formula is satisfiable
%(see \autoref{prop:ppsz-is-dnq}).

%This property yields the advertised upper bounds on PPSZ algorithms, as a random choice of branches has a $\Omega(2^{-(\gamma_k + \varepsilon) n})$ probability of hitting a satisfying assignment, which by repetition ensures an $O^\ast(2^{(\gamma_k + \varepsilon) n})$ Monte Carlo algorithm, as we already demonstrated with dncPPSZ (\autoref{algorithm:dnc-ppsz}).

%The limit of $(\gamma_k + \varepsilon) n$  implies a number of properties relevant for our purposes; first and foremost, it implies that we may assume that the search tree is of the size $2^{(\gamma_k + \varepsilon) n}$. Note that the branching number does not prohibit much smaller trees (it does larger), but if better bounds could be proved this would immediately imply a backtracking algorithm for $k$-SAT which beats the PPSZ Monte Carlo approach.

Since, our objective is to provide a hybrid algorithm which achieves a better provable upper bound, threshold-free, the first obstacle is obtaining a sufficiently dense tree shape 
(see Section~\ref{sec:criteria}). However, since no such algorithm is known,
we indeed may assume the trees are of size $\Theta^\ast(2^{(\gamma_k + \varepsilon) n})$ (at least in the worst cases).
This all suggests that (non-hybrid) speedups are obtainable with quantum backtracking,
        and with a Grover-based method as well;
recall that Grover achieves the same upper-bound performance as backtracking,
if the bound on the tree size is given by the exponential of maximal branching number, as discussed in Section~\ref{sec:gvb}.
Any results with Grover characterize what we can expect from a quantum tree search algorithm.
% (at least as good as with Grover, but with the added benefit of early termination ).

Next, in hybrid setting we need to keep track of space complexity as well,
and the depth $\kappa'$ at which to invoke the quantum subroutine.
Since PPSZ is a simple tree search algorithm over partial assignments,
the natural instance size is the number of variables which have not yet been set. 
It is relatively easy to construct quantum Grover-based and backtracking algorithms
which achieve a linear space complexity in this quantity.
In this case,  the cut-off point in the search tree decomposition happens at vertices corresponding to some number of variables not yet resolved.
However this will not suffice for threshold-free improvements. 

The problem is the following: while the number of branches is $(\gamma_k + \varepsilon) n$, along a path from the root to an assignment, there is no guarantee on \emph{where} along the path they occur. Indeed, since the formula simplifies as variables are set, it is more likely branches occur earlier on, and it is possible all branches are used up first, leaving a formula with $n-(\gamma_k +\varepsilon)n$ variables, which is trivial from this point.
So to achieve speed-ups in this scenario, our quantum device must be able to handle instances of size at least $\kappa' n > (1-\gamma_k- \varepsilon) n$,
hence the approach is not threshold-free.

%There are two obvious approaches one may attempt to circumvent this issue and achieve threshold-free speed-ups. First, one could prove that, for a relevant family of formulas, the branches remain sufficiently dense along the entire path; specifically, it would suffice to show that there exist good paths, for every $\kappa'$, the last $\kappa' n $ steps in the path before the satisfying assignment is hit, contain at least some $\gamma_{\kappa'} \in O(1)$ branches.\footnote{For instance, this would hold true if the sub-instances, corresponding to the $\kappa' n$-variable restricted formulas, were shown to still contain the hard (smaller) instances for PPSZ.} This would imply that the corresponding sub-trees are still exponential (for every $\kappa'$), which suffices for a speed-up using Grover-based methods or backtracking with a straightforward method.

%The second method, which we employ here, is to
Our solution is to construct algorithms that work in the branching representation,
discussed in Section~\ref{sec:gvb}, and achieve
linear space-efficiency in the remaining number of branches.
Since branches directly dictate the (exponential) tree size, starting the algorithm at a point with some-fraction-of-$n$ branches remaining will guarantee all the properties we highlighted in Theorem~\ref{th:metatheorem}, and more specifically, render the PPSZ case a uniform tree case relative to the branching cuts, as described in Section \ref{sec:special}.
However, coming up with reversible implementations of PPSZ traversal, whose space efficiency depends on the branching numbers (and time-efficiency is sub-exponential) is more complicated.

At this point let us be more precise. We are looking for reversible algorithms (circuit) which compute the children for any vertex of the PPSZ search tree in the branching representation, and which can decide whether a vertex specified in the branching representation is satisfying or a contradiction. That is, implementations of the functions $ch1,ch2, chNo, P$,
as this suffices to implement the walk operator (see  Section~\ref{sec:quantum}).

What complicates the realization of such subroutines is, as we discussed previously in Section~\ref{sect:spec-intro},  branching choices, i.e. the values of guessed variables alone, do not map trivially to individual variable values. For example, in the branching representation, a node's label  $111$, i.e.~the first three guessed variables have all been set to $1$, should first be converted into a partial assignment to the variables, before we can compute its children.\footnote{For instance, in the formula $(\overline{x_1} \lor x_2) \land
(\overline{x_2} \lor \overline{x_3} \lor \overline{x_4} \lor {x_5})$
(and order $x_1 < \dots < x_5$), the node with branching representation $111$
represents the assignment $x_1, \overline{x_2}, x_3 , x_4, \overline{x_5}$, because
$x_2$ and $x_5$ will be ($s>1$)-implied.}
To know which variable is the first guessed variable, we must first compute the $s$-implications. An obvious solution is to compute all implications and store the corresponding partial assignment, but this violates the objective of using less memory than what the natural representation allows.
Since, as mentioned, Grover-based search will already achieve speed-ups over a classical strategy, we tackle the above problem in this context.

First, we simplify the problem for the PPSZ setting
by using eager evaluation of forced variables.
This compresses all the line paths in the tree,
resulting in a tree with only binary nodes.
Algorithm~\ref{algorithm:dnc-ppsz2} illustrates this using the tree node
functions. The node $v$ consists of an advice label, e.g., $111$
(including a counter indicating its length).
Implementing $ch2$ is trivial, e.g., $ch2(111, b) = 111b$.
However, this complicates the search function $P$
since it now has to eagerly evaluate forced variables,
but it anyway needs to `decode' the advice (which is a similar process).

\begin{figure}
\hspace{-1em}
\begin{algorithm}[H]
  \caption{dncPPSZ($v$) using our functions, where
                $\sizeof v$ is the length of node~$v$'s advice label.}
  \label{algorithm:dnc-ppsz2}
  \hspace{-2em}
   \begin{algorithmic}\vspace{-.3em}
%   \State $\pi$ permutation in $S_n$, $s \in \N$
   \IfThen{$P(v) = 1$}{\Return{1}}       %\Comment{$\models F$}
   \IfThen{$P(v) = 0$ \textbf{or} $\sizeof{v} > (\gamma_k+\varepsilon) n $}{\Return{0}}   %\Comment{$F\models 0$ or $d > (\gamma_k +\varepsilon) n$}
%   \State $x,~ \pi \gets \pi[0],~ \pi[1\dots] $ \Comment{first var in $\pi$, postfix}  % first non-assigned variable in $\pi$
   \State \Return{dncPPSZ($ch2(v,0)$)$\vee$dncPPSZ($ch2(v,1)$)}
   \end{algorithmic}
\end{algorithm}
\end{figure}

\defmath\sadv{S_{\mathit{adv}}}

Next, we specify the problem that $P$ has to solve.
We call it \emph{s-implication with advice} (SIA),
which is intuitively and implicitly defined as follows (Appendix \ref{app:sia} provides a precise definition). Given a fixed formula $F$ and variable order $\pi$,
an algorithm solving SIA takes on input an advice string of size \sadv, specifying which branching choices will be made at the guessed variables. The algorithm should compute the path realized in an $s$-implication-resolution based process in the tree specified by $F$ by
going over the variables in the specified order, setting them to either
the forced value if $s$-implied, or to the next unused branching choice
specified in the advice register if not. It should then output:
%: 
$\ket{00}$, if more branches are encountered than the advice length;
$\ket{1b}$, if the path makes the formula (un)satisfied at some point,
such that $F = b$. 

The desired space-efficient implementation of SIA will utilize a \sadv-sized advice register for the choices, and ideally no-more than $o(n)$ ancillas (although, low-prefactor linear scaling is also acceptable as explained in Section~\ref{sect:spec-intro}). While these criteria are easily met in an irreversible computation, critically, it must be realized reversibly under these conditions, as in this case we achieve Grover-based search, by searching over the advice register.

Obtaining $o(n)$ space is not trivial for the problems we consider. In particular, it can be proven that the routines required to implement PPSZ (e.g.~repeated $s$-implication, unit resolution) are 
\P-complete under log-space reductions (as proven for unit resolution in~\cite{jones-laaser}).
%This can be shown based on the $s$-implication and unit resolutions involved in the various versions of PPSZ (and thus also in the implementations of the quantum random walk operators). 
It is still an open problem whether $o(n)$ space, poly-time algorithms exists for \P-complete problems~\cite[Section 5.4]{greenlaw1995limits}, but lower bounds on pebbling approaches~\cite{lengauer1982asymptotically} give a pessimistic outlook.
Therefore, an approach that focuses on identifying classes of formula which admit genuine hybrid speedups seems justified.

Accordingly, we focus on formulas where, for a given variable order, the formula has \emph{bounded index width} (biw) as defined in Definition~\ref{def:biw}.
In Appendix~\ref{app:sia}, we provide an $o(n)$-space reversible algorithm for SIA,
which works for this class of formulas.
Specifically, we provide an algorithm which has polynomial runtime,
and space complexity $S_{adv} + S_{w}$, where $S_{adv}$ is the advice-string register, and 
\EQ{
%S_{w} =  \log_2(n/w) w + O(\polylog(n)),
S_w = O(w \cdot \log(\nicefrac nw))
\label{eq:Sw}
}
for any $k$-CNF formula of biw $w$.
{We do so by providing a frugal, reversible implementation of SIA, and then by splitting this (deterministic) computation into blocks that require only
a partial assignment of $w$ variables,
applying Bennett's reversible pebbling strategy on those blocks~\cite{bennett}.

\begin{definition}\label{def:biw}
For a given Boolean formula $F$ and a variable ordering $\vec x$, we defined the index-width $\iw(F)$ of a formula to be the largest difference between two indices of variables in a clause of the formula $F$, 
i.e.
\[
\iw(F) = \max_{C \in F} \max_{x_i,x_j \in C} |i-j|.
\]
A formula $F$ has bounded index width (biw) $w$ if $\iw(F) \leq w$.
\end{definition}

Before giving the complexity theoretic analysis of the hybrid method obtained by using the SIA algorithm above,  we highlight a few facts.

First, the property of bounded index width is an order-dependent property. Randomizing the variable order breaks this, and consequently our result does not directly imply speed-ups for PPSZ-proper, just for the tree-search part.

Luckily, an inspection of the correctness proof of PPSZ in \cite{scheder21} reveals that not all variable orders need to be considered in the case of biw formulae.
This is because the biw property ensures that the critical clause tree, central to the proof, only reasons over variables up to distance $w \cdot \log(s)$ from the variable for which $s$-implication is computed. The critical clause tree can thus not distinguish between permutations with more than $w \cdot \log(s)$ displacement.
Permutations with no more than $k$ displacement can be efficiently obtained by $k$-sorting\footnote{A sequence $[a_1, \dots, a_n]$ is $k$-sorted iff $\forall i,j$ with $1 \leq i \leq j \leq n$, $i \leq j - k$, it holds that $a_i \leq a_j$.} a random permutation~\cite{ksorted}.
A $k$-sorted permutation of the variable order partially preserves the biw: if a formula $F$ has biw $w$ for the initial variable order, then a $k$-sorted permutation of the variable order yields a biw $w' \leq w + 2k$.

Second, bounded index width formulas have specialized SAT-solving algorithms with best run-time $O(2^w \poly(n))$, to our knowledge~\cite{index-width}. We will discuss the consequences shortly but for the moment we just focus on beating known PPSZ tree search bounds in these special cases.  
Next we continue with  complexity theoretic analysis of the hybrid algorithm with SIA.

Consider first the setting where the index width is sub-linear, $w \in o(n)$. In this case, the space complexity is sub-linear (barring the advice), which means that given a $\kappa n$-sized quantum computer, we can turn to the quantum strategy the moment $\kappa' n $ guesses remain with $\kappa' > \kappa + \epsilon$, for every $\epsilon > 0$ (we reserve $\epsilon n$ memory to hold the $o(n)$ ancillas).

From this point on, we instantiate the discussion regarding uniform trees with respect to branching cuts discussed in section \ref{sec:special}.

We can assume that the sub-trees are exponentially sized, of size $\Theta^\ast (2^{\kappa' n})$ (as this quantity is used as the upper bound on the time complexity of the classical algorithm, which we are trying to outperform), and we obtain a full quadratic sped-up  run-time of $O^\ast(2^{\kappa'n/2})$ on the same subtrees. 

The top part of the tree has size $O^\ast(2^{(\gamma_k-\kappa' ) n})$ (as tree sizes are dictated by branching choices), and by our previous analyses, this yields a hybrid run-time of $O^\ast(2^{(\gamma_k - \kappa'/2)n })$ (when $\kappa'<\gamma_k,$ otherwise we obtain a full quadratic speed-up). All in all, in terms of $\kappa$, we obtain $O^\ast(2^{(\gamma_k - (\alpha \kappa-\epsilon)/2)}),$  were $\alpha = \kappa'/\kappa$ is the coefficient from the space efficiency of SIA, and $\epsilon$ is an arbitrary small constant used to handle $o(n)$-sized ancillary registers. 

We note that interesting examples of bounded index-width formulas (with connections to statistical physics) arise when one consider restrictions of the $3$-SAT problem.
One example of such problem is Lattice SAT ($3$-SAT on a lattice, with $\sqrt{n} \in o(n)$ bounded index width), which is formally defined and proved to be NP-complete in Appendix~\ref{app:latticesat}.

At this point, we return to the fact that for bounded index width formulas specialized SAT-solving algorithms exist with best run-time of $O(2^w \poly(n))$~\cite{index-width}. In the case that $w \in o(n)$, we can decide the very satisfiability of the given formula in subexponential time.

In other words, in the PPSZ process, it is much more efficient (in terms of upper bounds) to actually solve SAT the moment we encounter a formula with a sub-linear index width, than to continue with PPSZ search.
Switching from the tree-search process to a specialized solver is of course no longer PPSZ, but a new algorithm whose properties are uncharacterized (but not worse than PPSZ). Consequently, technically, our previous results still entail an improvement over the basic PPSZ tree search.
 However, it is of course interesting to see if settings can be identified where switching to a bounded-index-width specialized algorithm actually constitutes a bad choice, and our hybrid strategy is best in general. We provide this in the next section.

\subsubsection{When hybrid quantum PPSZ search improves over known classical algorithms}
\label{sub:hybrid-vs-classical}

To defeat bounded-index-width specialized algorithms, we consider index widths which are linear in $n$.
As shown in Appendix~\ref{app:sia}, it is possible to implement $s$-implication with advice (SIA) reversibly using total space $S_{adv} + S_{w}$, where $S_{w} =  O(w \cdot \log(\nicefrac nw))$ (the ancillary space needed for reversible SIA), and where $S_{adv}$ is the size of advice itself. For all $w$, the runtime of this subroutine is polynomial in $n$, and the runtime of the overall quantum algorithm, which solves the subtree by exploring the advice string space, is $O(2^{S_{adv}/2} \times \poly(n))$.

In contrast, a classical algorithm can decide SAT on formulas of index width $w$ in time ${O}(2^{w} \poly(n))$~\cite{index-width}. It is unlikely that one can do much better than this: in the case $w=n$, achieving a bound better than $O^\ast(2^{c n}),$ for some constant $c>0$, would violate the exponential time hypothesis (ETH). 

Here we assume this also holds for index widths which are fractions of $n$.\footnote{Furthermore, under the strong ETH, c approaches 1 for large clause sizes ($k$ in $k$-SAT).} It will be convenient to fix $w = \zeta   n$ (note $\zeta$ is a constant). 
In this case, we can identify a regime in which the quantum search over the advice string is still faster than solving SAT on the formula. This is the case whenever $S_{adv}/2$ is less than $w$, as these are the exponents of the exponential part of the run time of the respective algorithms. However, we must still ensure that the quantum algorithm can be run at all. So, a part of the overall memory available must be split between the advice string, and the memory required to solve SIA with advice.

We set $S_{adv}  = \beta \kappa n$ for $\beta \in (0,1)$, and the remainder of the memory of $(1-\beta)\kappa n$ qubits is spent as ancillary memory $S_w$ for the SIA subroutine. 
Note that when setting $\zeta = \nicefrac w n$ as a constant, we have that $S_w \notin o(n)$, as $w$ is then linear in $n$. However, in what follows we show that we can always find a pair $(\beta, \zeta)$ such that at least $S_w < n$.

%We introduce the buffer $\epsilon>0$ (which can be set arbitrarily small), to account for any polylogarithmic memory requirement provided $n$ is large enough, i.e.~for any $\zeta$, the $O(\polylog_2(n))$ work space can be fit in the $\epsilon\kappa n$-sized available register.
%Other memory requirements are also not more than polylog, so fit in the  $\epsilon$-buffer.
In summary we have that 
\EQ{
\log(1/\zeta)\zeta b \leq (1-\beta)\kappa \label{EQ:1}
}
since to process $w =  \zeta  n$-index width formula we need $\log(1/\zeta) \zeta b  n$ bits (for some constant $b$; Eq.~(\ref{eq:Sw})), and since we allocated $(1-\beta)\kappa n$ for this purpose.
Moreover, the exponents of the run-times of the quantum and classical algorithms are  $\beta \kappa n/2$ (Grover speed-up) and $\zeta c n$, respectively, so it must hold that 
\EQ{
\beta \kappa < 2 c \zeta. \label{EQ:2}
}

For our purposes, it suffices to show that for every $\kappa$,$b$,$c$, there exist a $(\beta,\zeta)$ pair for which both conditions hold.

First, we fix a $\beta$, to some (small) value $\beta'$. Next, find $\zeta$ satisfying Eq. (\ref{EQ:1}); note such a $\zeta$ exists as $\log(1/\zeta) \zeta$ decreases in $\zeta$, when $\zeta$ is small enough, converging to zero. Note, if Eq.~(\ref{EQ:1}) holds for $\beta =\beta'$ it also holds for any smaller $\beta$.
Then choose $\beta \leq \beta'$ such that Eq.~(\ref{EQ:2}) is satisfied, which can be done by choosing $\beta$ small enough.
Since these $(\beta,\zeta)$ exist for any $b$, the space overhead of $S_w$ can be made an arbitrarily small (constant) fraction of $n$.

This guarantees the existence of regions in the  space set by the advice size (controlled by $\beta$) and index width (controlled by $\zeta$) with polynomial speed-ups (for a given $\kappa$). However, there are no guarantees that the PPSZ process is guaranteed to generate formulas which will fall in this region, as discussed shortly.

\subsubsection{Grover vs quantum backtracking}
We note that, instead of running Grover's search over an advice string, the same algorithms can be utilized to perform quantum backtracking over the same trees, as we briefly announced previously.
The backtracking tree corresponds to restricted formulas as usual, where
nodes have two children (applying $s$-implication eagerly
to collapse forced nodes with one child).  

Determining children corresponds to the step $\text{SIAB}$ (see Appendix~\ref{app:siab}), and the evaluation of the leaves (final satisfiability) will actually involve running the entire scheme developed for the Grover-based approach.
The advantage here is that we will only explore the actual tree, but this does not provide a better theoretical speed up as the classical bounds assume a full tree.

However, the downside is that the walk operator utilizes two copies of the search space, and search-space-sized ancillary register, so that the available space is reduced by a factor of 4 (see the construction of the walk operator in Appendix~\ref{sub:implement-walk}). This nonetheless allows for regions of threshold-free polynomial speed-ups ($\kappa$ would be replaced with $\kappa/4$ in the above analysis). This would lead to smaller improvements in the upper bounds (i.e.~ignoring the speed-ups from exploring smaller trees), but may overall be more efficient in practice. We note that this pre-factor of 4 can probably be further improved.

\subsubsection{Assumption on the variable order}

We now discuss two clashing requirements from previous subsections.
First, to allow for a space-frugal reversible implementation of $s$-implication, in Subsection~\ref{secsec:biw} we consider formulae with bounded index width $w$, and permutations of these formulae with biw $w' = w + 2w \cdot \log(s)$.
Second, to obtain a speedup over classical algorithms for biw formulae, in Subsection~\ref{sub:hybrid-vs-classical} we require that $w$ grows linearly in $n$.
This gives rise to an issue with regard to the space requirement: with $w$ linear in $n$ and $s$ growing very slowly in $n$, we have that the ancillary space requirement (Eq.~\ref{eq:Sw}) for the permuted variable orders grows as $O(w \cdot \log(s))$, which grows faster than $O(n)$.

The permuted variable orders are used in the PPSZ proof to guarantee an upper bound of $\gamma_k n$ on the expected number of guessed variables during a PPSZ run~\cite{scheder21}.
Limiting the $k$-sorted permutations to some smaller $k$, e.g. $k = w$ would solve the space requirement, but fails to retain the guarantee on the expected number of guesses.

We therefore add the assumption that, for a formula $F$ with biw $w$, at least $1/\poly(n)$ of its $w$-sorted permutations are ``good'' variable orders. With a good variable order we mean a variable order for which the expected number of guessed variables during the PPSZ run is at most $\gamma_k n$.

\subsubsection{Putting it together}
Under the assumption made on the variable order in the previous subsection, we can put together the hybrid PPSZ algorithm as follows below. The result is summarized in \autoref{th:ppsz-speedups}. 

In a hybrid run of PPSZ-proper, which calls PPSZ tree search on order-randomized instances, the algorithm 
%based on the observations above 
would keep track of the current restricted formula. 
Then, it would switch to the quantum subroutine whenever the constraints for speed-ups, depending on the advice size (which is equal to the number of guesses remaining before $\gamma_k n$ guesses have been used up), and index width can be satisfied. Unfortunately, it is not the case that all the subtrees (for all initial formulas) will necessarily end up in the  regions satisfying the criteria, while the subtrees are still exponentially sized.\footnote{The formulas will become eventually small enough but if number of guesses decreases earlier than the index width, we may end up with trivial trees; in this case, interestingly, PPSZ will be faster than the upper bounds on index-width-specialized SAT solver.} However, in those cases, we are also not dealing with a hard instance, in the sense of the exponential time hypothesis.

When the trees are large enough, we will obtain a polynomial speed-up on that subtree, but this must occur for a constant fraction of the trees, to achieve an overall polynomial speed-up. This is true even if the initial formula is of bounded-index width, which is convenient as index width can only decrease, as index width, and number of guesses remaining can decrease at different rates.

%More problematically, PPSZ must randomize the ordering a few times, and generically, this ensures that the index width will be approaching $n$ as the number of clauses increases (if variables are assigned uniform at random from the set of possible indices, then the average distance of two random variables is $n/4$, so the index width will be at least this {with high probability}). 
%Thus even if the original formula was of a given suitable index width (depending on $\kappa$), the PPSZ process will shatter this with high probability. In other words, the above results cannot be used as stated to provide better theoretical bounds, although, in practice, we can easily detect when we do reach trees where the quantum method is better, and at least never perform worse than the classical algorithm.

%We remind the reader that in a setting in which we do not care about the initial formula being of bounded index width, it is more efficient to use $s$-resolution as pre-processing, followed by unit resolution rather than just $s$-implication.

%In general, the above examples show that for some classes of formulas, i.e., those those bounded index width, sufficiently efficient implementations of SIA are possible.
%\todo[inline]{move this? and refer to Appendix~\ref{app:sia} instead of ``above examples''?}

\begin{theorem}
\label{th:ppsz-speedups}
For a class of Boolean formulas $F$, the hybrid PPSZ algorithm achieves a polynomial speed-up if the following conditions are met:
\begin{enumerate}
\item $F$ is a hard instance (in the sense of the exponential time hypothosis).
\item $F$ has bounded index width $w$, and $w$ grows linearly in $n$.
\item Out of all $w$-sorted permutations of the variables of $F$, at least $1/\poly(n)$ have the property that the expected number of guessed variables during the PPSZ run (given that permutation) is at most $\gamma_k n$.
\end{enumerate}
\end{theorem}

Similar results may be obtainable for planar formulas. This is interesting as planarity does not depend on the variable order, so will not be obstructed by the
randomized order introduced by PPSZ.
However for planar formulas, we have sub-exponential algorithms with runtime $O({2^{\sqrt{n}}})$, so PPSZ is not the best choice to begin with.

%It remains an open question whether a space-and-time efficient algorithm for SIA exists for a family of formulas which is invariant under the PPSZ process (so that the restricted formulas are in the class for any variable order) and such that there is no obviously more efficient algorithm than PPSZ for any of the members of the family.

Finally, note that, outside of the context of our hybrid divide-and-conquer approach, our results imply that there is a quantum implementation of PPSZ via quantum backtracking which achieves a quadratic speed-up over classical implementations of PPSZ.
\begin{theorem}
The existence of a quantum implementation of SIA (Appendix~\ref{app:sia}) together with Algorithm~\ref{algorithm:dnc-ppsz2} implies a quantum backtracking implementation of PPSZ which achieves a quadratic speed-up over classical PPSZ.
\end{theorem}

\section{Discussion}
\label{sec:discussion}

The previous results constitute settings where we could obtain speed-ups for well characterized cases. In this discussion section, we consider the applicability of the hybrid method to the DPLL algorithm, and briefly discuss the consequences of polynomial time cut-offs, and alternative scalings, and finally, the limitations of hybrid methods.

\subsection{Potential for DPLL speed-ups}
\label{sub:potential}

In~\cite{qbacktracking}, Montanaro has demonstrated how quantum backtracking can be used to speed-up basic DPLL algorithms, which utilize just unit rule and pure literal rule resolution methods. This suffices for polynomial speed-ups whenever the search trees are exponentially large, and the satisfying solution does not exist, or appears late in the search order in the classical algorithm.\footnote{More precisely, we require that the effective explored trees, as discussed in Section~\ref{sec:hybrid} are exponentially sized.} 

However, to achieve improvements in hybrid settings, with a $\kappa n$-sized quantum computer, we have additional criteria on the structure of the tree and the algorithm as discussed in detail previously. This makes matters more complicated for DPLL. In general, there is less theory for DPLL we could apply to these questions in comparison to the PPSZ method. 
Nonetheless, we can at least resolve some of the technical concerns regarding the subroutines. 

In Appendix~\ref{app:backtracking}, we provide space-efficient implementations of quantum backtracking for pure literal rule and the unit rule, which suffice to obtain essentially linear space complexity with respect to the natural instance size (number of free variables).
In this case, any set of formulas which generate dense restricted formulas at depth $\kappa' n$ (where $\kappa' n$ is the effective quantum computer size) can be sped-up in a hybrid scheme. At present, we can only state that this is guaranteed (under SETH)  $k$-CNF formulas (for high $k$), as explained in Section~\ref{sub:seth}.
However, as characterizations of the lower bounds of DPLL trees improve, it is possible we obtain provable speed-ups for interesting $\kappa$ ratios, if not threshold free.

Using the same methods we provided for $s$-implication with advice for bounded-index width formulas, we can also provide equally space efficient algorithms for a
space-efficient 
%unit-resolution-with-advice for the same set of formulas.  The same can be generalized for the 
pure literal rule with advice. 

As stated earlier, we can implement all these algorithms in the quantum backtracking setting, achieving speed-ups in terms of the tree size, whereas branching number just determines the space complexity.
In this case, we could employ the quantum algorithm earlier, depending not on the natural instance size, but rather, on the number of branches like in PPSZ. However, the problem is that in the case of DPLL we have no meaningful upper bounds on the needed advice size. This can cause false negatives: if the quantum algorithm is utilized too early, we will run out of advice even if there exists a path to a satisfying assignment. 

Since running the search still constitutes an exponential effort (in the advice string size), we cannot simply run the algorithm at each vertex of the tree. Nonetheless, this approach may offer a path to viable heuristic for the classes of formulas where we have a reasonable upper bound on the number of branches, or additional information on the tree structure -- since DPLL is predominantly a heuristic method, such a result is fitting. 

We end our discussion on the hybrid method for DPLL by considering an additional constraint: real world considerations, specifically that all the run-times are polynomially bounded.

\subsubsection{DPLL in the poly-domain}
\label{sub:dpll-practice}

In practice, DPLL is often used as a heuristic, on formulas for which it can find a satisfying assignment with high probability in polynomial time.

In this subsection, we consider the possible consequences of obtaining small polynomial improvements over classical DPLL with polynomial cut-offs using a hybrid, or fully quantum method. The problem is: in the poly-world, poly-overheads, which we ignored in all previous considerations, matter.

For some $c > 0$, we define a polynomial cutoff for DPLL to be a polynomial size limit $n^c$ for the subtree explored by DPLL for an arbitrary formula. On such a subtree, classical DPLL will take $\mathcal{C} = \text{poly}_1(n) \cdot n^c$ time to terminate, while our hybrid DPLL will take $\mathcal{Q}= \text{poly}_2(n) \cdot n^{\alpha c}$, where $\alpha \in [\frac{1}{2},1)$ depends on the shape of the subtree defined by the first $n^c$ vertices explored by DPLL, and the size of the quantum computer that we are given access to, and $\text{poly}_1(n)$ and $\text{poly}_2(n)$ are  the run-times of individual subroutines in involved in one query of the  classical and quantum algorithm, respectively. 

Note that in this setting, the size $T'$ of the subtree explored by the algorithm is much smaller than the size of the whole search tree $T$, and therefore, we need to exploit a variant of quantum backtracking, whose inner working is explained in Appendix~\ref{sec:quantum-tree-size-estimation}.

The classical runtime is greater than the hybrid divide-and-conquer runtime whenever 
$\mathcal{C} > \mathcal{Q}$, i.e. 
\[
\frac{\text{poly}_2(n)}{\text{poly}_1(n)} < n^{(1-\alpha)c},
\]

so that our hybrid divide-and-conquer approach improves on the classical runtime whenever 
$c > \frac{\beta}{1-\alpha}$,
where $\beta$ is such that $\frac{\text{poly}_2(n)}{\text{poly}_1(n)}=n^\beta$.

It would be interesting to estimate how the ratio $\alpha \in [\frac{1}{2},1)$ evolves depending on $\kappa$, for various ordering of the vertices of the search tree (depth-first search and breadth-first search in particular).

A careful reader may notice that  numerous problematic assumptions have to be taken into account to achieve, arguably, very small improvements.
We point out that this is all a consequence of our setting chosen to enable us to provide \emph{clean statements about asymptotic speed-ups}. 
In particular, it is for these reasons that we assume that the quantum computer scales with some quantity which can be used to characterize run times and space efficiencies. This is the ``natural instance size''. 

However, one can easily switch to a less demanding model.
Let $c(n)$ be the space complexity of the fastest known quantum algorithm for the problem class under considerations.
If we assume that the quantum computer is of the size $\kappa \cdot c(n),$ which is still smaller than what we need to run the basic algorithm hence interesting, stronger results may be possible, with the expense being a less clean, and more conditional, analysis.

Note that, in the real world, a quantum computer will offer an advantage in real-compute time, whenever it is used in a hybrid setting and the quantum device achieves a real-compute-time speed-up (not as a scaling statement, but in the units of seconds), on the particular sub-instance. Real-world analyses of speed-ups, for fixed instance sizes must thus take into account real-world parameters.

Although in the present work, we are interested in the more theoretical questions, we believe that further improvements are still possible in this much more stringent, and more general model.  
This leads us to the next section, which discusses the general limitations on the speed-ups obtainable by the the hybrid method.

\subsection{Limitations of the hybrid approach and the
framework of networked smaller quantum devices}
\label{sub:limitations}

In the approaches we have explored, the best improvement one can obtain given access to a fractionally smaller quantum device is polynomial.
This is not just a consequence of the fact that we use quantum backtracking or Grover's search as the backbone of the speed-up (which themselves only allow a quadratic improvement), as one may think.
It is also a consequence of the hybrid divide-and-conquer setting; since the idea is to speed up the explorations of exponentially large trees by delegating subtrees of fractional height (say $\kappa' n$) to a quantum device, by construction, we still rely on the classical algorithm to explore the ``top'' (we denoted this tree $\mathcal{T}_0$ previously) of the tree. This will generically take exponential time in the fraction $(1-\kappa')n $ which is still exponential in $n$. 

In more detail, in the hybrid divide and conquer approach, the total run-time essentially attains the form  
$t_{hybrid} \in O^\ast(2^{\tau_h n})$, with $\tau_h = ((1-\kappa') \tau_c + \kappa'\tau_q )$ in the case of uniform trees where each (large enough) tree of height $n'$ is of size $2^{\tau_c n'}$.
We assume here that $\tau_q$ dictates the relative speed of the quantum algorithm (e.g. the quadratic speed-up of backtracking implies $\tau_q = \tau_c/2$.
Even if the quantum query complexity and runtime was exactly zero (or, exponentially faster than the classical method), what remains is $O^\ast(2^{\tau_c(1-\kappa') n} )). $ This constitutes a just polynomial speed up over $t_{classical} = O^\ast(2^{\tau_c n})$, which is the classical runtime on the entire tree.

 This we summarize in the following lemma given without further proof.
\begin{lemma}
The speed-up attainable by the hybrid divide and conquer approach with a $\kappa'$-effective size quantum computer is at best polynomial. The ``speed-up'' is subquadratic or quadratic at best, i.e.~it holds that $ t_{hybrid} (n) \in O((t_{classical}(n)))^{1-\alpha}$, where the degree of speed-up $\alpha$ is bounded $\alpha \leq \kappa'$ (if the quantum algorithm has polynomial run-time on the subtrees), and by $\alpha = \nicefrac 12$ (i.e.~quadratic improvement) otherwise (if a Grover-type speed-up).
\end{lemma}

The limitation of the above approach is that the quantum device gets used late in the game. Conceivably we can imagine settings where the quantum computer takes a more active role in the top of tree, or something similar. 
Indeed, beyond standard backtracking settings, better speed-ups are also possible, under mild, yet unavoidable assumptions.

In what follows we will assume that the evaluation of an arbitrary quantum circuit of size $\poly(n)$ on a classical computer takes exponential time. For concreteness we assume that solving BQP-complete promise problems, e.g. the problem of, for a given quantum circuit realizing some unitary $U$, determining whether the measurement outcome probability of one output qubit is $0$ (of the register in the state $U \ket{0}^{n}$) is either larger than $2/3$ or below $1/3$ (under the promise that it is one of the two), requires $\Omega^\ast(2^{\tau n})$ classical computing steps (ignoring polynomial terms).\footnote{Unless we assume that quantum computations cannot be simulated in polynomial time, no better than polynomial improvements can be proved.} This is a circuit output problem.

In this case it is trivial to construct pathological examples of computational problems where exponential speed-ups can be attained (given a quantum computer of size $\kappa' n$) by arbitrarily choosing what the natural notion of the instance size should be.
For instance, consider the circuit output problem, where the quantum circuit is special: no gates act on $(1-\kappa') n$ wires. In this case obviously a quantum computer of size  $\kappa' n$ allows for a polynomial time solution, whereas the classical computer, by assumption requires $\Omega^\ast(2^{\kappa' \tau n})$ steps, which is an exponential separation for any $\kappa'$.

While this example is obviously pathological, one can easily imagine a more complicated yet related computation, where the $n$ input bits are processed first by an involved classical computation which produces a specification of a quantum circuit on $\kappa'n$ wires, and the output of the overall computation is the output of that circuit.

This is an example of a broad spectrum of scenarios, where a (fewer-qubit) quantum computation is called as a subroutine of an over-arching classical computation.

One class of such computations are the hybrid approaches we investigate in this work. Another involves settings where classical computations are broken down, to distill the computationally hardest part, which is then delegated to a quantum machine, see e.g.~\cite{Benedetti_2018}.

It is also clear that, to obtain the best speed-ups, the quantum computer should be used at wherever possible in the computation, as is discussed in~\cite{2004.00026}.

In what follows we consider a broad framework, the main purpose of which is to connect our work to less-than-obviously related works in quantum computing and quantum machine learning, such as the ones we exemplified above.

\paragraph{QuNets} We wish to define a hybrid computational model which captures some of the facets of the limitations that quantum computers face in the near-term, namely size.
We imagine access to a $k$-qubit quantum device, and want to consider all computations that can be run, when such a device is controlled, and augmented by, a (large) classical computer. This classical computer can pre-and-post process data, in between possibly many calls to the quantum device. In our model we will describe everything sequentially, although it will be a natural question how this can be parallelized when many $k$-sized quantum machines, which can communicate only classically, are available.

\newcommand{\QuNet}{\text{QuNet}}
We name such a model a $\QuNet$, and with $\QuNet(n,k)$ (with other qualifiers, described shortly) we denote the set of functions such a hybrid computation system can realize, given $n$-bit (classical) inputs and a $k$-qubit device. The number of output bits is specified when needed. It is worth noting that related models have been introduced and studied under various names by other authors, both explicitly and implicitly.\footnote{For instance, the standard diagrammatic representation of circuits we can work with``double'', classical wires, which can be classically processed, constitutes one such model~\cite{nielsen-chuang,PhysRevLett.76.3228}.}
Specific to our setting, however, is the focus placed on the  limitations of the qubit numbers $k$, relative to the instance size $n$.

It makes sense to distinguish two types of $\QuNet$s: adaptive ($\text{a-QuNet}$) and non-adaptive ($\QuNet$).
There are many ways to formalize both models, here we provide one approach; we define the latter first.

\newcommand{\cirqF}{\text{cirqF}}
Let $\cirqF_k(\vec{c})$ denote the family of randomized functions, in which each randomized function takes a bitstring $\vec{c}$ as input and specifies the realization of a quantum circuit for some unitary $U$ over $k$ qubits. That is, the output of such a function is some $k$-sized bitstring $\vec{o}$, occuring  with probability $|\langle \vec{o} | U | 0^{k}\rangle |^2$, i.e.~the function outputs what the quantum circuit would output.

 $\QuNet(n,k)$ is then a (randomized) Boolean circuit, with $n$ input wires, an arbitrary number of ancilla wires, where a standard Boolean gate-set is augmented with the  
gate set $\text{cirqF}_k(\vec{c})$, which take $|\vec{c}|$ input classical wires,\footnote{There are many ways how a bitstring may specify a circuit, and how the circuit depth is encoded in the bitstring, but this is not relevant for us. All that matters is that some encoding exists.} and output $k$ wires.

Note a $\QuNet(n,k)$ captures the two previous examples where an exponential speed up between a fully classical model $(\QuNet(n,0))$ and the genuine hybrid $\QuNet(n,\kappa' n)$ is provable, assuming quantum computers are not efficiently simulatable.

In the above model, the quantum computation is used essentially as a ``black box''. But in principle, more interaction is possible, once partial measurements are allowed. 
Here the classical computation can request that some of the quantum wires be measured, and the rest of the circuit may depend on the outcome.

The $\text{a-QuNet}$ captures this additional freedom. It is easiest to characterize in a hybrid classical-quantum reversible circuit model common in quantum computing (where single wires are quantum, double classical)~\cite{nielsen-chuang}.
We consider a (quantum-like) circuit of $n$ classical input wires, $m$ other ancillary classical wires pre-set to zero, and a quantum register of $k$ quantum wires pre-set in the state $\ket{0}$. We allow three types of gates: fully classical reversible gates (e.g. Toffoli and $X$ -- negation -- will suffice); standard quantum gates, acting only on the $k$ qubit wires; and CQ gates and QC gates. 
The CQ gates are classically controlled quantum gates, i.e.~a quantum gate is applied depending on the state of a classical wire; QC gates are measurements: a number of quantum wires is measured in the computational basis, and the outcome is xored with the value of some target classical wires, matching in number. After measurement, we assume the state of that particular quantum wire is again re-set to $\ket{0}$. 

It is not difficult to see that  $\text{a-QuNet}$s contain $\QuNet$s:  the measurements are done on all wires, and where we note that any $\cirqF_k(\vec{c})$ gate can be implemented by using CQ gates. In terms of which functions they can realize, the two models are clearly equivalent; in fact, if complexity is not taken into account, the classical computation can simulate the entire quantum computation so indeed $\QuNet(n,0)$ is already universal.
However in terms of efficiency and in other scenarios these two models can differ.
For instance, $\text{a-QuNet}$ captures error correction and  fault-tolerant quantum computation protocols, and also the measurement-based quantum computing paradigm~\cite{PhysRevLett.86.5188}, where classical feedback from quantum measurements is extremely beneficial or assumed by construction. 

Further it makes sense to limit the sizes of classical and quantum computations in both models, which allows for a more fine-grained comparison. With $\QuNet_{x,y}(n,k)$ we denote the function family that can be realized using no more than $x$ gates (including the quantum functions), and where the quantum circuits used use no more than $y$ gates (note that in general $y \in O(|\vec{c}|)$). In the adaptive model we can simply count the classical gates vs the QC and CQ gates.
Instead of particular values $x$ and $y$ can denote function families, e.g.  $\poly$ or $\text{exp}$, so $x=\poly$ is a short hand for $x = O(\poly(n))$.\footnote{Note that with these limitations the number of classical ancillas we need to allow is also bounded by $x+y$.}

With complexity-theoretical considerations, the relationship between  $\text{a-QuNet}(n,k)$ and $\QuNet(n,k)$ is not entirely clear, in that, in general, classical adaptivity may reduce the number of quantum gates needed -- it is well known that any classical control can be raised to fully quantum (with no classical feed-back), but at the expense of more quantum wires. For instance in the case of an algorithm using QPE to some precision $\epsilon$, in many cases it is known that one can perform most computations using $1$ ancillary wire adaptively\footnote{All that is required is that the ancillary qubit is measured, and reset for the next step in QPE.} (instead of $\log(1/\epsilon)$ ancillas which can be measured at once). In turn it is not clear the same can be done non-adaptively, where at each step the entire register must be measured, 
without at least $\polylog(1/\epsilon)$ multiplicative additional computational costs (see~\cite{OBrien2019} for state-of-art approaches to single qubit QPE).

More generally, to our knowledge, not much is known about the costs of rewriting an adaptive circuit with partial measurements as a circuit with complete measurements without introducing ancillary qubits; but efficient methods for this could simplify the execution of quantum algorithms on small machines that are also limited in coherence times. This topic goes beyond the scope of the present paper. 

We finalize this section by highlighting the connections between our hybrid model and other related lines of investigation, in the context of $\QuNet$s.

One example is our hybrid divide and conquer method, where the corresponding $\QuNet(n,k)$ has a near-trivial repeating structure as all quantum computations are of the same type, tackling smaller problem instances generated by a classical pre-processing step (the `top of the tree').
In particular, previous results in the hybrid approach are examples of $\QuNet_{\text{exp},\text{exp}}(n,k)$ which solve various NP-hard problems exactly, faster than their classical counterparts.
In the context of quantum annealers, all schemes developed for the purpose of fitting a larger computation on a fixed sized device (e.g. see~\cite{AJAGEKAR2020106630}) fit in this paradigm, and they ostensibly `get more' out of the device, however mostly in a heuristic setting where little can be proved. These are examples of $\QuNet_{\poly,\poly}(n,k)$ which tackle various NP-hard and quantum chemistry problems heuristically.
A relatively recent paper that also focuses on getting the most out of a smaller device utilizes data reduction, see~\cite{2004.00026}.
In all these examples, the computational problem is from a classical domain, and the approach is to `quantize' subroutines.

In an opposite vein, in~\cite{peng,bravyi}, the authors present hybrid computations which compute the output of a large quantum circuit (given on input), calling a smaller quantum device. These are examples of a $\QuNet_{\text{exp},\poly}(n,k)$ which solves the problem of simulating quantum computations. In this case, only the number of calls to the quantum device, and the classical processing is exponential, whereas the quantum circuits are polynomially sized.
It is clear that constructing $\QuNet$s with small $k$ for some hard problem is appealing for near-term quantum devices. However, it is not clear what are the families of functions and the complexity classes which can be captured by this model. 

Similarly the relationship between QuNets and classical and quantum parallel complexity classes which also care about splitting of the computation on multiple units (however, not caring about the required space) also remains to be clarified.
Coming up with a structure and a $\QuNet$ algorithm for a target problem may be difficult.

In the domain of variational methods, specifically applied to machine learning this problem could be circumvented. Any such network where the quantum circuits are externally parametrized is a valid Ansatz for a parametrized approach to supervised learning, or to generative modeling.
In particular, such networks generalize neural networks. Individual neurons are replaced by a circuit, a subset of parameters of which depend on the input values, and the remainder is free. The free parameters play the role of tunable weights in an artificial neuron. 
Ways to construct meaningful $\QuNet$s for machine learning purposes of this type are a matter of ongoing research.
\ \\

\acknowledgements
VD thanks Tom O'Brien for discussions regarding adaptive and non-adaptive QuNets.
This research was supported by the Dutch Research Council (NWO/OCW), as part of the Quantum Software Consortium program (project number 024.003.037), 
%is part of the research program VENI with project number 639.021.649, which is (partly) financed by the Netherlands Organization for Scientific Research (NWO)
and has been sponsored by the European Union's Horizon 2020 research and innovation program under the NEASQC project, grant agreement No 951821.

%Bibstyle only for quantum
\bibliographystyle{quantum}

\bibliography{dnq-backtrack}

%\bibliography{dnq-backtrack}

\appendix

%Proof for dncPPSZ
\clearpage
\newpage 

\section{Proof of Proposition~\ref{prop:ppsz-is-dnq}}
\label{proof:ppsz-is-dnq}

In this section, we prove that dncPPSZ (see Algorithm~\ref{algorithm:dnc-ppsz}) is correct and has a runtime of $2^{(\gamma_k+\varepsilon)n}$, i.e.,
it equals the best-known SAT algorithm as recent analysis of
PPSZ confirms~\cite{scheder21}
(it surpassed Biased PPSZ~\cite{biased-ppsz} again).
We focus here on unique-$k$-SAT, in which the input formula has a unique satisfying assignment 
if satisfiable. (The results can be generalized to general $k$-SAT based on~\cite{scheder21}.)
Let $F$ be a $k$-SAT formula on $n$ variables, $\vec u$ be the unique satisfying assignment of $F$ (provided $F$ is satisfiable).

Proving the runtime of dncPPSZ is easy, as shown by \autoref{th:dncPPSZ-runtime}, which we mainly introduce to define the best-known value of $\gamma_k$.
%\footnote{While more optimal versions of PPSZ exist, e.g. Biased PPSZ~\cite{biased-ppsz}, the main aim of these improvements is to show that the upper bound on PPSZ is not tight. % with a $10^{-6}$ margin.
%We therefore focus on the canonical version of PPSZ.
%We therefore focus on the PPSZ version in~\cite{scheder21}. We see no reason why the optimizations of Biased PPSZ cannot be incorporated in the quantum backtracking version of the algorithm which we are proposing.
}

\begin{theorem}\label{th:dncPPSZ-runtime}
Algorithm~\ref{algorithm:dnc-ppsz} has a runtime of $2^{(\gamma_k+\varepsilon)n}$
with $\gamma_3 \approx 0.386229$~\cite{scheder21}.
\end{theorem}
\begin{proof}
The algorithm only explores a tree truncated to depth $(\gamma_k+\varepsilon)n$ 
for a given (fixed) permutation $\pi$.
\end{proof}

We now show that dncPPSZ is a correct Monte-Carlo algorithm, i.e., it finds the unique satisfying assignment with constant probability  if it exists
(hence, if repeated a constant number of times, it will decide unique satisfiability).
To this end, we first reiterate some results from the PPSZ algorithm studied in~\cite{scheder21}.

In contrast to the dncPPSZ algorithm, which tries all assignments with backtracking, recent versions of PPSZ~\cite{scheder21} sample assignments using an \concept{advice string} $\vec a$ and the \concept{Decode} function shown in Algorithm~\ref{algorithm:decode}.
According to a random variable order $\pi$, the function only uses a bit of the advice string $\vec a$ when the next variable is guessed, i.e., when the heuristic fails to establish its truth value (see the \texttt{else} branch).
It also does not terminate evaluation early, as in backtracking, i.e., it naively continues evaluation even when the formula is already (un)satisfied.
This suffices because the only point of the Decode function is to study the expected number of guesses under $s$-implication and to demonstrate how the number of guesses dictates the expectation of finding the satisfying assignment.

{
\setlength{\intextsep}{0em}
\begin{figure}
\hspace{-1em}
\begin{algorithm}[H]
  \caption{Decode($F$, $\vec a$, $\pi$) where the value $s$ of $s$-implication is a large constant or grows very slowly in $n$~\cite{scheder21}. The function returns false or the unique satisfying assignment, as a CNF formula, if the advice $\vec a$ and order $\pi$ `decode' to it.}
  \label{algorithm:decode}
   \begin{algorithmic}
   \State \textbf{if} $\pi = \vect{}$ \textbf{then return} $F = \emptyset$ \Comment{SAT?}
%   \State $\pi$ permutation in $S_n$, $s \in \N$
   \State $x \gets \pi[0] $ \Comment{first var in $\pi$}  % first non-assigned variable in $\pi$
   	\If{$F \models_s (x = b)$ \textbf{for} $b \in \set{0,1}$}\Comment{forced:}
%        \State \Return{Decode($F_{|x = c}$,  $\vec a$,   $ \pi[1\dots]$)} 
	      \State $c := b$
	\Else\Comment{guessed:}
	      \State $c := a[0]$
	      \State $\vec a := \vec a[1\dots]$ \Comment{postfix of array}
   \EndIf
    \State \Return{Decode($F_{|x = c}$,  $\vec a$, $ \pi[1\dots]$) $\land\,\, (x = c)$  } 
   \end{algorithmic}
\end{algorithm}
\end{figure}
}

%Let $\mathit{guesses}(\vec a, \pi)$ denote the number of guesses that Decode performs when
%Decode($F, \vec a, \pi) = \vec u$.
\autoref{th:decode} shows that Decode uses $\gamma_k n$ guesses to find the unique satisfying assignment of a $k$-SAT formula $F$. By Corollary~\ref{cor:decode}, it can be found with high probability with a sub-exponential number of (random) calls to Decode.
The PPSZ algorithm is therefore defined as the loop that calls Decode sufficiently many times.

\begin{theorem}[Adapted from \cite{scheder21} Th.4]
\label{th:decode}
Let formula $F$, advice $\vec a$ and order $\pi$ be such that Decode($F, \vec a, \pi) = \vec u$ 
 (Algorithm~\ref{algorithm:decode}).
Then the expected number of guesses in Decode
equals  $\gamma_k n$.
\end{theorem}

%\todo[inline]{advice $\neq$ sat assignment}

\begin{corollary}[\cite{scheder21} Ob.5]
\label{cor:decode}
The unique satisfying assignment is found with high probability with $\staro(2^{(\gamma_k+\varepsilon)n})$ calls to Decode with random assignment $\vec a$ and order $\pi$.
\end{corollary}

Instead of trying random advices, the dncPPSZ algorithm searches all advices of length $(\gamma_k {+} \varepsilon)n$.
Since it implements the Decode function along each search branch, merely adding  backtracking when a formula becomes (un)satisfiable on a partial assignment, it also inherits the probability of finding the unique satisfying assignment from  PPSZ.
\autoref{th:ppsz} shows this. (It implies \autoref{prop:ppsz-is-dnq}.)

\begin{theorem}
\label{th:ppsz}
Let $F$ be a $k$-SAT formula defined on $n$ variables and $\pi$ a permutation in $S_n$.
Then the probability that dncPPSZ$_\pi$ returns a satisfying assignment is at least 
$1 - \frac1{1+\nicefrac{\varepsilon}{\gamma_k}}$.
\end{theorem}
\begin{proof}
Let $\mathit{g}(\pi)$ denote the number of guesses Decode performs for $\vec a$ with  
Decode($F, \vec a, \pi) = \vec u$.\\  By applying Markov's inequality, we obtain:

\noindent
\begin{align*}
\text{Pr}_\pi\left[\mathit{g}(\pi) > (\gamma_k+\varepsilon)n \right] &\leq
\text{Pr}_\pi\left[\mathit{g}(\pi) \geq (\gamma_k+\varepsilon)n \right] \\&\leq
\frac{\mathbb{E}_\pi[\mathit{g}(\pi)]}{(\gamma_k+\varepsilon)n }
    =\frac{\gamma_k}{\gamma_k + \varepsilon},
\end{align*}
which is an upper bound on the probability that dncPPSZ$_{\pi}$ fails to find a satisfying assignment, given it exists.
Subsequently, $1 - \frac{\gamma_k}{\gamma_k + \varepsilon}$ is then a lower bound on the probability that dncPPSZ$_{\pi}$ succeeds.
\end{proof}

\section{Space-efficient quantum subroutines for tree search algorithms}
\label{app:efficient-qalgs}

In this section, we implement quantum subroutines for tree search algorithms. We implement a space-efficient Grover-based search for $k$-SAT, and then explain how to efficiently implement Quantum Phase Estimation (QPE) for realizing quantum backtracking.
% how can one write ``in the context of `` here??

\subsection{Space-efficient formula evaluation oracles for $k$-SAT}
\label{app:efficient-grover}

%\begin{figure*}
%\begin{minipage}{\textwidth}

%\hspace{1em}
%\begin{minipage}{.45\textwidth}
\begin{figure}[b!]
\centering
\mbox{
 \Qcircuit @C=.7em @R=.7em {
  \lstick{\ket{x}} & \qw & \qw & \ctrl{3} \qw & \qw &\qw \\
  \lstick{\ket{y}} & \qw & \qw & \ctrl{2} \qw & \qw &\qw \\
  \lstick{\ket{z}} & \qw & \gate{X} & \ctrl{1} \qw & \gate{X} &\qw \\
  \lstick{\ket{0}} & \qw & \gate{X} & \targ  & \qw &\qw &\rstick{\ket{C(x)}}
 }}
\caption{Na\"ive implementation of the evaluation of clause $C \equiv (\neg x \vee \neg y \vee z) \equiv \neg (x \wedge y \wedge \neg z)$}
\label{fig:toffoli-clause}
\end{figure}

\begin{figure*}
\centering
\mbox{
 \Qcircuit @C=1em @R=.7em {
   \lstick{\ket{0}_\text{anc}} & /^p \qw & \qw
  & \multigate{2}{\mathcal{O}_f} & \qw & \qw & \qw & \qw & \cdots \\
  \lstick{\ket{0}} & /^n \qw & \gate{H^{\otimes n}} 
  & \ghost{\mathcal{O}_f} & \gate{H^{\otimes n}} & \gate{2 \ket{0^n}\bra{0^n} - I_n}         & \gate{H^{\otimes n}} & \qw & \cdots \\
  \lstick{\ket{1}} & \qw     & \gate{H}             
  & \ghost{\mathcal{O}_f}        & \qw                  & \qw                                       & \qw                  & \qw & \cdots \\
 }
 }
\caption{\label{fig:grover} Grover's algorithm: the oracle $\mathcal O_f$ xor's the result on the phase wire (bottom).}
\end{figure*}

%\end{minipage}

\begin{comment}
\begin{minipage}{.5\textwidth}
\begin{figure}[H]
\centering
 \mbox{
 \Qcircuit @C=1em @R=.7em {
  \lstick{\ket{0}} & /^n \qw & \qw & \ctrl{1} \qw & \qw & \ctrl{1} \qw & \qw & \cdots & &  \qw & \ctrl{1} \qw \\
  \lstick{\ket{0}} & \qw   & \gate{U_0}  & \gate{X} & \gate{U_1} & \gate{X} & \qw & \cdots & & \qw & \gate{U_l} \qw \\
 }
 }
\caption{Schematic for single qubit formula evaluation~\cite{one-qubit-evaluation}}
\label{fig:one-qubit-eval}
\end{figure}
\end{minipage}
\end{comment}

%\end{minipage}
%\end{figure*}

To implement Grover search of Figure~\ref{fig:grover} for $k$-SAT, the oracle $\mathcal O_f$ should evaluate a formula $F$ on an assignment of $n$ variables.
A naive oracle implementation uses $n+m+1$ qubits for a formula of $m$ clauses, including $p \equiv m+1$ ancilla qubits to evaluate each clause and their conjunction. This oracle implementation evaluates each clause and stores the result in a dedicated ancilla qubit. Figure~\ref{fig:toffoli-clause} shows an example which uses a $3$-controlled Toffoli gate.
After all clauses are evaluated, an $m$-controlled Toffoli gate can computed the final result from the dedicated clause qubits.
The reversal is trivial.

More efficient oracle implementations are possible: previous work~\cite{swenne} has shown that multi-controlled incrementers can be used to reduce the ancilla count $p$ from $m+1$ to $\lceil \log(m) \rceil +1$. Furthermore, using Fixed-Point Amplitude Amplification~\cite{fpaa}, the ancilla count can be further reduced to $2$~\cite{swenne}.

In the present work, we exploit another space-efficient quantum algorithm for formula evaluation, which only uses $1$ ancilla qubit, relying on prior work \cite{one-qubit-evaluation} which defines an one-qubit program which computes any Boolean function exactly: this program takes $x$ as input and the output of the measurement is $f(x)$. This one-qubit model, which alternates single-qubit unitaries and controlled gates, can be seen as a quantum variant of Barrington's theorem \cite{barrington}. Note that $k$-CNF formulas (which are CNF formulas) correspond to depth-2 (AC) circuits, so the theorem (a statement on NC$_1$) applies to $k$-SAT input.

We can use the one-qubit model to implement the oracle $\mathcal{O}_f$, using one ancilla qubit to apply single-qubit unitaries and CNOT gates, controlled over the qubits of the input.
%(see Figure~\ref{fig:one-qubit-eval}).\todo{Explain $U_1$ etc or leave out this figure} 
We refer the interested reader to \cite[Sec.~4]{one-qubit-evaluation} for the details of the implementations of the single-qubit unitaries of the one-qubit model.

Using the one-qubit model for formula evaluation to implement the oracle, we obtain a space-efficient implementation of Grover which uses $1$ qubit for each variable of the formula, $1$ qubit for the phase-flip oracle, and $1$ qubit for formula evaluation, for a total space complexity of $n+2$, as shown in Figure~\ref{fig:grover} with $p=1$. 

\subsection{Space-efficient quantum phase estimation}
\label{app:qpe}

The quantum backtracking framework which is used in this paper relies on the Quantum Phase Estimation (QPE) algorithm, a quantum algorithm which outputs  with high probability a $t$-bit estimate of the phase of its inputs: a unitary $U$ and an eigenvector $\ket{\psi}$. The canonical implementation of QPE uses $t$ qubits to produce its $t$-bit estimate. In this section, we study a more space efficient implementation of a specific version of QPE, sufficient for quantum backtracking, which only determines whether the phase is $0$.

\begin{proposition}
There is an implementation of QPE which checks whether the phase of the eigenvector (of a given unitary $U$ operating on $m$ qubits) is equal to zero with a $t$-bit precision using no more than $p+1$ ancilla qubits, where $p=\lceil\log(t)\rceil$ qubits are allocated for a counter from $0$ to $t$.
\end{proposition}

\begin{figure*}
\centering
\mbox{
 \Qcircuit @C=.7em @R=.7em {
  \lstick{\ket{0}_c^{\otimes p}} & /^p \qw & \qw & \qw & \qw & \qw & \qw & \qw & \qw & \gate{inc} & \qw & \gate{inc} & \qw & & \ket{\iota}\\
  \lstick{\ket{0}}  \gategroup{2}{2}{5}{8}{.9em}{--}  & \gate{H}  & \ctrl{3} & \qw  & \cdots & & \qw                & \gate{H} & \qw & \ctrl{-1} & \qw & \qw & \qw & & \ket{b_1}\\
  \lstick{\vdots\ \ } & \vdots   &     &          &            &     &        & &                    &  &        & \vdots &     \\
  \lstick{\ket{0}}    & \gate{H} & \qw      & \qw & \cdots & & \ctrl{1}           & \gate{H}	& \qw  & \qw & \qw & \ctrl{-3} & \qw & & \ket{b_t} \\
  \lstick{\ket{\psi}} & /^m \qw & \gate{U} & \qw & \cdots & & \gate{U^{2^{t-1}}} & \qw & \qw & \qw & \qw & \qw & \qw & \qw \\
 \lstick{\ \ }  &  &     &          &            &     &        & &                    &  &        & &    
 }
 }
\caption{\label{fig:t-bit-incrementer} The specialized QPE circuit, which which also counts the bits of the phase that are equal to 1. The circuit in the dotted rectangle determines whether the $t$-bit estimate of $\theta$ is null (as in canonical QPE).}
\end{figure*}

\begin{figure}
\centering
\mbox{
 \Qcircuit @C=.7em @R=.7em {
  \lstick{\ket{0}_c^{\otimes p}}  & /^p \qw & \qw & \qw & \qw & \gate{inc} & \qw & \qw & \cdots &  & \ket{\iota}\\
  \lstick{\ket{0}_a}    & \gate{H}  & \ctrl{1} & \gate{H} & \qw & \ctrl{-1} & \qw &  \qw & \cdots &  & \ket{a} \\
  \lstick{\ket{\psi}} & /^m \qw & \gate{U^{2^0}} & \qw & \qw & \qw & \qw & \qw &  \cdots & 
 }
 }
\caption{\label{fig:log-incrementer} Space efficient circuit for determining whether the $t$-bit  estimate of $\theta$ is null. The circuit is repeated for $U^{2^1}, \dots, U^{2^{t-1}}$}
\end{figure}

\begin{proof}
Consider the original QPE circuit represented in Figure~\ref{fig:t-bit-incrementer}, which uses $t$ ancilla qubits to obtain a $t$-bit estimate of $\theta$ such that $U\ket{\psi}=e^{2i\pi\theta}\ket{\psi}$, given $\ket{\psi}$ as input. %\todo{the circuit in figure 4 still uses $t$ bits: $b_1..b_t$. We can sequentialize this?}

From the initial state $\ket{0}^{\otimes t}\ket{\psi}$, we apply the Hadamard gate $H$ to each qubit of the first register, obtaining the state $\ket{+}^{\otimes t}\ket{\psi}$. We then apply unitaries $U,\ldots,U^{2^{t-1}}$, each controlled over one qubit of the first register, obtaining the state
\[
\ket{\psi'} = \frac{1}{\sqrt{2^t}} 
\bigotimes_{j} \left( \ket{0} + e^{2i\pi 2^j\theta} \ket{1}\right)
\]
To this state, we apply $H^{\otimes t}$, obtaining the final state.
\[
\ket{\psi_\text{f}} = \frac{1}{2^t} 
\bigotimes_j \left( (1 + e^{2i\pi 2^j\theta})\ket{0} + (1-e^{2i\pi 2^j\theta}) \ket{1}\right)
\]

Let us write $p_0$ for the probability that the $t$-bit estimate that in the final state $\ket{\psi_\text{f}}$ of this circuit is $b_1\ldots b_t=0\ldots0$, so that $p_0=|\alpha_0|^2$ where 
\[
\alpha_0 = \frac{1}{2^t} \prod_j \left(1+e^{2i\pi 2^j\theta}\right)
\]

Now, consider the circuit defined in Figure~\ref{fig:log-incrementer}. It is a more space-efficient implementation of QPE which only detects whether the phase is equal to $0$. It can be implemented using $\lceil\log(t)\rceil$ ancilla qubits, which add up to the number of qubits on which $U$ is defined.

This implementation stems from the following observation. Consider the QPE circuit which detects whether the phase is equal to $0$, to which we add an ancilla register which counts up to $t$ (requireing $p = \lceil \log(t) \rceil$ qubits) and gates which increment the counter whenever one of the bits of the $t$-bit estimate is not equal to $0$, as depicted in Figure~\ref{fig:t-bit-incrementer}. Since the counter is only increased (and never decreased) throughout the computation, the value of the counter is only non-zero in the final state if a $1$ appeared in one of the bits. In other words, the counter can only have a null value when the all zero state is the $t$-bit estimate.

This observation leads us to the implementation of a circuit which, for each $j$ ($0 \leq j < t$), applies the Hadamard gate $H$ on an ancilla qubit $a$, followed by the unitary $U^{2^j}$ controlled on $a$, the Hadamard gate $H$ on $a$, and an incrementation of the counter if the value of $a$ is non-zero, as pictured in Figure~\ref{fig:log-incrementer}. 

For any $j$, assume that the counter is in the all zero state and $a$ is also $\ket{0}$. From our previous observation, this means that the counter has not been increased yet. Write $\ket{\psi_j}$ and $\ket{\psi'_j}$ respectively for the overall states before and after the controlled incrementation.
\begin{align*}
    \ket{\psi_0} &= \ket{0}_c^{\otimes p}\frac{1}{2}\left( (1 + e^{2i\pi \theta})\ket{0} + (1-e^{2i\pi \theta})\ket{1} \right)\\
    \ket{\psi'_0}&= \frac{1}{2}\ket{0}_c^{\otimes p}(1 + e^{2i\pi\theta})\ket{0}\\
    &+ \frac{1}{2}(1-e^{2i\pi \theta})\ket{0\ldots01}_c \ket{1}
\end{align*}

The state $\ket{\psi'_0}$ contains a `non-zero branch', which carries over throughout the next step of the computation, so that
\begingroup
\allowdisplaybreaks
\begin{align*}
    \ket{\psi_1} &= \frac{1}{2}(1 + e^{2i\pi\theta})\ket{0}_c^{\otimes p}\\
    &\otimes \frac{1}{2}\left( (1 + e^{4i\pi \theta})\ket{0} + (1-e^{4i\pi \theta})\ket{1} \right)\\
    &+\text{ non-zero branch}\\
    \ket{\psi'_1}&= \frac{1}{4}(1 + e^{2i\pi\theta})(1 + e^{4i\pi\theta})\ket{0}_c^{\otimes p}\ket{0}\\
    &+\text{ non-zero branch}\\
    &\vdots \\
    \ket{\psi'_t} &= \frac{1}{2^t} \prod_{j=0}^{t-1}(1+e^{2\pi i 2^j\theta})\ket{0}^{\otimes p}_c\ket{0}\\
    &+\text{ non-zero branch}
\end{align*}
\endgroup
where $\ket{\psi'_t}$ is the final state $\ket{\psi'_\text{f}}$ of the circuit.

The probability $p'_0$ of obtaining a zero value in the final state of the counter is $p'_0=|\alpha'_0|$ where
\[
\alpha'_0 = \frac{1}{2^t} \prod_j \left(1+e^{2i\pi 2^j\theta}\right) = \alpha_0
\]
Therefore, $p_0=p'_0$, which means that the probability of observing the all zero $t$-bit estimate in the first circuit and the probability of observing the value $0$ in the counter of the second circuit are exactly the same. By linearity this observation holds for all input states, implying that we have come up with a more efficient implementation of QPE which detects whether the phase is equal to $0$ (up to a $t$-bit estimate) while using only $\lceil\log(t)\rceil$ qubits.
\end{proof}

{Note that this construction is important to the present work, as the quantum backtracking framework's use of QPE requires a $O(\log(T))$-bit precision, where $T$ is the size of the search tree considered. When $T$ is exponential in $n$, the standard implementation of QPE yields a linear overhead which does not prevent us from using the hybrid approach, but directly weakens the computational efficiency of hybrid schemes. Bringing down QPE's overhead from $O(\log(T))$ to $O(\log(\log(T)))$ implies almost no loss between real and effective quantum computer size for many problems.}

%Efficient quantum backtracking, and its variant
\section{Quantum backtracking for DPLL-like algorithms}
\label{app:backtracking}

In this section, we develop a space-efficient implementation of the quantum backtracking framework~\cite{qbacktracking} for DPLL-like algorithms. It is an essential element for the hybrid approach developed throughout this paper:
{Whenever the classical computation hands of a sub-problem corresponding to a (restricted) formula $F$ to the quantum computer, it generates the space-efficient quantum circuit from $F$ as described here.}

As DPLL provides no guarantee on the maximal number of guesses which have to be done to reach a satisfying assignment, we represent tree vertices with the full partial assignment they correspond to, rather than the branching choices (guesses) as we do for PPSZ in Appendix~\ref{app:sia}.

\subsection{The quantum backtracking framework}
\label{sub:framework}

We present the quantum backtracking framework developed in \cite{qbacktracking}, closely following their exposition.

The search tree underlying a classical $k$-SAT-solving backtracking algorithm is formalised as a rooted tree $\mathcal T$ of depth $n$ and with $T$ vertices $r,1,\ldots,T-1$. 
Each vertex is labeled by a partial assignment, and a marked vertex is a vertex labeled by a satisfying assignment.
We work under the promise that the formula given as input is not trivially satisfied, and therefore the root is promised not to be marked.

We write $\ell(x)$ for the distance from the root $r$ to a vertex $x$, and assume that $\ell(x)$ can be determined for each vertex $x$ (even without full knowledge of the structure of the tree). The examples developed in this paper make this consideration trivial, as vertices are labeled by list of variable assignments, with the special symbol $\ast$ marking unassigned variables, and therefore the distance from the root can be calculated from the number of assigned variables.

We write $A$ (resp. $B$) for the set of vertices an even (resp. odd) distance from the root, with $r \in A$. We write $x \rightarrow y$ to mean that $y$ is a child of $x$ in the tree. For each $x$, let $d_x$ be the degree of $x$ as a vertex in an undirected graph. 
So for every vertex $x$ which isn't the root, we have $d_x = |\{y \mid x \rightarrow y\}| + 1$ and $d_r = |\{y \mid r\rightarrow y\}|$.

We define the quantum walk\footnote{Note that this notion of quantum walk does not involve a separate ``coin'' space.} as a set of diffusion operators $D_x$ on the Hilbert space $\mathcal{H}$ spanned by $\{\ket{r}\} \cup \{ \ket{x} :  x \in \{1,\dots,T-1\} \}$, where $D_x$ acts on the subspace $\mathcal{H}_x$ spanned by $\{\ket{x}\} \cup \{ \ket{y}: x \rightarrow y \}$. We take $\ket{r}$ to be its initial state.

Such diffusion operators $D_x$ are defined as the identity if $x$ is marked, and as follows otherwise:
\begin{itemize}
\item for $x \neq r$, $D_x = I - 2 \dm{\psi_x}$, where
\[ \ket{\psi_x} = \frac{1}{\sqrt{d_x}} \left( \ket{x} + \sum_{y, x \rightarrow y} \ket{y} \right). \]
\item $D_r = I - 2 \dm{\psi_r}$, where
\[ \ket{\psi_r} = \frac{1}{\sqrt{1+ d_r n}} \left( \ket{r} + \sqrt{n} \sum_{y, r \rightarrow y} \ket{y}\right). \] 
\end{itemize}
Note that when $x$ is an unmarked leaf (in the context of $k$-SAT, a vertex which corresponds to a contradiction), it has no neighbors and therefore the reflectors are about the state itself.

We define two Szegedy-style walk operators as follows:
\[
R_A = \bigoplus_{x \in A} D_x
\text{ and }
R_B = \dm{r} + \bigoplus_{x \in B} D_x.
\]
 
Assume you have access to $R_A$ and $R_B$, and consider the following algorithm.

\begin{enumerate}
\item[] {\textbf{Detecting a marked vertex}}\\
{\bf Input:} Operators $R_A$, $R_B$, a failure probability $\delta$, upper bounds on the depth $n$ and the number of vertices $T$. Let $\beta, \gamma > 0$ constants given in Ashely's paper.
\begin{enumerate}
\item Repeat the following subroutine $K = \lceil \gamma \log (1/\delta) \rceil$ times:
\begin{enumerate}
\item Apply phase estimation to the operator $R_B R_A$ with precision $\beta/\sqrt{Tn}$, on the \textit{initial} state $\ket{r}$.
\item If the eigenvalue is 1, accept; otherwise, reject.
\end{enumerate}
\item If the number of acceptances is at least $3K/8$, return ``marked vertex exists''; otherwise, return ``no marked vertex''.
\end{enumerate}
\end{enumerate}

In essence, this algorithm detects whether the tree $T$ has a marked vertex using $O(\sqrt{Tn} \log(1/\delta))$ queries to $R_A$, $R_B$. Detection is enough: to find a satisfying assignment, it suffices to traverse the tree and ask which subtree leads to a satisfying assignment at each branching. The remainder of this appendix is dedicated to a space-efficient implementation of quantum backtracking for DPLL-like algorithms.

\subsection{Encoding sets of variables}
\label{sub:efficient-encoding}

We first describe how partial assignments are encoded and define subroutines which allow us to manipulate assignments. 
As a partial assignment for a set Vars of variables can be seen as a function $\text{Vars} \to \{0,1,\ast\}$, a vertex $x$ can be uniquely labeled in this $n$-trit system, and can hence be stored in a quantum state $\ket{\vec x}$ of $\log_2(3)n$ qubits.

Since we have the restricted formula available whenever the classical algorithm constructs a quantum circuit, we could also hard code the wires according to the clauses, but to ease drawing of the circuits, we use a unitary instead.
Whether a given variable $x_i$ is in a clause $C_j$ can be easily checked using a Toffoli gate and Pauli-$X$ gates to determine whether a given index $i$ is in the clause or not. 
Such unitaries form a family of unitaries
\[
\text{Check}_{j,i}:\ket{\vec x}\ket{0}\ket{0}\mapsto\ket{\vec  x}\ket{b}\ket{s},
\] 
where $b=1, s=0$ ($s=1$) if the literal $x_i$ ($\overline{x_i}$) appears in $C_j$.

%Moreover, one can check whether a variable in a partial assignment has been set (i.e. whether its value is different from $\ast$) using a controlled unitary.
%Such unitaries form a family of unitaries $\text{Assigned}_{j,i}:\ket{\vec x}\ket{0}\mapsto\ket{\vec x}\ket{b}$, where $b=1$ {if the variable $x_i$ has an assigned value in the clause $C_j$.}
% Using controlled unitaries, one can query the index and the sign of a variable within a clause.
 
To realize the search predicate $P(\vec{x})$, i.e, whether $\vec x$ is a leaf of the search tree, 
we define the unitary $V_\text{leaf}$. This is a unitary such that $V_\text{leaf}\ket{\vec x}\ket{0} = \ket{\vec x}\ket{b}$, where $b$ is a Boolean value equal to $1$ if and only if $\vec x$ corresponds to a trivial formula (whenever $P(\vec x)$ should be either 0 or 1, as the formula $F_{|\vec x}$ is (un)sat). 
The find solutions, we make use of the unitary $V_\text{marked}$ such that $V_\text{marked}\ket{\vec x}\ket{0}=\ket{\vec x}\ket{b}$, where $b$ is a Boolean value equal to $1$ if and only if $\vec x$ is a satisfying assignment. Those unitaries are hardcoded based on the restricted formula, i.e, by creating a reversible circuit for $F_{|\vec x}$.
%How these unitaries are actually implemented depends on the algorithm, and we provide the examples for PPSZ and DPLL shortly.

\subsection{Implementing the walk operator}
\label{sub:implement-walk}

We describe step by step an implementation of the walk operator $W=R_BR_A$, for DPLL-like backtracking algorithms. Let $\vec x$ be a partial assignment of the formula (i.e. a vertex in our search tree).
We provide a quantum reversible implementation of the routines $ch1(\vec x), ch2(\vec x, b)$ and $chNo(\vec x)$ specified in Section~\ref{sec:quantum}.

We define reversible routines which implement both the reduction rule and the branching heuristic (see Section~\ref{sec:quantum}), in order to check whether the unit rule or the pure literal rule can be applied to one of the unassigned variables in assignment $\vec x$ according to a fixed (static) variable ordering. In other words, the first unassigned variable $x$ for which there is a clause $C \in F_{|\vec x}$ that is unit, i.e., $C = \set{x}$ or $C= \set{\bar x}$, is forced accordingly. The same is done for the pure literal rule.
We provide implementations of the unit rule in Appendix~\ref{sub:implement-unit} and of the pure literal rule in Appendix~\ref{sub:implement-pure}.

We implement $ch1(\vec x), ch2(\vec x, b)$ and $chNo(\vec x)$ as unitaries
$V_1$ and $V_2$ which respectively compute the first child $\vec x_1$ (assuming $\vec x$ is a non-leaf), and the second child $\vec x_2$, assuming $\vec x$ is not forced (guessed) and non-leaf. In other words, in the 2 children case, $V_i$ implements $\vec x \rightarrow ch2(\vec x, i)$.
The specification adds $V_0$ for computing the identity function for leaf $\ket{\vec x}$
i.e., $\vec x_0 \defn \vec x$:
\[
V_i: \ket{\vec x}\ket{\vec 0} \mapsto \ket{\vec x}\ket{\vec x_i}
\]

\noindent
We implement a unitary $V_A(c)$ which, given a vertex $\vec x$ outputs the superposition of the vertex $\vec x$ and its $c$ children. Its specification is:
\begin{align*}
V_A(c): \ket{\vec x} \ket{\vec 0} \ket{0}
&\xmapsto{H} \frac{1}{\sqrt{c}} \sum_{0 \leq i \leq c} \ket{\vec x}\ket{\vec 0}\ket{i}\\
&\xmapsto{\text{ctrl-}{V_i}}   \frac{1}{\sqrt{c}} \ket{\vec x} \sum_{0 \leq i \leq c} \ket{\vec x_i}\ket{i}\\
&\xmapsto{V_C(c)} \frac{1}{\sqrt{c}} \ket{\vec x} \sum_{0 \leq i \leq c} \ket{\vec x_i}\ket{0}
\end{align*}
%where $c$ is the number of children of $\vec x$. 

In other words, we apply each $\text{ctrl-}V_i$ controlled over a qutrit (the `index register') which determines which $V_i$ to apply, with $i=0$ being the operation which copies the parent vertex. The result is an entangled state between the index register and specification of the children, rather than a superposition over children. 
The operation $V_C(c)$ disentangles the state, using the fact that we can prepare, and hence, unprepare the $i$-th child. For each child index $i$, we erase the index register as follows; we uncompute the $i$-th child by inverting $V_i$. Now, in the child register, in the branch with the $i$-th child we have the ``all-zero'' state. Since all children differ, this is the only branch with the all-zero state. Then, conditioned on the child-specifying register being in the all-zero state, we subtract $i$, from the child-index register. Finally, we recompute the $i$-th child, which restores the children specifications in all branches. This erases the which-child information for each child, leaving the index register in the $\ket{0}$ state for all children.

%\begin{align*}
%V_C &\mapsto 
%\frac{1}{\sqrt{c}} \ket{\vec x} \sum_{0 \leq i \leq c} \ket{\vec x_i}\ket{i}\\
%&\mapsto
%\frac{1}{\sqrt{c}} \ket{\vec x} \sum_{0 \leq i \leq c} \ket{\vec 0}\ket{i}\\
%&\mapsto
%\frac{1}{\sqrt{c}} \ket{\vec x} \sum_{0 \leq i \leq c} \ket{\vec 0}\ket{0}\\ 
%&\mapsto
%\frac{1}{\sqrt{c}} \ket{\vec x} \sum_{0 \leq i \leq c} \ket{\vec x_i}\ket{0}
%\end{align*}

Now, the operator $R_A$ can be implemented 
%by checking whether the vertex $\vec x$ is marked, and otherwise applying 
as $U_A\left(\text{Id} - \ket{0}\bra{0}\right)U^\dagger_A$, where $U_A = \bigoplus_{\vec x \in A} U_A^{\vec x}$ is the unitary which computes the state $\ket{\varphi_x}$, i.e., %~\todo{$\phi_x$?}
\[
U_A\ket{\vec x}\ket{0}=\ket{\vec x}\ket{\varphi_{\vec x}}
\] 

%In order to implement $R_A$, 

Recall that the implementation of the diffusion operators depends on the parity of their distance from the root. We implement each $D_x$ for $x\in A$ as a unitary $U_A^{x}$ by checking whether $\vec x$ has zero, one or two children. This is done by checking the number $c$ of children, and controlled on the result bit of this operation, applying (or not) the corresponding operation $V_A(c)$ which generates the superposition over the child(ren) and original vertex. 

The operator $R_B$ is implemented in a similar fashion to $R_A$, assuming that we have access to a unitary $V_\text{root}$ which checks whether $\vec x = r$. Observing that the root $r$ is associated to the all-undetermined satisfying assignment, i.e. where $\mathbf r = \ast^n$. Such a unitary can easily be implemented with a counter to $n$ and incrementation controlled on each variable being equal to $\ast$.

In the generic quantum backtracking algorithm, a depth counter is maintained to determine the parity of the depth at which the vertex $\vec x$ is. We forgo of such a register by observing that each variable assignment takes us one level deeper into the tree, and therefore the depth at which $\vec x$ is at is given is defined by the number of variables which have already been assigned a value. Therefore, to check the parity of the depth, it suffices to count every time $x_i \neq \ast$.

\subsubsection*{Cost analysis of the implementation}
\label{sub:cost-analysis-walk}

The following theorem shows that our implementation of the walk operator, as part of a space-efficient implementation of quantum backtracking, uses only a near-linear amount of qubits. Note that the routine implemented in this section have a polynomial time complexity.

\begin{theorem}
\label{thm:implement-walk}
%Assuming access to the unitaries $V_A$, $V_C$, $V_\text{leaf}$, $V_\text{marked}$ and $V_\text{root}$, one can implement 
There is a polynomial-time implementation of the walk operator of the quantum backtracking framework for DPLL-like algorithms which uses at most $4n+w$ qubits, with $w \in O(\log(n))$.
\end{theorem}
\begin{proof}
Having access to the unitary $R_A$ and $R_B$, we implement the walk operator $R_BR_A$ with an ancilla qubit which checks whether the vertex $\vec x$ that we are considering is odd or even. Defining $\cost{U}$ to be the space (in terms of qubits) required by our implementation of a unitary $U$, we obtain that 
\[\cost{R_BR_A} \leq \text{max}(\cost{R_A},\cost{R_B}) +1.\] 
The operators $R_A$ and $R_B$ have very similar implementations, and under our implementation:
\begin{align*}
\cost{R_A} &\leq \cost{U_A} +1 \\
	&\leq \lceil \log_2(3) \rceil \cdot n + \cost{V_A}\\
	&+ \cost{V_C} + O(1)
\end{align*}
where $\lceil \log_2(3) \rceil \cdot n$ corresponds to the space required to store a vertex, the implementation of $V_A$ and $V_C$ requires $\lceil \log_2(3)\rceil \cdot n + O(\log(n))$ additional ancillas (see Section~\ref{sub:implement-unit}), adding up to an overall $4n+O(\log(n))$ space complexity for $R_A$ (as $\lceil 2\log_2(3)\rceil = 4$).
\end{proof}

\subsection{Implementing the unit clause rule}
\label{sub:implement-unit}

In order to determine the next vertices in the search tree according to the unit rule, one needs to determine whether there exists a unit clause (i.e. a clause $C=\{l\}$ with only one literal $l$), given a partial assignment of the formula that we are considering. Whenever a unit clause is found, no branching occurs and the literal is simply set to true.

In order to implement this process, we need two operations:
\begin{itemize}
 \item An operation $V^{(i)}_\text{unit}$ which checks whether the $i$-th clause is a unit clause.
 \item An operation $V_\text{next}$ which outputs the next partial assignment.
\end{itemize}

\begin{figure*}
\centering
\mbox{
 \Qcircuit @C=1em @R=.7em {
 & & & & & & & & & \mbox{~~~~~\textbf{U}}\\
  \lstick{\ket{\vec x}} & /^n \qw & \multigate{2}{V_\text{unit}^{(1)}} & \qw & \qw & \cdots & & \qw & \multigate{2}{V_\text{unit}^L} & \qw & \qw & \qw & \multigate{3}{U^\dagger} & \qw\\
  \lstick{\ket{0}_1} & \qw  & \ghost{V_\text{unit}^{(1)}}  & \ctrl{2} & \qw & \cdots & & \qw & \ghost{V_\text{unit}^{(L)}} & \ctrl{2} & \ctrl{3} & \qw & \ghost{U^\dagger} & \qw\\
  \lstick{\ket{0}_2} & \qw & \ghost{V_\text{unit}^1} & \qw & \qw & \cdots & & \qw & \ghost{V_\text{unit}^{(L)}} & \qw & \qw  & \ctrl{3} & \ghost{U^\dagger} & \qw\\
  \lstick{\ket{0}} & \qw & \ctrlo{-1} & \gate{inc} & \qw & \cdots & & \qw & \ctrlo{-1} & \gate{inc} \gategroup{1}{2}{5}{10}{.7em}{--} & \qw &\qw & \qw & \qw \\
  \lstick{\ket{0}_3} & \qw & \qw & \qw & \qw & \qw & \qw & \qw & \qw & \qw & \targ & \qw & \qw & \qw\\
  \lstick{\ket{0}_4} & \qw & \qw & \qw & \qw & \qw & \qw &\qw &\qw &\qw &\qw &\targ &\qw &\qw
 }
 }
\caption{\label{fig:unit-clause} $V_\text{unit}$ checks whether there is a unit clause}
\end{figure*}

For each clause $C_i$, we implement the unitary 
\[
V^{(i)}_\text{unit}:\ket{\vec x}\ket{0}_1\ket{0}_2\mapsto\ket{\vec x}\ket{j}_1\ket{s}_2
\]
which checks whether the $i$-th clause $C_i$ is a unit clause under the partial assignment $\vec x$.
It does so by checking whether the clause $C_i$ is still `alive'
(not already satisfied under the partial assignment $\vec x$), and then checking whether it is a unit clause, to finally output whether it is a unit clause ($j \neq 0$), and if yes, the index $j$ and polarity $s$ of the forced variable. We let the unitary $\mathcal{C}_i$ determine whether $C_i$ is alive by assignment $\vec x$ (by hardcoding the clause). And the unitary $\text{IsUnit}_i$ checks whether $k-1$ variables of the clause $C_i$ have been assigned a value: it is implemented with a counter from $0$ till $k$, with an incrementation controlled on variables having a value different from $\ast$; and it outputs an index $j$ and a sign $s$ if $C_i$ is a unit rule, $0$ otherwise. Then the operation $V^{(i)}_\text{unit}$, which checks whether the $i$-th clause $C_i$ is a unit clause (and if yes output the index and sign of its variable), is implemented in Figure~\ref{fig:unit-clause-i} ($\log(k)$-qubit counter ancilla is omitted).

In order to apply the unit rule, we apply $V^{(i)}_\text{unit}$ to each clause $C_i$ in the formula studied, and stop whenever we find a unit clause (i.e. whenever $V^{(i)}_\text{unit}$ outputs $\ket{\vec x}\ket{j}\ket{1}$), see Figure~\ref{fig:unit-clause} for the implementation of the unitary
\[
V_\text{unit}:\ket{\vec x}\ket{0}_3\ket{0}_4 \mapsto \ket{\vec x}\ket{j}_3\ket{s}_4
\]

Note that, to go through all the $L$ clauses of a formula, we need a clause counter which is implemented using $\lceil\log(L)\rceil \in O(\log(n))$ ancilla bits, since a k-SAT formula has at most ${2n\choose k} \in O(n^k)$ clauses.

If a unit clause is found, the current vertex only has one child, given by $V_1$, which is obtained by applying the following unitary $V_\text{next}$ where $b,j$ are the outputs of $V^{(i)}_\text{unit}$:
%if the formula contains the $\{x_j\}$, and $\ket{b}=\ket{1}$ if the formula contains the $\{\overline{x}_j\}$
\[
V_\text{next}: \ket{\vec x}\ket{0}\ket{j}\ket{b} \mapsto \ket{\vec x}\ket{\vec x'}\ket{j}\ket{b}
\]
%where $\vec x'$ is defined as the partial assignment $\vec x$ with $x_j$ set to~$b$. 
Such a unitary is implemented by copying $\vec x$ to the output register, while setting the $j$-th index to the value $b$.

If the unit rule is not applied, we can determine the $j$ as next unassigned variable in partial assignment $\vec x$
(the first $\ast$), in a similar fashion as above.
Then, we obtain $V_1$ (resp. $V_2$) by applying $V_\text{next}$ with and
 $\ket{b}=\ket{0}$ (resp. $\ket{b}=1$), so that given $\ket{\vec x}\ket{0}\ket{j}\ket{0}$, $V_1$ (resp. $V_2$) outputs $\ket{\vec x}\ket{\vec x[x_j=0]}\ket{j}\ket{0}$ (resp. $\ket{\vec x}\ket{\vec x[x_j=1]}\ket{j}\ket{1}$).

%If no unit clause was found, we branch over the current variable, and t

\subsection{Implementing the pure literal rule}
\label{sub:implement-pure}

The pure literal rule eliminates variables $x_i$ which only appear as the literal $x_i$ or only appear as the literal $\overline{x_i}$. In which case, the variable is set to the value which makes the literal true, eliminating all the clauses which contains it in the process. 

Classically, it is a convenient way to reduce the number of clauses manipulated by the backtracking algorithm considered. However, the quantum backtracking implementation that we present in this paper is not concerned with formula rewriting. 

We implement a unitary 
\[
V_\text{pure}:\ket{\vec x}\ket{0}_1\ket{0}_2\mapsto 
\ket{\vec x}\ket{j}_1\ket{s}_2
\]
which checks whether there is a pure literal, and if yes, outputs  want the index $j$ of its variable and its polarity $s$. For completeness, $(j,s)=(0,0)$ if no pure literal is found. 

\begin{figure}
\centering
\mbox{
 \Qcircuit @C=1em @R=.7em {
  \lstick{\ket{0}_1} & \qw & /^{\log(n)} \qw & \qw & \qw & \multigate{2}{\text{IsUnit}_i}  & \qw & \qw & \qw & \qw & \qw \\
  \lstick{\ket{0}_2} & \qw & \qw & \qw & \qw & \ghost{\text{IsUnit}_i}  & \qw & \qw & \qw & \qw & \qw \\
  \lstick{\ket{\vec x}} & \qw & /^n \qw & \multigate{1}{\mathcal{C}_i} & \qw & \ghost{\text{IsUnit}_i} & \qw & \qw & \qw & \multigate{1}{\mathcal{C}_i^\dagger} & \qw\\
  \lstick{\ket{0}} & \qw & \qw  & \ghost{\mathcal{C}_i}  & \qw & \ctrl{-1} & \qw & \qw & \qw & \ghost{\mathcal{C}_i^\dagger} & \qw\\
 }
 }
\caption{\label{fig:unit-clause-i} $V_\text{unit}^{(i)}$ checks whether $C_i$ is a unit clause}
\end{figure}

\begin{figure*}
\centering
\mbox{
 \Qcircuit @C=1em @R=.7em {
  \lstick{\ket{0}_4}  & \qw & \qw & \qw & \multigate{2}{\text{Check}_{j,i}}  & \qw & \ctrl{3} & \ctrlo{3} & \qw  & \qw \\
  \lstick{\ket{0}_3} & \qw & \qw & \qw & \ghost{\text{Check}_{j,i}}  & \qw & \ctrl{3} & \ctrl{4} & \qw & \qw \\
  \lstick{\ket{\vec x}} & /^n \qw & \multigate{1}{\mathcal{C}_j} & \qw & \ghost{\text{Check}_{j,i}} & \qw & \qw & \qw & \multigate{1}{\mathcal{C}_j^\dagger} & \qw\\
  \lstick{\ket{0}} & \qw  & \ghost{\mathcal{C}_j}  & \qw & \ctrl{-1} & \qw & \qw & \qw & \ghost{\mathcal{C}_j^\dagger} & \qw\\
  \lstick{\ket{c_+}} & \qw & \qw & \qw & \qw & \qw & \gate{inc} & \qw & \qw & \qw  \\
  \lstick{\ket{c_-}} & \qw & \qw & \qw & \qw & \qw & \qw & \gate{inc} & \qw & \qw \\
 }
 }
\caption{\label{fig:pure-sub} $V_\text{pure}^{(j,i)}$ checks whether the variable $x_j$ appears in $C_i$ and if yes, increases the counter corresponding to its polarity}
\end{figure*}

The implementation of the unitary $V_\text{pure}$ is quite similar to the implementation of the unit rule. For each variable $x_j$, we only check whether it is a pure literal if no pure literal was found beforehand. In order to check whether a given variable $x_j$ is part of a pure literal, we implement a unitary 
\[
V_\text{pure}^{(i)}:\ket{\vec x}\ket{0}_1\ket{0}_2 \mapsto \ket{\vec x}\ket{j}_1\ket{s}_2
\]
which implements the following steps:
\begin{enumerate}
 \item count the number of `alive' clauses in which literals $x_i$ and $\bar{x}_i$ respectively appear with two clause counters, by applying for each clause $C_i$ the circuit $V_\text{pure}^{(j,i)}$ described in Figure~\ref{fig:pure-sub} and uncomputing the ancillas in register $3$ and $4$;
 \item use a Toffoli gate to determine whether only one of the counters is equal to $0$, and if yes, we copy the index and sign of $x_i$ in the output registers. 
\end{enumerate}
Then, the unitary $V_\text{pure}$ simply applies $V_\text{pure}^{(i)}$ for each variable $x_i$, provided that no pure literal rule was found so far (this information is tracked with an ancilla variable counter of size $\lceil\log(n)\rceil$); and then uncompute the ancillas by applying the inverse of the circuit so far (as we did for the unit rule, see Figure~\ref{fig:unit-clause}). 

This procedure is implemented using $O(\log(n))$ ancillas. As in the implementation of the unit rule (see~Appendix~\ref{sub:implement-unit}), we use $V_\text{next}$ to compute the next partial assignment. Note that, we can implement in the same way a reduction rule which first applies the unit rule, and then applies the pure litteral rule if no unit clause is found.

\subsection{Quantum tree size estimation}
\label{sec:quantum-tree-size-estimation}
One drawback of Montanaro's quantum backtracking is that the runtime depends on the estimate of the size of the search tree (which is a parameter of the algorithm), and not on the size of the subtree that the classical backtracking algorithm explores. 

The efficiency of a classical backtracking algorithm relies on its ability to explore the most promising branches first, which means than in practice, the algorithm may find a marked vertex after exploring $T'$ vertices, where $T' \ll T$. Using a quantum tree size estimation subroutine to estimate the size of the subtree explored by the classical backtracking algorithm, 
the original quantum backtracking algorithm can be improved for the $T' \ll T$ case (see Th.~\ref{th:ambainis}).

\begin{theorem}[\cite{ambainis-kokainis}]
\label{th:ambainis}
Consider a classical backtracking algorithm $\mathcal A$ which generates a search tree $\mathcal T$. There is a quantum algorithm which outputs $1$ with high probability if $\mathcal T$ contains a marked vertex and $0$ if it doesn’t, with query complexity $O(n^{\frac{3}{2}}\sqrt{T'})$         
where $T'$ is the number of vertices actually explored by $\mathcal{A}$.
\end{theorem}

The overall algorithm generates subtrees which contain the first $2^i$ vertices explored by the classical backtracking algorithm, increasing $i$ until a marked vertex is found, or the whole search tree is considered. It is on each subtree containing the first $2^i$ vertices that we run the quantum backtracking algorithm.

Note that quantum backtracking with tree size estimation is only considered when $T' \ll T$. Because this algorithm is less performant than the original when $T'$ is close to $T$, one can just switch to Montanaro's algorithm whenever the complexity of the generation of the path exceeds the complexity of the original quantum backtracking.

The main component of this variant of Montanaro's quantum backtracking  are quantum backtracking itself and a quantum tree size estimation algorithm (Algorithm~1 in~\cite{ambainis-kokainis}), 
{which both run QPE as a subroutine to detect whether the phase of a given unitary is equal to $0$ (and therefore we can still apply the logarithmic space construction of Appendix~\ref{app:qpe}).}   
Therefore it is sufficient to prove that quantum backtracking can be implemented efficiently in order to benefit from the speedup provided by Theorem~\ref{th:ambainis} in the hybrid framework.

%Eppstein's algorithm
\section{Eppstein's algorithm for the cubic Hamiltonian cycle problem}
\label{app:eppstein}

In \cite{dnq-eppstein}, a hybrid divide-and-conquer algorithm for the Eppstein's algorithm for the cubic Hamiltonian cycle problem was provided based on Grover's search methods. 

In the algorithm presented in \cite{dnq-eppstein}, the quadratic improvement was achieved over the possible number of branching choices $(n/2)$ whereas the size of the overall search tree is upper bounded by $O(2^{n/3})$ (see~\cite[Appendix~A]{dnq-eppstein} for detailed explanations\footnote{This $n/2$ bound on the number of branching choices follows from the fact that the backtracking algorithm presented in \cite{dnq-eppstein} generates \textit{full} binary trees of depth at most $s/2$ (see~\cite[Proposition~10]{dnq-eppstein}), where $s$ is the effective problem size which is such that $s \leq n$.}), where $n$ represents the number of edges in the graph.

{Consequently, Grover's approach yields a polynomial improvement in (it has a $O^\ast(2^{n/4})$ time complexity), which is not a near-quadratic speed-up over classical $O^\ast(2^{n/3})$ time implementations of Eppstein algorithm (such quantum algorithm has an $O^\ast(2^{n/6})$ time complexity).}
This issue has been, outside of context of hybrid methods clarified and resolved in \cite{moylett-linden-montanaro} using quantum backtracking.

In this section, we succinctly show how the quantum backtracking method for this problem can be applied in our hybrid context, achieving a near-quadratic speed up in the sub-tree, without any relevant loss off space efficiency (which translates to cut-points, and hence overall efficiency of the hybrid method) caused by relying on backtracking rather than Grover's search.

We first provide the following background for the benefit of the reader.
The forced cubic Hamiltonian cycle problem (FCHC) is a \NP-complete problem which asks whether a cubic graph $G=(V,E)$ (i.e., with degree 3) has a Hamiltonian cycle which contains at least all the edges in a given subset $F \subseteq E$.

Given a FCHC instance $(G,F)$ as input, there is a classical divide-and-conquer algorithm $\mathcal E$ which solves the FCHC problem in time $O^\ast(2^{n/3})$ (a bound which constitutes an upper bound on the tree size) by selecting an unforced edge $e$ and branching over the two subinstances $(G,F \cup \{e\})$ and $(G \setminus \{e\}, F)$ (see Algorithm~2 in \cite{dnq-eppstein}). 

On a high-level, $\mathcal E$ is a backtracking algorithm which, on a given instance, does the following:
\begin{enumerate}
 \item apply several reductions (in order to simplify the problem)
 \item check whether one of the terminal conditions of the problem has been fulfilled (i.e. checks whether we can directly answer true or false)
 \item choose the next edge to branch upon.
\end{enumerate}

One can see $\mathcal E$ as an algorithm which explores binary search trees, whose root is labelled by the instance given as input, and each node (labelled by $(G,F)$) is either a leaf, or has two children (labelled by $(G,F \cup \{e\})$ and $(G \setminus \{e\}, F)$).

Now, for an implementation of quantum backtracking for algorithm $\mathcal E$ in the context of the hybrid approach, it is sufficient to implement space efficiently routines which checks whether the instance at any given node is a leaf (a routine denoted $V_\text{leaf}$), and otherwise what its children are ($V_A$), as explained in Appendix~\ref{sub:implement-walk}.

In \cite[Section~V.C]{dnq-eppstein}, reversible subroutines are introduced to re-construct the instance (labelling a node) from $\vec v$ (using $O(s \log(n/s) + s + \log(n))$ ancillas and $O(\poly(n))$ gates \cite[Corollary~2]{dnq-eppstein}), and test whether it satisfies a terminal condition (using $O(\log(n))$ ancillas and $O(\poly(n))$ gates \cite[Corollary~3]{dnq-eppstein}), where $s$ is the effective problem sized defined by $s=n-|F|-|C(G,F)|$ and $C(G,F)$ is the set of $4$-cycles which are disconnected from $F$.  

Now, we take those subroutines as an implementation $V_\text{leaf}$, the routine of quantum backtracking which checks whether the current node is a leaf. The algorithm $\mathcal E$ always branches over two choices (adding or removing an edge). In the implementation of quantum backtracking for $\mathcal E$, one can construct an unitary $V_A$ which given the label $\vec v=v_1\ldots v_{i}$ of the current node, determines the label of its children $\vec w$ and $\vec w'$, which are respectively $v_1\ldots v_i 0$ and $v_1\ldots v_i 1$ (see Appendix~\ref{sub:implement-walk}).

%Presentation of SIA circuit
\newcommand\cmmnt{\hfill$\triangleright$~}

\section{Reversible Simulation of SIA}
\label{app:sia}

% "general intro"
\noindent
This appendix implements the circuit SIA of Section~\ref{sub:ppsz} reversibly for formulas of bounded-index width (biw; see Definition~\ref{def:biw}) in polynomial time and in a space-efficient way (as a function of $w$ and $n$).
We assume the input is a formula $F$ defined on $n$ variables with biw $w$.

For space efficiency, as discussed in Section~\ref{sub:ppsz}, we will not manipulate the formula $F$, but instead work on partial assignments in the branching representation,
which we call the \concept{advice}
 %for designating search tree nodes (we can however create the circuit from $F$). 
(see Section~\ref{sec:gvb}).
The downside of this approach is of course that for a given node label,
our circuits need to reconstruct the corresponding partial assignment.
To avoid using $n \log(3)$ qubits for reconstructing the partial assignment,
we exploit the bounded index width to split the computation into blocks of $w$ variables.
We then use Bennett's pebbling strategy~\cite{bennett} to compose these blocks into a space-efficient reversible circuit which computes SIA.

The following three subsections treat the reconstruction of the partial assignments (Sec.~\ref{sub:reconstruct-assignments}), split the computation into blocks of $w$ variables (Sec~\ref{sub:splitting-sia}), and combine these sub-computations in a space-efficient way (Sec~\ref{sub:combining-siab}).
At the end of this we obtain Theorem~\ref{t:main}. 

\begin{theorem}\label{t:main}
There is a reversible circuit that, given an advice of size $S_{\mathit{adv}}$,
computes $s$-SIA for a $w$-biw formula $F$ with $n$ variables in space
\[
S_w \defn O(w \cdot \log(\nicefrac nw))
\]
and time   
\[
%O(3^{\log(\nicefrac nw)} \cdot w \cdot  w^s \cdot \lolo(n)) =
O((\nicefrac nw)^{\log(3)} \cdot  w \cdot (2 + sw)^{ks} \cdot \polylog(n))
\] 
\end{theorem}

Note that the space complexity $S_w$ excludes the space required to store the advice ($S_{\mathit{adv}}$). The space requirements $S_w$ and $S_{\mathit{adv}}$ are also treated separately in Section~\ref{sub:ppsz}.
As a final remark we also note that while in practice the size of $s$ for $s$-implication can be considered a large constant, it should in fact grow (slowly) in $n$~\cite{og-ppsz} (the inverse Ackermann function would suffice).
% Because of this we cannot treat $s$ as a constant.
%Since the SIA routine has to compute the implication for $s$ clauses, we need to set $w$, so the first variable of the next block can be computed from the partial assignment recorded in the previous

\subsection{Reconstructing partial assignments}
\label{sub:reconstruct-assignments}
The simplifications in the search algorithm discussed in Section~\ref{sub:ppsz} (see Algorithm~\ref{algorithm:dnc-ppsz2}) shift the difficulty to 
creating a space-efficient unitary implementing the
search predicate $P$. It should take as input an advice string
$\vec a$ of length $S_{\mathit{adv}} = \sizeof{\vec a}$,
representing the node in the search tree and determine whether this
node is a solution leaf ($P(\vec a) = 1$), a normal leaf ($P(\vec a) = 0$)
 or a branch (such that we can query children $ch2(\vec a, 0)$ and $ch2(\vec a, 1)$). 
For conciseness, we will assume here that the advice length $S_{\mathit{adv}}$ is fixed
and no counter is required for recording the length of $\vec a$.
%, omitting the counter for the length of the advice.)
% of length $\gamma_k + \epsilon$ (the maximum branching depth in dncPPSZ).
%This could for example be a node $111$, i.e., the first three guessed variables have all been set to $1$ (see Fig.~\ref{fig:tree} for and example).

Ignoring space requirements for a moment,
%in order to determine the node's status as lead or parent, we first need to reconstruct the partial assignment corresponding to this advice.
Algorithm~\ref{alg:spec} provides a specification of $P(\vec a)$. We call this procedure  $s$-implication with advice (SIA).
It reconstructs the partial assignment on variables $x_1, \ldots, x_n$ by testing $s$-implication on each variable $x_i$ in order. If $x_i$ is forced ($s$-implied) the variable is assigned accordingly. If  $x_i$ should be guessed,
then the next value from the advice is used to assign it, using the advice pointer $p$.
%, restricting the formula further for the next iteration (we do not need to compute this restriction in the circuit, but instead can build the circuit using the formula).
If at one point the formula becomes trivial under the partial assignment,
i.e., $F_{|\vec x} = b$, for $b \in \{0,1\}$, then the specification returns $b$
(simplifying from the $\ket{1b}$ discussed in Section~\ref{sub:ppsz}).
If the advice runs out ($p = \sizeof{\vec a}$) and the formula is not yet trivial after eagerly forcing variables then the return value is $\bot$ ($\ket{00}$).

\begin{algorithm}[t]
\caption{SIA specification for reconstructing the full partial assignment
from the advice in branching representation (indices~start at~1).}
\label{alg:spec}
\begin{lstlisting}
SIA$_F$ ($\vec a$)
   $p$ := 1           (* \cmmnt advice pointer  *)
   for $i$ in 1 .. $n$
      if $F_{x_1, .., x_{i-1}} = b$
         return $b$     (* \cmmnt 0 children*)
      if $F_{x_1, .., x_{i-1}} \models_s x_i \lor F_{x_1, .., x_{i-1}} \models_s \bar x_i$
         $x_i$ := $F_{x_1, .., x_{i-1}} \models_s x_i$
      else if $p = \sizeof{\vec a}$(*\cmmnt out of guesses  *)
         return $\bot$  (* \cmmnt 2 children  *)
      else
         $x_i$ := $\vec a[p]$
         $p$ := $p$ + 1    
\end{lstlisting}
\end{algorithm}

% SIAB (high level)
% - splitting the variables into blocks of width w
% - define SIAB_i
\subsection{Splitting SIA into blocks}
\label{sub:splitting-sia}
In order to create a space-efficient, reversible implementation of the entire SIA operation, we split SIA into a number of blocks, each of which acts on a subset of variables.
The variables $x_1, \dots, x_n$ of $F$ are split up into $l = n/w$ blocks of size $w$ (for simplicity we assume $w$ divides $n$). Each block $B_i$, for $i \in \{1, \dots, l\}$, contains variables 
\begin{equation}
B_i \defn \{ x_{j+1}, \dots, x_{j+w-1} \}, \text{ with $j = (i-1)w$}. \label{eq:Bi}
\end{equation}
Note that having a bounded index width of $w$ guarantees that computing the $s$-implication of $x_i$ can only rely on $x_{i-w}, \dots, x_{i-1}$, and so in order to compute SIA for variables in $B_i$ we only need the assignments on the previous $w$-sized block $B_{i-1}$. This is made more explicit in Appendix \ref{sec:s-imp-biw}.

%We split the SIA computation into $l$ blocks of size $w = n/l$ (for simplicity we assume $w$ divides $n$), which we will call $\SIAB_1, \dots, \SIAB_l$. Each of these blocks is a unitary

We define a subcircuit $\SIAB_i$ (see Figure~\ref{fig:siab-circ}) which computes SIA (Algorithm~\ref{alg:spec}) for variables in $B_i$, unless input $r$ (a flag) is set, in which case it is the identity function. Moreover, if $\SIAB_i$ encounters a contradiction or runs out of advice, it sets $r$. 
The first three registers contain the inputs: an assignment $\vec y_{i-1}$ on the variables in $B_{i-1}$, the advice pointer $p$ (of $\log(n)$ qubits) which keeps track of the next unused advice bit, and the flag $r$ (1 qubit). The outputs are written to the next three registers and contain an assignment $\vec y_i$ to the variables in $B_i$, and the updated values of $p$ and $r$. Additionally, each $\SIAB_i$ block uses ancilla qubits $\vec a_b$. These ancilla's are initialized in an all-zero state, and are reset to zero by means of uncomputation within the $\SIAB_i$ block. As such, this ancilla register can be reused between all the blocks. Finally, $\SIAB_i$ needs read-only access to the advice $\vec a$, which is also shared between all $\SIAB_i$ blocks. Before $\SIAB_1$ the registers for $p$, $r$, and $\vec a_b$ are initialized to all zeros.
%A detailed implementation of $\SIAB_i$ is given in Appendix~\ref{app:siab}.

\begin{figure}[h]
\centering
\mbox{
 \Qcircuit @C=.7em @R=.7em {
  \lstick{\ket{\vec y_{i-1}}\ket{p}\ket{r}} & /^{w'} \qw & \multigate{3}{\text{SIAB$_i$}} & \qw & \rstick{\ket{\vec y_{i-1}}\ket{p}\ket{r}} \\
  \lstick{\ket{0}\ket{0}\ket{0}} &/^{w'} \qw & \ghost{\text{SIAB$_i$}} & \qw & \rstick{\ket{\vec y_i}\ket{p'}\ket{r'}} \\
  \lstick{\ket{\vec a_b}} & /^{S_a} \qw & \ghost{\text{SIAB$_i$}} & \qw & \rstick{\ket{\vec a_b}} \\
  \lstick{\ket{\vec a}} &/^a \qw & \ghost{\text{SIAB$_i$}} & \qw & \rstick{\ket{\vec a}}
 }}
\caption{A $\SIAB_i$ block, for $i \geq 2$. $\SIAB_1$ is a simpler version of $\SIAB_i$ without the top register. Here $w' = w + \log(n) + 1$, $S_a = w + O(\log(n))$, and $a = S_{\mathit{adv}}$.}
\label{fig:siab-circ}
\vspace{-0.2em}
\end{figure}

A more detailed implementation of $\SIAB_i$ is given in Appendix~\ref{app:siab}.

% How these SIAB blocks can be combined into a reversible circuit which does SIA
% - naively, space = O(~w * w/n) ==> seems to be the same as Remark E.4 (l = w/n; the number of blocks, and S = w; the number of input wires of SIAB
% - Bennett space = O(~w * log(w/n))
\begin{figure*}[t]
\centering
\mbox{
 \Qcircuit @C=.7em @R=.7em {
  \lstick{B_1 : \ket{0}} & /^{w'} \qw & \gate{\text{SIAB$_1$}} & \multigate{1}{\text{SIAB$_2$}} & \qw & \qw & \qw & & \cdots \\
  \lstick{B_2 : \ket{0}} & /^{w'} \qw & \qw & \ghost{\text{SIAB$_2$}} & \multigate{1}{\text{SIAB$_3$}} & \qw & \qw & & \cdots \\
  \lstick{B_3 : \ket{0}} & /^{w'} \qw & \qw & \qw & \ghost{\text{SIAB$_3$}} & \multigate{1}{\text{SIAB$_4$}} & \qw & & \cdots \\
  \lstick{B_4 : \ket{0}} & /^{w'} \qw & \qw & \qw & \qw & \ghost{\text{SIAB$_4$}} & \qw & & \cdots \\
  \vdots \\ \\
  \lstick{B_l : \ket{0}} & /^{w'} \qw & \qw & \qw & \qw & \qw & \qw & & \cdots \\
 }}
\caption{A (partial) visualization of the naive composition of all $\SIAB_i$ blocks. The bottom register of each $\SIAB_i$ block (see Fig.~\ref{fig:siab-circ}) has been omitted from this figure. This method requires $l \cdot w' + S_a$ qubits}
\label{fig:siar-naive}
\end{figure*}

\begin{figure*}[t]
\centering
\mbox{
 \Qcircuit @C=.7em @R=.7em {
  % wire 1:
  \lstick{\ket{0}} & /^{w'} \qw & \gate{\text{SIAB$_1$\vphantom{\textsuperscript{-1}}}} & \qw & \qw & \multigate{1}{\text{SIAB$_2$}} & \gate{\text{SIAB$_1^{-1}$}} & \qw & \ket{0} \\
  % wire 2:
  & & & \lstick{\ket{0}} & /^{w'} \qw & \ghost{\text{SIAB$_2$}} & \qw & \qw & \qw & \qw & \multigate{1}{\text{SIAB$_3$}} & \qw & \qw & \qw & \qw & \multigate{1}{\text{SIAB$_3^{-1}$}} & \qw & & \cdots \\
  % wire 3:
  & & & & & & & & \lstick{\ket{0}} & /^{w'} \qw & \ghost{\text{SIAB$_3$}} & \qw & \qw & \multigate{1}{\text{SIAB$_4$}} & \qw & \ghost{\text{SIAB$_3^{-1}$}} & \qw & \ket{0} \\
  % wire 4:
  & & & & & & & & & & & \lstick{\ket{0}} & /^{w'} \qw & \ghost{\text{SIAB$_4$}} & \qw & \qw & \qw & & \cdots \\
 }}
\caption{A (partial) visualization of $\SIAR$; the composition of the $\SIAB_i$ blocks using Bennett's algorithm \cite{bennett}. The bottom register of each $\SIAB_i$ block (see Fig.~\ref{fig:siab-circ}) has been omitted from this figure. Uncomputing $\SIAB_1$ frees up a ${w'}$ sized register which can be used by $\SIAB_3$. This method requires ${w'} \cdot (\log(l) + 1) + S_a$ qubits.}
\label{fig:siar-bennett}
\end{figure*}

\subsection{Combining SIAB blocks}
\label{sub:combining-siab}

Having defined $\SIAB_i$, the composition $\SIAB_l \circ ... \circ \SIAB_1$ now computes the value we are interested in, still disregarding the low-space requirement. Note that in order to ``compose'' two blocks $\SIAB_{i+1} \circ \SIAB_i$, the second output register (Fig.~\ref{fig:siab-circ}) of $\SIAB_i$ needs to be connected to the first input register of $\SIAB_i$.
A circuit which composes the $\SIAB_i$ blocks in this manner is shown in Figure~\ref{fig:siar-naive}.
In order to keep to computation unitary, the top output register of each $\SIAB_i$ block cannot simply be discarded, even though these registers are not used as inputs anymore. 

\begin{remark}
Naively composing all $\SIAB_i$ by simply allocating a new $w$-sized register for each block (see Fig.~\ref{fig:siar-naive}) requires $l \cdot w' + S_a$ qubits on top of the $S_{adv}$ qubits storing the advice.
\end{remark}

To meet the space requirement of $o(n)$, we must reduce the $l \cdot w = n$ qubits used
by combining block. To achieve this, we turn to Bennett's strategy \cite{bennett}. It relies on  reducing the problem of turning a deterministic computation into a space-efficient reversible computations to a \emph{reversible pebble game}. Such a pebble game is played on a rooted, directed graph, which corresponds to the dependency graph of inputs and outputs of different sub-computations. The rules of a reversible pebble game can be defined as follows: one can (un)pebbled a leaf freely, whereas any other node can only be (un)pebbled if their predecessors are pebbled; one wins the game if one can put a pebble on the root node. After the game is won, the maximum number of pebbles used at any time during the game corresponds to the amount of space the reversible circuit requires.

For our problem (composing $\SIAB_l \circ ... \circ \SIAB_1$) the dependency graph is simply a line graph of length $l$. Bennett's algorithm constructs a concrete pebbling strategy for the line graph corresponding to a deterministic computation.
The strategy shows that a line graph of length $l$ can be reversibly pebbled with only $\floor{\log(l)} + 1$ pebbles. A clear visualization of this algorithm is shown in \cite[Table 2]{bennett}. Figure~$\ref{fig:siar-bennett}$ shows the $\SIAB_i$ blocks composed using Bennett's algorithm, and illustrates how this saves on the space requirements: For example, by temporarily uncomputing $\SIAB_1$, a $w'$ sized register is freed up which can be used by $\SIAB_3$. We will refer to this specific reversible composition of the $\SIAB_i$ blocks as \SIAR.
% Later down the line, when the algorithm wants to uncompute $\SIAB_2$, $\SIAB_1$ needs to be recomputed first.

\begin{lemma}[Rephrased from Bennett \cite{bennett}:]\label{lem:bennett}
The composition of (irreversible) operations $f_t \circ \dots \circ f_1$, each taking $T$ time and $S$ space, can be implemented reversibly in $O(t^{\log(3)} \cdot T)$ time and $(\floor{\log(t)} + 1) \cdot S$ space.
\end{lemma}

\begin{theorem}
There is a reversible circuit, $\SIAR$, that computes SIA for a $w$-biw formula $F$ using only 
\[
\tsiar \defn O(l^{\log(3)} \cdot w \cdot (2 + sw)^{ks} \cdot \polylog(n))
\] 
time and
\[
\ssiar \defn (\floor{\log(l)} + 1) \cdot w' + S_a  = O(w \cdot \log(l))
\]
space (wires), where $w' = (w + \log(n) + 1)$ and $S_a = w + O(\log(n))$.
\end{theorem} 
\begin{proof}
By Theorem \ref{theorem:siab}, $\SIAB$ can be reversibly implemented in
$\tsiab \defn O(w \cdot (2 + sw)^{ks} \cdot \polylog(n))$ time. The time for $\SIAR$ follows then from Lemma~\ref{lem:bennett}. In terms of space, all $\SIAB$ blocks share a $S_a = O(w + \log(n))$ sized register, on top of which Bennet's strategy requires $(\floor{\log(l) + 1}) \cdot w'$ space to reversibly compose all $\SIAB_i$ (Lem.~\ref{lem:bennett}).
\end{proof}

Theorem~\ref{t:main} simplifies this statement.

\section{Implementation of SIAB}
\label{app:siab}

\newcommand{\rev}[1]{\text{rev}(#1)}

In this note, we implement the circuit $\texttt{SIAB}_i$ defined in Appendix~\ref{app:sia}, which realizes a $w$-sized block of computation of s-SIA, by computing the values of the variables in $B_i$ (Eq.~\ref{eq:Bi}) given access to the variables in $B_{i-1}$. First, the dependency of $\texttt{SIAB}_i$ on only the previous $w$ assigned variables is proven in Section~\ref{sec:s-imp-biw}. After that, the specification of the circuit $\texttt{SIAB}_i$ (given in Sections~\ref{sub:impl_siab} and \ref{sub:siab-routines}) ultimately leads to a proof of the space and time complexity (see Sec.~\ref{sub:siab_complexity}) stated in following theorem.

%\begin{theorem}
%The circuit $\texttt{SIAB}_i$ can be implemented reversibly in
%$O(w + \log(n))$ space (including $O(\log(n))$ for ancillas $a_b$), and in time
%$O(w \cdot (sw)^s \cdot \log(n))$  for $w$-biw $k$-SAT formulas.
%\end{theorem}

\begin{theorem}\label{theorem:siab}
Each circuit $\SIAB_i$ can be implemented reversibly in  $\tsiab \defn O(w \cdot (2 + sw)^{ks} \cdot \polylog(n) )$ time, with $\ssiab \defn  2(w + \log(n) + 1)$ wires and
$S_a \defn w + O(\log(n))$ ancillas, omitting the space for the advice $\vec a$.
%where reversibly means that a circuit $\SIAB_i^{-1}$ exists which 
\end{theorem}

\subsection{S-implication under bounded index width}
\label{sec:s-imp-biw}
In this subsection we prove that in order to do $s$-implication in the dncPPSZ algorithm for a variable $x_j$, only two sets of variables need to be considered: the previous $w$ assigned variables, and any of the unassigned variables $x_k$, with $j \leq k \leq n$.

We recall the definition of $s$-implications.

\begin{definition}
A literal $x_j$ $(\overline{x_j})$ is $s$-implied by $F$, written $F \models_s x_j$ $(F \models_s \overline{x_j})$, if and only if there is a subset of clauses $G\subseteq F$ of size at most $s$ such that all satisfying assignments of $G$ set $x_j$ to $1$ $(0)$.
\end{definition}

The following lemma specifies that it is not necessary to remember a partial assignment in its entirety in order to determine $s$-implications.% We simplify the presentation of this result by focusing on positive $s$-implications $F \models_s x$, but the same work can be done for negative $s$-implication.

\begin{lemma}
\label{lem:construct-s-subformula}
Consider an $n$-variables formula $F$ of bounded index width $w$, with a fixed variables order $x_1<\ldots<x_n$. 
Let $F_{|\alpha}$ be the formula obtained by assigning $\alpha = [\alpha_1, \dots, \alpha_{j-1}]$ to variables $[x_1, \dots, x_{j-1}]$, and assume
that $F_{|\alpha}$ is not yet trivial.
Then, in order to determine whether $F_{|\alpha} \models_s x_j$, it is sufficient to consider the assignment $\beta = [\alpha_{j-w}, \dots, \alpha_{j-1} ]$ to $[x_{j-w}, \dots, x_{j-1}]$.
% values assigned by $\alpha$ 
%to the variables $[x_{j-w}, \ldots, x_{j-1}]$.
\end{lemma}
\begin{proof}
Let us define $F_1$, $F_2$, and $F_3$ as subformulas of $F$, such that each only contains clauses with the variables as visualized in Figure~\ref{fig:vars-viz}.
By the property of index width, we have that the formula $F_{|\beta}$ consists of two independent formulae
$F_1$ and  $F_3$.
Hence, we have $F_3 = F_{2|\beta} = F_{|\alpha}$.
It immediately follows then that $F_{|\alpha} \models_s x_j \iff F_{2|\beta} \models_s x_j$.
%We work with simplified formulas $F'$, that is, we successively replace each variable by its value in the partial assignment $\alpha$. 
%Therefore, when the variable $x_j$ is the current variable, there is no clause left whose variables have an index below $j$. In this setting, the non-trivial subformulas of $F'$ are necessarily such that they have at least one variable with an index greater or equal to $j$. 
%
%Moreover, under the assumption that $F$ is of bounded index width $w$, if the variable $x_j$ appears in a clause $C$, then the indices of the variables of $C$ are between $j-w$ and $j+w$, with the restriction that $\iw(C)=w$. 
%
%Therefore, if a variable $x_{j'}$ (with $j' > j$) appears in the same clause of $F$ as the variable $x_j$, then the variable $x_{j'}$ can only appear in clauses of $F$ whose variables have index at least $j'-w > j-w$. In other words, all the clauses needed to determine the value of $x_j$ are such that they do not contain a variable with an index below $j-w$. It is therefore sufficient to know the formula $F$ and the variables $x_{j-w},\ldots,x_{j-1}$ in order to determine $F \models x_j$, for any $x_j$.
\end{proof}

{
\vspace{-1.5em}
\begin{figure}[H]
    \centering
    \includegraphics[width=0.8\columnwidth]{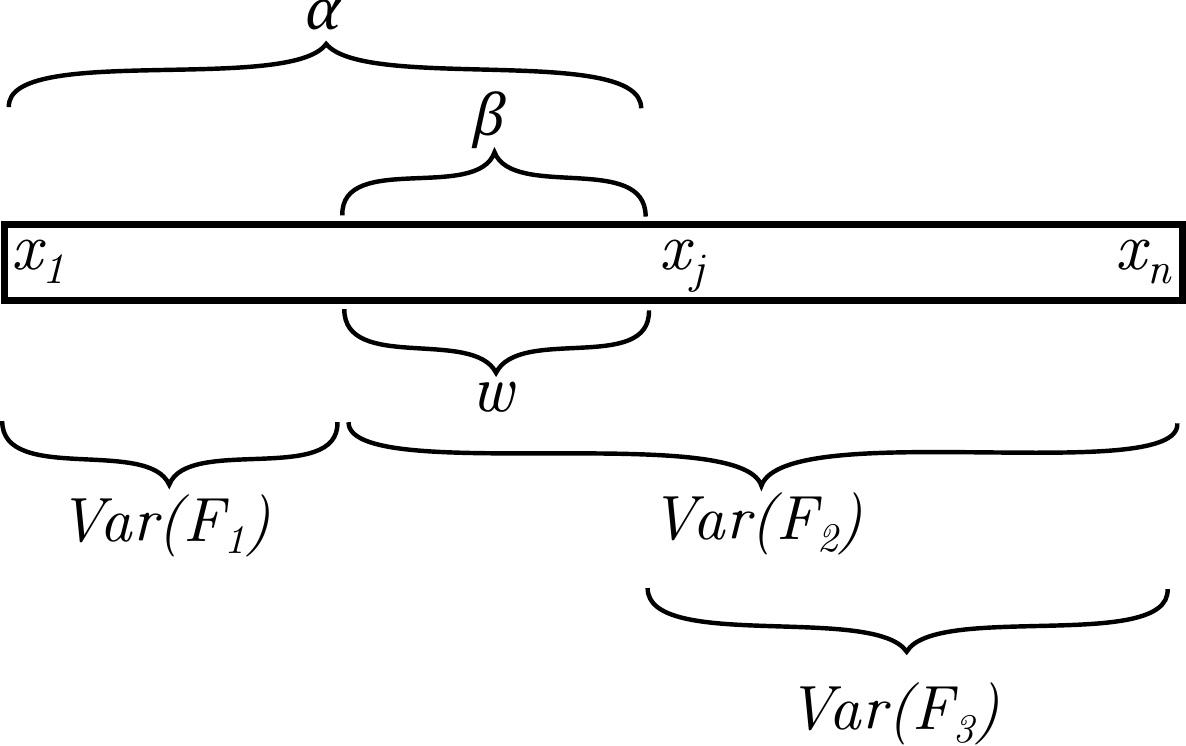}
    \vspace{-0.5em}
    \caption{Visualization of subsets of variables of $F$.}
    \label{fig:vars-viz}
    \vspace{-0.5em}
\end{figure}
}

%In order to see how Lemma~\ref{lem:construct-s-subformula} affects the subformulae $G$, let us make a distinction between variables which have already been assigned, and variables which are still free:
%\begin{align*}
%    V^j_{\text{fixed}} &= \{ x_1, \dots, x_{j-1} \}, \\
%    V^j_{\text{fixed-}w} &= \{ x_{j-w}, \dots, x_{j-1} \}, \\
%    V^j_{\text{free}} &= \{ x_j, \dots, x_n \}.
%\end{align*}
Generally, in order to check if $F_{|\alpha} \models_s x_j$, all $s$-sized $G \subseteq F$ need to be considered, where each $G$ can take variables from $\{x_1, \dots, x_n\}$. However, since $F$ has bounded index width $w$, from Lemma~\ref{lem:construct-s-subformula} we have that $F_{2|\beta} = F_{|\alpha}$ and so checking all $G \subseteq F_2$ is sufficient.

% , it is sufficient to only consider sub formulae $G$ which take variables from $V^j_{\text{fixed-}w} \cup V^j_{\text{free}}$.

%The remainder of this appendix is dedicated to the implementation of the circuit $\SIAB_i$, which is given as inputs: a pointer $p$ to the next unused advice index (defined as a counter to $a = |\vec a|$ using $\ceil{\log(a)}$ bits), a 1 bit flag $r$ which is raised (by means of a bit flip) whenever we encounter a contradiction or the advice is fully used, and a vector $\vec y_{i-1}$ of size $w$ which contains the value of the variables in block $B_{i-1}$. The circuit $\SIAB_i$ outputs the values assigned to the current block $B_i$, the current value of the pointer $p$ and the flag $r$. In what follows, we detail the implementation of the circuit  $\SIAB_i$ and analyze its space and time complexity, leading to a proof of the following theorem.

\begin{figure*}
\centering
\mbox{
 \Qcircuit @C=1em @R=.7em {
 & & & & & & & & & & \mbox{$\mathcal{G}'$} & & & \\
  % wire 1
  \lstick{\ket{\vec y}_1} & /^{w} \qw & 
  \multigate{9}{\mathcal{H}_G} & \qw & & &
  \multigate{4}{\mathcal{G}}
  \gategroup{1}{6}{9}{12}{.7em}{--}
  & \qw & \qw & 
  \multigate{4}{\mathcal{G}^\dagger} & \qw & \qw & 
  \multigate{6}{\mathcal{G}'} & \qw & \cdots & & \multigate{6}{\mathcal{G}'} & \qw & 
  \qw & \qw & \qw & \qw \\
  % wire 2
  \lstick{\ket{\vec z}_{a0}} & /^{v} \qw & 
  \ghost{\mathcal{H}_G} & \qw & & & 
  \ghost{\mathcal{G}}  
  & \qw & \qw & 
  \ghost{\mathcal{G}^\dagger} &
  \gate{inc} & \qw & 
  \ghost{\mathcal{G}'} & \qw & \cdots & & \ghost{\mathcal{G}'} & \qw & 
  \qw & \qw & \qw & \qw \\
  % wire 3
  \lstick{\ket{j}_{a1}} & /^{\log n} \qw  &  
  \ghost{\mathcal{H}_G} & \qw & & &
  \ghost{\mathcal{G}}  
  & \qw & \qw &
  \ghost{\mathcal{G}^\dagger} & \qw
  & \qw & 
  \ghost{\mathcal{G}'} & \qw & \cdots & & \ghost{\mathcal{G}'} & \qw & 
  \qw & \qw & \qw & \qw \\
  % wire 4
  \lstick{\ket{0}_{a2}} & \qw  &  
  \ghost{\mathcal{H}_G} & \qw & & &
  \ghost{\mathcal{G}} 
  & \ctrl{1} & \ctrl{1} & 
  \ghost{\mathcal{G}^\dagger} & \qw
  & \qw & 
  \ghost{\mathcal{G}'} & \qw & \cdots & & \ghost{\mathcal{G}'} & \qw & 
  \qw & \qw & \qw & \qw \\
  % wire 5
  \lstick{\ket{0}_{a3}} & \qw  &  
  \ghost{\mathcal{H}_G} & \qw & \equiv & &
  \ghost{\mathcal{G}}  
  & \ctrl{1} & \ctrlo{2} & 
  \ghost{\mathcal{G}^\dagger} & \qw
  & \qw & 
  \ghost{\mathcal{G}'} & \qw & \cdots & & \ghost{\mathcal{G}'} & \qw 
  & \qw & \qw & \qw & \qw \\
  % wire 6
  \lstick{\ket{c_+}_{a4}} &  /^{v} \qw &  
  \ghost{\mathcal{H}_G} & \qw & &
  & \qw & \gate{inc} & \qw & \qw & \qw  
  & \qw & 
  \ghost{\mathcal{G}'} & \qw & \cdots & & \ghost{\mathcal{G}'} & \qw & 
  \ctrlo{1} & \ctrl{1} & \ctrlo{1} & \qw \\
  % wire 7
  \lstick{\ket{c_-}_{a5}} &  /^{v} \qw &  
  \ghost{\mathcal{H}_G} & \qw & & &
  \qw & \qw & \gate{inc} & \qw & \qw  & \qw &
  \ghost{\mathcal{G}'} & \qw & \cdots & & \ghost{\mathcal{G}'} & \qw  & 
  \ctrl{2} & \ctrlo{1} & \ctrlo{3} & \qw \\
  % wire 8
  \lstick{\ket{0}_{a6}} & \qw &  
  \ghost{\mathcal{H}_G} & \qw & & &
  \qw & \qw & \qw & \qw & \qw & \qw & \qw & 
  \qw & \qw & \qw & \qw & \qw & 
  \targ & \targ & \qw & \qw \\
  % wire 9
  \lstick{\ket{0}_{a7}} & \qw &  
  \ghost{\mathcal{H}_G} & \qw & & &
  \qw & \qw & \qw & \qw & \qw & \qw & \qw & 
  \qw & \qw & \qw & \qw & \qw  & 
  \targ & \qw & \qw & \qw \\
  % wire 10
  \lstick{\ket{0}_{a8}} & \qw &
  \ghost{\mathcal{H}_G} & \qw & & &
  \qw & \qw & \qw & \qw & \qw & \qw & \qw & 
  \qw & \qw & \qw & \qw & \qw  & 
  \qw & \qw & \targ & \qw 
 }
 }
\caption{\label{fig:imp-circ} Routine $\mathcal{H}_G$, used to check whether $x_j$ is $s$-implied by a given subset of clauses $G$. Each $\mathcal{G}'$ checks a single assignment $\vec y$ to the variables in $G$. The $\mathcal{G}'$ block is repeated $|V'|$ times. The registers for $\vec z$, $c_+$, and $c_-$ need $v = \ceil{\log(|V'|)} \leq ks$ qubits. The controls on $\ket{c_+}_{a4}$ and $\ket{c_-}_{a5}$ check whether $c_{\pm} = 0$ (open dot) or $c_{\pm} \neq 0$ (solid dot).}
\end{figure*}

\begin{figure*}
\centering
\mbox{
 \Qcircuit @C=1em @R=.7em {
 & & & & & & & & & \mbox{~~~~~~~~~$\mathcal{H}$}\\
  % wire x_1
  \lstick{\ket{\vec y}_1} & {/} \qw & 
  \multigate{4}{\mathcal{H}_{G_1}} & \qw & \qw & \cdots & & \qw & 
  \multigate{4}{\mathcal{H}'_{G_L}} & \qw & \qw & \qw & \qw &
  \multigate{5}{\mathcal{H}^\dagger} & \qw & \qw \\
  % wire j_a1
  \lstick{\ket{j}_{a1}} & {/} \qw & 
  \ghost{\mathcal{H}_{G_1}} & \qw & \qw & \cdots & & \qw & 
  \ghost{\mathcal{H}'_{G_L}} & \qw & \qw & \qw & \qw &
  \ghost{\mathcal{H}^\dagger} & \gate{inc} & \qw \\
  % wire a6
  \lstick{\ket{0}_{a6}} & \qw & 
  \ghost{\mathcal{H}_{G_1}} & \ctrl{3} & \qw & \cdots & & \qw & 
  \ghost{\mathcal{H}'_{G_L}} & \ctrl{3} & \qw & \ctrl{1} & \ctrlo{4} &
  \ghost{\mathcal{H}^\dagger} & \qw & \qw \\
  % wire a7
  \lstick{\ket{0}_{a7}} & \qw & 
  \ghost{\mathcal{H}_{G_1}} & \qw & \qw & \cdots & & \qw & 
  \ghost{\mathcal{H}_{{G_L}}} & \qw & \qw & \ctrl{4} & \qw &
  \ghost{\mathcal{H}^\dagger} & \qw & \qw \\
  % wire a8
  \lstick{\ket{0}_{a8}} & \qw & 
  \ghost{\mathcal{H}_{G_1}} & \qw & \qw & \cdots & & \qw & 
  \ghost{\mathcal{H}_{{G_L}}} & \qw & \ctrl{2} & \qw & \qw &
  \ghost{\mathcal{H}^\dagger} & \qw & \qw \\
  % wire a9
  \lstick{\ket{0}_{a9}} & \qw & 
  \ctrlo{-1} & \targ & \qw & \cdots & & \qw & 
  \ctrlo{-1} & \targ \gategroup{1}{2}{7}{10}{.7em}{--} & \qw & \qw & \qw & 
  \ghost{\mathcal{H}^\dagger} & \qw & \qw \\
  % wire a10
  \lstick{\ket{0}_{a10}} & \qw & \qw & \qw & \qw & \qw & \qw & \qw &\qw & \qw & \targ & \qw &
  \multigate{2}{\mathcal{A}} & \qw & \qw & \qw  \\
  % wire a11
  \lstick{\ket{0}_{a11_j}} & \qw & \qw & \qw & \qw & \qw & \qw & \qw &\qw & \qw &\qw & \targ &
  \ghost{\mathcal{A}}& \qw & \qw & \qw  \\
  % wire p
  \lstick{\ket{p}_2} & \qw & \qw & \qw & \qw & \qw & \qw & \qw &\qw &\qw &\qw & \qw & 
  \ghost{\mathcal{A}} & \qw & \qw & \qw  \\
 }
 }
  \caption{Inner loop $\mathcal{L}_j$ of $\texttt{SIAB}_i$, checking all $s$-sized subsets of clauses and fixing the value of the current variable. For ease of visualization, wires in $\mathcal{H}_G$ (Fig.~\ref{fig:imp-circ}) which are not directly relevant in $\mathcal{L}_j$ have been omitted.}
\label{fig:siab-inner}
\end{figure*}

\begin{figure*}
\centering
\mbox{
 \Qcircuit @C=1em @R=.7em {
 & & & & & & & & & \mbox{~~~~~~~~$\mathcal{L}$}\\
  % wire x_1
  \lstick{\ket{\vec y}_1} & /^{\ } \qw & \multigate{3}{\mathcal{L}_1} & \qw & \qw & \cdots & & \qw & \multigate{5}{\mathcal{L}_w} & \qw & \qw & \qw & \qw & \qw & \multigate{6}{\mathcal{L}^\dagger} & \qw\\
  % wire p
  \lstick{\ket{p}_2} & /^{\ } \qw & \ghost{\mathcal L_1} & \qw & \qw & \cdots & & \qw & \ghost{\mathcal L_w} & \qw & \qw & \qw & \ctrl{9} & \qw & \ghost{\mathcal{L}^\dagger} & \qw\\
  % wire a10
  \lstick{\ket{0}_{a10}} & \qw & \ghost{\mathcal L_1} & \ctrl{4} & \qw & \cdots & & \qw & \ghost{\mathcal L_w} & \ctrl{4} & \qw  & \qw & \qw & \qw & \ghost{\mathcal{L}^\dagger} & \qw\\
  % wire a11 (1)
  \lstick{\ket{0}_{a11_1}} & \qw & \ghost{\mathcal L_1} & \qw & \qw & \cdots & & \qw & \ghost{\mathcal L_w} & \qw & \ctrl{4}  & \qw & \qw & \qw & \ghost{\mathcal{L}^\dagger} & \qw\\
  % vdots
  & \vdots &  \\
  % wire a11 (w)
  \lstick{\ket{0}_{a11_w}} & \qw & \qw & \qw & \qw & \cdots & & \qw & \ghost{\mathcal L_w} & \qw & \qw  & \ctrl{4} & \qw & \qw & \ghost{\mathcal{L}^\dagger} & \qw\\
  % wire r
  \lstick{\ket{r}_{3}} & \qw & \ctrlo{-3} & \targ & \qw & \cdots & & \qw & \ctrlo{-1} & \targ \gategroup{1}{2}{8}{10}{.7em}{--} & \qw &\qw & \qw & \ctrl{5} & \ghost{\mathcal{L}^\dagger} & \qw \\
  % wire o1 (1)
  \lstick{\ket{0}_{o1}} & \qw & \qw & \qw & \qw & \qw & \qw &\qw &\qw &\qw &\targ &\qw &\qw &\qw &\qw & \qw \\
  & \vdots & \\
  % wire o1 (w)
  \lstick{\ket{0}_{o1}} & \qw & \qw & \qw & \qw & \qw & \qw &\qw &\qw &\qw &\qw &\targ &\qw &\qw & \qw & \qw\\
  % wire o2
  \lstick{\ket{0}_{o2}} & /^{\ } \qw & \qw & \qw & \qw & \qw & \qw &\qw &\qw &\qw &\qw &\qw &\targ &\qw &\qw & \qw \\
  % wire o3
  \lstick{\ket{0}_{o3}} & \qw & \qw & \qw & \qw & \qw & \qw &\qw &\qw &\qw &\qw & \qw &\qw &\targ &\qw & \qw\\
 }
}
\caption{\label{fig:siab} Circuit for $\texttt{SIAB}_i$, running the inner loop for every variable in the current $w$-block, provided that the remaining formula has not been found to be unsatisfiable (unsatisfiability would be indicated by $r > 0$). For ease of visualization, wires in $\mathcal{L}_j$ (Fig.~\ref{fig:siab-inner}) which are not directly relevant this high-level overview have been omitted.}
\end{figure*}

\subsection{Implementing $\texttt{SIAB}_i$}
\label{sub:impl_siab}
%Observe that to check whether a variable $x_j\in B_i$ is $s$-implied, it is sufficient to check whether there is one valid $s$-sized subformula for which all satisfying assignments agree on $x_k$. The values of the variables of the previous block $B_{i-1}$ are fixed and given as input to $\texttt{SIAB}_i$. 
%Recall the specification of $\texttt{SIAB}_i$:
Algorithm~\ref{fig:siab-pseudo} gives a high level presentation of the implementation of  $\texttt{SIAB}_i$, with references to the corresponding reversible subroutines
(see Section~\ref{sub:siab-routines}). We omit $\texttt{SIAB}_1$ 
%additionally sets the counters $\vec p$ and $\vec r$ to $0$.
which is a simpler version of $\texttt{SIAB}_i$ without input block.
We write $\gets$ for the action of copying the value of a bit with a XOR gate.
% The uncomputations of ancillas are implemented by reversing all the computations done up until the point where the values of ancillas are back to their initial value ($\ket 0$).

\begin{algorithm}[H]
\caption{\label{fig:siab-pseudo}Pseudocode for $\texttt{SIAB}_i$($\vec{y}_{i-1}$, $p$, $r$).}
\begin{lstlisting}
for $j$ in $i w \ldots iw +w -1$ (*\cmmnt Fig.~\ref{fig:siab}*)
  if $(r = 0)$
    $c \gets 0$
    for each $s$-sized $G \subseteq F_{2|\vec{y}_{i-1}}$ (*\cmmnt Fig.~\ref{fig:siab-inner}*)
      if $(c = 0)$
        $(b',b'_j) \gets$check if $ G \models x_j, \overline{x_j}$ (*\cmmnt Fig.~\ref{fig:imp-circ}*)
        if $(b' = 1)$ (*\cmmnt $x_j$ is implied*)
            $c \gets c \oplus 1$
        if $G$ is unsat
            $c \gets c \oplus 1$
            $r \gets r \oplus 1$
    if $(b' = 1)$ (*\cmmnt $x_j$ is implied *)
        $x_j \gets b'_j$
    else (*\cmmnt $\mathcal{A}$ in Fig.~\ref{fig:siab-inner}*)
        if $p < a$
            $x_j \gets \vec a[p]$
            $p \gets p + 1$
        else
            $r \gets r \oplus 1$
$o1, o2, o3 \gets [x_{i\cdot w}, \ldots, x_{(i+1)\cdot w}], p, r$
uncompute workspace (*\cmmnt $\mathcal{L}^\dag$ in Fig~\ref{fig:siab} *)
\end{lstlisting}
\end{algorithm}
%$\vec {p'} \gets \vec p$
%$\vec r_{i+1} \gets \vec r$

Each circuit $\texttt{SIAB}_i$ has the following input registers:
register 1 contains an assignment $\vec y_{i-1}$ to previous $w$-sized block $B_{i-1}$,
register $2$ contains the pointer $p$ to the next unused advice index (defined as a counter from 0 to $S_{adv} - 1 = |\vec a| - 1$, using $\ceil{\log(S_{adv})}$ bits),
register $3$ contains a one bit flag $r$ which is raised whenever we encounter a contradiction or the advice is fully used.
The registers $o1,o2,o3$ are used to output the next values $\vec y_i$, $p'$, and $r'$ respectively.
We also allocate ancilla registers $a0, \ldots, a11$.
\begin{align*}
  \SIAB_i \colon 
    &\ket{\vec y_{i-1}}_1
    \ket{p}_2
    \ket{r}_3
    \ket{0}_{o1}
    \ket{0}_{o2}
    \ket{0}_{o3}
    \ket{\vec a}
    \ket{\vec a_{b}} 
\mapsto \\
    &\ket{\vec y_{i-1}}_1
    \ket{p}_2
    \ket{r}_3
    \ket{y_{i}}_{o1}
    \ket{p'}_{o2}
    \ket{r'}_{o3}
    \ket{\vec a}
    \ket{\vec a_{b}} 
\end{align*}

The ancilla registers $a0,\ldots,a8$ are the ones primarily manipulated by the subroutines $\mathcal H$ and $\mathcal H_G$, as defined in Section~\ref{sub:siab-routines}. Ancilla registers $a9$ and $a10$ are used to store $c'$ and $r'$ (see Alg.~\ref{fig:siab-pseudo}), such that $c$ and $r$ can be modified within the loop-body, without immediately affecting the control qubits $c'$ and $r'$. Finally, ancilla register $a11$ is a buffer of size $w$ in which one stores the values of the current $w$-block $B_i$.

The circuit in Figure~\ref{fig:siab-inner} corresponds to the the inner loop of $\texttt{SIAB}_i$ and the subsequent instructions which fix the value of $x_j$. The unitary $\mathcal A$ corresponds to the instructions in the pseudocode which set $x_j$ to the next advice bit and increase the advice pointer. Figure~\ref{fig:siab} gives the implementation of the whole circuit $\texttt{SIAB}_i$ (note that running $\mathcal L_j$ is conditional on $r'=0$, and the circuit flips $r$ when it finds that the remainder of the formula is unsatisfiable). To simplify our representation, we omit ancilla wires which are only manipulated in the subcomponents of the circuit, but not the overall circuit.

\subsection{Reversible subroutines of $\texttt{SIAB}_i$}
\label{sub:siab-routines}
For each $s$-sized subformula $G \subseteq F$, we define a unitary
\begin{align*}
    \mathcal{G} : &\ket{\vec y}_1 \ket{\vec z}_{a0} \ket{j}_{a1} \ket{0}_{a2} \ket{0}_{a3} \mapsto  \\
    &\ket{\vec y}_1 \ket{\vec z}_{a0} \ket{j}_{a1} \ket{b}_{a2} \ket{b_j}_{a3}
\end{align*}
%\begin{figure}[h]
%\centering
%\mbox{
% \Qcircuit @C=.7em @R=.7em {
%  \lstick{\ket{\vec x}} & \qw & /^w \qw & \qw & \multigate{3}{\mathcal{G}} & \qw & \rstick{\ket{\vec x}} \\
%  \lstick{\ket{\vec y}} & \qw &/^{ks} \qw & \qw & \ghost{\mathcal{G}} & \qw & \rstick{\ket{\vec y}} \\
%  \lstick{\ket{j}} & \qw &/^{\log n} \qw & \qw & \ghost{\mathcal{G}} & \qw & \rstick{\ket{j}} \\
%  \lstick{\ket{0}\ket{0}} & \qw & /^2 \qw & \qw & \ghost{\mathcal{G}} & \qw & \rstick{\ket{b \oplus 0}\ket{b_j \oplus 0}}
% }}
%\end{figure}

\noindent
which takes the following inputs: an index $j$, an assignment $\vec y$ to the $w$ variables preceding $x_j$, an  assignment $\vec z$ of length (at most) $ks$ to variables of indices greater or equal to $j$, and two ancilla qubits to write the output to.

%Only the register $\ket{\vec y_{i-1}}$ is a quantum register. 
%The result bits $(b,b_j)$ are written to two separate qubits, which will contain the values
%\begin{align*}
%    (b, b_j) = 
%    \begin{cases}
%        (0,0) &\text{ if $x_j$ does not appear in $G$}, \\
%        (1,b_j) &\text{ if G is satisfied by the partial assignment}
%        123
%    \end{cases}
%\end{align*}

Each unitary $\mathcal G$ hardcodes the Boolean formula $G$, and encodes the following reversible procedure: if variable $x_j$ does not appear in $G$, output $(0,0)$; else if $G$ is satisfied by the partial assignment specified by $\vec y$ on $x_{j-w}, \dots, x_{j-1}$ and $\vec z$ on 
\[ V' = \text{Var}(G) \setminus \{x_{j-w}, \ldots, x_{j-1}\}, \] 
output $(1,b_j)$, where $b_j$ is the assignment to variable $x_j$,
and else output $(0,0)$.

Next, we define a unitary
\begin{align*}
\mathcal{G}' :&\ket{\vec y}_1\ket{\vec z}_{a0}\ket{j}_{a1}\ket{c_+}_{a4}\ket{c_-}_{a5} \mapsto\\
 &\ket{\vec y}_1\ket{\vec z \oplus 1}_{a0}\ket{j}_{a1}\ket{c'_+}_{a4}\ket{c'_-}_{a5}
\end{align*}
which applies $\mathcal G$, increase the counter $c_+$ ($c_-$) if $\mathcal G$ outputs $(1,1)$ ($(1,0$)), and then uncomputes the ancilla registers. Additionally $\mathcal{G}'$ increased the value of the $\ket{\vec z}_{a0}$ register, such that the next $\mathcal{G}'$ is applied to the next partial assignment. This is shown in Figure~\ref{fig:imp-circ}.
The registers $\ket{c_+}_{a4}$ and $\ket{c_-}_{a5}$ (as well as $\ket{\vec z}_{a0}$) require $\lceil\log(|V'|)\rceil \leq ks$ qubits to count up to $2^{|V'|}$, the total number of possible assignments to $G$.

For each subformula $G \subseteq F$, we implement an unitary 
\begin{align*}
    \mathcal{H}_G :&\ket{\vec  x}_1\ket{j}_{a1}\ket{0}_{a6}\ket{0}_{a7}\ket{0}_{a8} \mapsto \\
    &\ket{\vec x}_1\ket{j}_{a1}\smallket{b'}_{a6}\smallket{b'_j}_{a7}\ket{r''}_{a8}
\end{align*}
which evaluates $G$ for all its possible partial assignments, to determine whether the satisfying assignments of $G$ agree on the value assigned to $x_j$, and if yes ($b'=1$), outputs that value ($b'_j$). If $G$ is unsatisfiable, $\mathcal{H}_G$ outputs $r''=1$.  

Our implementation of $\mathcal{H}_G$ initializes the ancilla registers $a0,a2,\ldots,a5$ to $0$, and repeatedly applies $\mathcal{G}'$ and increments $\vec y$ (see Figure~\ref{fig:imp-circ}) up until $\vec y = 2^{|V'|} - 1$.
After applying $\mathcal{G}'$ for all assignments $\vec y$, the following can be concluded based on the counters $c_+$ and $c_-$:
\begin{align*}
\begin{cases}
    \text{$x_j$ not implied, $G$ is sat} & \text{ if $c_+ > 0, c_- > 0$} \\
    \text{$G$ implies $x_j$} & \text{ if $c_+ > 0, c_- = 0$} \\
    \text{$G$ implies $\overline{x_j}$} & \text{ if $c_+ = 0, c_- > 0$} \\
    \text{$G$ is unsat} & \text{ if $c_+ = 0, c_- = 0$}
\end{cases} 
\end{align*}
If both $c_+$ and $c_-$ are non-zero, $G$ has satisfying assignments, but they do not all agree on $x_j$. In this case $\mathcal{H}_G$ outputs $(b',b_j') = (0,0)$. If exactly one of $c_+$ and $c_-$ is non-zero, all satisfying assignments of $G$ agree on $x_j$, and either $x_j$ or $\overline{x_j}$ is implied by $G$. Here $\mathcal{H}_G$ outputs $(1,1)$ or $(1,0)$ respectively. If both $c_+$ and $c_-$ are zero, $G$ is not satisfiable at all given the variables assigned before $x_j$, and so neither is $F$. In this case $\mathcal{H}_G$ outputs $r''=1$, which ultimately leads the circuit to raise the flag $r$.

\subsection{Complexity of the implementation of $\texttt{SIAB}_i$}
\label{sub:siab_complexity}
First, let us analyze the time complexity of $\texttt{SIAB}_i$. To simplify the analysis, we drop any $\polylog(n)$ factor by writing $\tilde{O}(f(n)) = O(\polylog(n) \cdot f(n))$. By doing so the increment operations, as well as the $\mathcal{A}$ operation can be done in $\tilde{O}(1)$ time.

The unitary $\mathcal{G}$ hardcodes a Boolean formula $G$ which has (at most) $s$ clauses and $k$ variables per clause, taking $T_\mathcal{G} = O(ks) = \tilde{O}(1)$ time. The subroutine $\mathcal{G}'$ runs $\mathcal{G}$ and $\mathcal{G}^\dagger$, and a number of $\tilde{O}(1)$ operations, which gives $T_{\mathcal{G}'} = \tilde{O}(1)$. Next, $\mathcal{H}_G$ repeats $\mathcal{G}'$ for all assignments to $\vec z$, which are at most $2^{ks}$, giving $T_{\mathcal{H}_G} = \tilde{O}(2^{ks})$.

The inner loop $\mathcal{L}_i$ loops over all $s$-sized formulae $G$.
Since we consider formulas with bounded index width $w$, for any variable $x_j$, any subformula $G\subseteq F$ with $s$ clauses can either be split into independent formulae or contains only unassigned variables from $\{x_j, \dots, x_{j+sw} \}$ (a chain of $s$ clauses each with index width $w$).
Each $s$-sized $G$ has at most $ks$ variables, and which $sw$ options per variable this results in at most $O((sw)^{ks})$ formulae $G$. This yields a time complexity of $T_{\mathcal{L}_j} = \tilde{O}((sw)^{ks} \cdot 2^{ks}) = \tilde{O}((2 + sw)^{ks})$.

Finally, the outer loop of $\SIAB_i$ loops over $w$ variables. Giving $\tsiab = \tilde{O}(w \cdot (2 + sw)^{ks}) = O(w \cdot (2 + sw)^{ks} \cdot \polylog(n) )$.

%\begin{itemize}
%    \item $T(inc_n) = n^2$ (?)
%    \item time $n$-control gate: $T(ctrl_n) = n^2$  (?)
%    \item say both operations above on $\log(n)$ qubits take $O(\polylog(n))$ time?
%    \item $T(\mathcal{A}) = $
%    \item number of functions $G$: $(sw)^{ks}$
%    \item $T(\mathcal{G}) = O(s \cdot (k + T(ctrl_k)))$
%    \item $T(\mathcal{H}_G) = O( 2^{ks} \cdot (T(\mathcal{G}) + T(inc_v)))$
%    \item $T(\mathcal{L}_j) = O((sw)^{ks} \cdot T(\mathcal{H}_G) + T(\mathcal{A}) + T(inc_v))$
%    \item $T(\SIAB_i) = O(w \cdot T(\mathcal{L}_j))$
%\end{itemize}

Including the input registers $1, 2, 3$ and output registers $o1, o2, o3$ the space requirement for is $\SIAB_i$ is $2(w + \log(n) + 1) + S_{a}$, where $S_a$ is the space for the ancilla registers $a0$ through $a11$, which adds up to $S_{a} = w + \log(n) + 3ks + O(1) = w + O(\log(n))$.

%Lattice SAT NP-completeness proof
\section{Lattice SAT is \NP-complete}
\label{app:latticesat}

We define Lattice SAT to be the restriction of the $3$-SAT problem to formulas defined on a lattice, so that each clause is associated to a  constraint defined on a tile (unit square of the grid), with corners of the tiles labelled by variables. 

Formally, an instance of Lattice SAT is a formula defined on $n$ variables, and whose clauses are defined by $2$ to $3$ corners on the same tile (see Figure~\ref{fig:latticesat}), with at least one tile defined on $3$ corners, in such a way that the overall lattice fits within a square of length $\sqrt{n}$. 

By construction, there are permutations of indices which are such that any given instance of Lattice SAT has bounded index width $\sqrt{n} \in o(n)$.

First, observe that Lattice SAT is in \NP because the validity of an assignment for the underlying 3-SAT instance can be verified in polynomial time. Let us show that Lattice SAT is \NP-hard.

A polynomial-time reduction from $3$-SAT to Lattice SAT can be implemented as follows. Consider a $3$-SAT formula $F$ defined on $n$ variables over $L$ clauses. Consider a lattice $\mathbb{F}$ defined as a $nL$-by-$nL$ $2$-dimensional grid. In what follows, for each clause $C$, we place variables (which appear in $C$) in the lattice $\mathbb F$, and we add to the lattice $\mathbb F$ a set of tiles which is equisatisfiable to the clause $C$.

We associate $L$ copies $x_{i,1},\ldots,x_{i,L}$ (one per clause) to each variable $x_i$.
We introduce new constraints to ensure that all copies have the same truth value, that is
\[
x_{i,l} \vee \bar{x}_{i,l+1}
\text{ and }
\bar{x}_{i,l} \vee {x}_{i,l+1}.
\]
Having $L$ copies of each variable $x_i$ allows for a spatial arrangement on the lattice of any two copies  $x_{i,l}$ and $x_{i,l+1}$ in a constrained relationship, by defining one tile in the lattice where they meet.
We place such copies on the diagonal of the lattice $\mathbb{F}$, from the top left corner to the bottom right one in the order determined by the order of variables in $\text{Vars}(F)$.

Consider a clause $C_l = x_{i_1} \vee x_{i_2} \vee x_{i_3}$ in $F$ (where we assume without loss of generality that $i_1$, $i_2$ and $i_3$ are variables indices such that $i_1 \leq i_2 \leq i_3$). Let us define a set of tiles which corresponds to a restricted formula which is equisatisfiable to $C_l$. 
Observe that each clause $C_l$ can be decomposed into three clauses $C'_l = x_{i_1,l} \vee x_{i_2,l} \vee t$, $C''_l = x_{i_2,l} \vee x_{i_3,l} \vee \bar{t'}$ and $C'''_l = \bar{t} \vee t'$ (for some fresh variables $t,t'$), so that $C_l \equiv C'_l \wedge C''_l \wedge C'''_l$.

\begin{comment}
\begin{figure}
    \centering
\begin{tikzpicture}
  \node[draw, circle] (xi11) at (1,5) {\tiny $x_{i_1,1}$};
  \node[draw, circle] (xi1l) at (2,4) {\tiny $x_{i_1,l}$};
  \node[draw, circle] (xi2l) at (5,1) {\tiny $x_{i_2,l}$};
  \node[draw, circle] (xi2L) at (6,0) {\tiny $x_{i_2,L}$};
  \node[draw, circle] (t) at (5,4) {\tiny $t$};
  \node[draw,circle] (y2) at (3,4) {\tiny $y_2$};
   \node[draw,circle] (z2) at (5,2) {\tiny $z_2$};
   \node[draw,circle] (yp) at (4,4) {\tiny $y_p$};%{\tiny $y_{p-1}$};
  \node[draw,circle] (zq) at (5,3) {\tiny $z_q$};%{\tiny $z_{q-1}$};

  \draw[dotted, thick] (xi11) -- (xi1l);
  \draw[dotted, thick] (xi1l) -- (xi2l);
  \draw[dotted, thick] (xi2l) -- (xi2L); 
  
  \draw (xi1l) -- (y2);
  \draw[dotted, thick] (y2) -- (yp);
  \draw (xi2l) -- (z2);
  \draw[dotted, thick] (z2) -- (zq);
  \draw (yp) -- (t);
  \draw (t) -- (zq);
\end{tikzpicture}
    \caption{Spatial representation of the construction of $C'_l$ on the lattice $\mathbb{F}$, with plain black lines representing the lines of $\mathbb{F}$ between directly adjacent corners}
    \label{fig:latticesat-red}
\end{figure}
\end{comment}

\begin{figure}
%    \centering
\hspace{-.5cm}
\begin{tikzpicture}[minimum width=1.8cm]
  \node (ghost1i) at (-1,1) {};
  \node (ghost2i) at (1,-.25) {};
  \node[draw, circle] (xi) at (1,1) {$x_i$};
  \node (ghost1in) at (7,1) {};
  \node (ghost2in) at (5,-0.25) {};
  \node[draw, circle] (xin) at (5,1) {$x_{i+\sqrt{n}}$};
  \node (ghost1iu) at (-1,5) {};
  \node (ghost2iu) at (1,6.5) {};
  \node[draw, circle] (xiu) at (1,5) {$x_{i+1}$};
  \node (ghost1iun) at (7,5) {};
  \node (ghost2iun) at (5,6.5) {};
  \node[draw, circle] (xiun) at (5,5) {$x_{i+1+\sqrt{n}}$};
  
  \node[minimum width=.2cm,draw] (clause) at (3,3) {$C$};
  
  \draw[thick] (ghost1i) -- (xi);
  \draw[thick] (ghost2i) -- (xi);
  \draw[thick] (xi) -- (xiu);
  \draw[thick] (xi) -- (xin);
  \draw[thick] (ghost1in) -- (xin);
  \draw[thick] (ghost2in) -- (xin);
  \draw[thick] (xiu) -- (xiun);
  \draw[thick] (xin) -- (xiun);
  \draw[thick] (ghost1iu) -- (xiu);
  \draw[thick] (ghost2iu) -- (xiu);
  \draw[thick] (ghost1iun) -- (xiun);
  \draw[thick] (ghost2iun) -- (xiun);
  
  \draw[red, dashed, very thick] (clause) -- (xi);
  \draw[red, dashed, very thick] (clause) -- (xin);
  \draw[red, dashed, very thick] (clause) -- (xiu);
      
\end{tikzpicture}
    \caption{$C = x_i \vee \neg x_{i+\sqrt{n}} \vee x_{i+1}$ in Lattice SAT (black lines) and in Planar 3-SAT (red dotted lines)}
    \label{fig:latticesat}
\end{figure}

{In what follows, we construct a set of constraints which is equisatisfiable to the clause $C'_l$.}
First, fresh variables $y_1,\ldots,y_p,z_1,\ldots,z_q,t$ are spatially arranged on the lattice $\mathbb F$ so that: $y_1,\ldots,y_p$ are placed on the same horizontal line as $x_{i_1,l}$ in $\mathbb F$, with $y_1$ directly at the right of $x_{i_1,l}$, and each $y_{i+1}$ is placed directly at the right of $y_i$; $z_1,\ldots,z_q$ are placed on the same vertical line as $x_{i_2,l}$, with $z_1$ directly above $x_{i_2,l}$, and each $z_{i+1}$ is placed directly above $z_i$; $t$ is placed on the same tile as $y_p$ and $z_q$ and $t$, so that $t$ is directly at the right of $y_p$ and directly above $z_q$. Then, we add the following set of the constraints to the tiles of lattice $\mathbb{F}$ on which those variables $y_1,\ldots,y_p,z_1,\ldots,z_q,t$ are located:
\begin{itemize}
    \item $\bar{y}_i \vee {y}_{i+1}$, $\bar{y}_{i+1} \vee {y}_{i}$ (for $0 \leq i \leq p$),
    \item $\bar{z}_j \vee {z}_{j+1}$, $\bar{z}_{j+1} \vee {z}_{j}$ (for $0 \leq j \leq q$),
   \item $y_p \vee t \vee y_1$,
\end{itemize}
where $y_0$ is $x_{i_1,l}$ and $z_0$ is $x_{i_2,l}$,
ensuring that the variables $y_1,\ldots,y_p,x_{i_1,l}$ take the same truth value, and the variables $z_1,\ldots,z_q,x_{i_2,l}$ take the same truth value.

We repeat the same process for $x_{i_2}$ and $x_{i_3}$, adding fresh variables $y'_1,\ldots,y'_p,z'_1,\ldots,z'_q,t'$, with the same constraints but this time $\bar{t'}$ in the constraint where $t$ previously appeared. We obtain a second set of constraints which is equisatisfiable to $C''_l$.

Now, observe that reiterating the same process a third time for $t$ and $t'$, we obtain a third set of constraints which is equisatisfiable to $C'''_l$.
Combining the three sets of constraints, we obtain a set of constraints which is equisatisfiable to the formula $C'_l \wedge C''_l \wedge C'''_l$, which is itself equisatisfiable to the clause $C_l$.

\usetikzlibrary{decorations.shapes}
\tikzset{decorate sep/.style 2 args=
{decorate,decoration={shape backgrounds,shape=circle,shape size=#1,shape sep=#2}}}

\begin{figure}
    \centering
\begin{tikzpicture}
\draw[decorate sep={2mm}{4mm},fill] (0.5,0) -- (1.5,0);
\draw[decorate sep={2mm}{4mm},fill] (0.5,0.5) -- (1.5,0.5);
\draw[decorate sep={2mm}{4mm},fill] (0.5,1) -- (1.5,1);
\draw[red,line width=0.5mm] (0.2,0.5) --(1.5,0.5);
\draw[blue,line width=0.5mm] (0.9,-0.2) -- (0.9,1.2);

\draw[->,thick] (2,0.5) -- (3,0.5);
\draw[decorate sep={2mm}{4mm},fill] (4,0) -- (5,0);
\draw[decorate sep={2mm}{4mm},fill] (4,0.5) -- (5,0.5);
\draw[decorate sep={2mm}{4mm},fill] (4,1) -- (5,1);
\draw[red,line width=0.5mm] (4.87,0.7) -- (4.25,-0.3);
\draw[red,line width=0.5mm] (4.55,-0.3) -- (3.85,0.7);
\draw[blue,line width=0.5mm] (4.4,0.3) -- (4.4,1.2);
\draw[blue,line width=0.5mm] (5,-0.3) -- (4.3,0.7);
%\draw[red,thick,rotate=-45] (4,0) ellipse (0.4 and 0.2);
%\draw[red,thick] (4.6,0.5) ellipse (0.4 and 0.2);
%\draw[red,thick] (5,0.5) ellipse (0.4 and 0.2);

\end{tikzpicture}
    \caption{Rewriting overlaps}
    \label{fig:rewriting-latticesat}
\end{figure}

We repeat this process for every clause $C_l$, obtaining an instance $\mathbb{F}$ of Lattice SAT defined by $O((nL)^2)$ constraints on a $nL$-by-$nL$ squared grid.
The instance $\mathbb{F}$ is not equisatisfiable to $F$, as our reduction potentially generate overlaps, which can be eliminated by repeatedly inserting empty rows and lines and applying the rewriting gadget described in a simplified representation in Figure~\ref{fig:rewriting-latticesat}, with coloured lines corresponding to tiles defined on two variables (black dots). At most $(nL)^2$ overlaps exist, and $L \in O(\text{poly}(n))$, so that this reduction can be done in polynomial time.

\begin{theorem}
Lattice SAT is \NP-complete.
\end{theorem}

% Walk operator
%\input{appendix-walk-operator}

\end{document}